\documentclass[authoryear, 1p]{elsarticle} 
\usepackage{etex}
\usepackage{listings}
\usepackage[utf8]{inputenc} 
\usepackage{fixltx2e} 
\usepackage{mathtools}
\usepackage{wrapfig}
\usepackage{complexity}
\usepackage{listings}
\usepackage{balance}
\usepackage{mathptmx} 
\usepackage{mathdots}
\RequirePackage{latexsym}
\RequirePackage{mathtools}
\RequirePackage{amsfonts}
\usepackage{amsthm}
\usepackage{amssymb}
\RequirePackage{array}
\RequirePackage{xspace}
\newcommand{\val}[2]{\texttt{val}^{#1}(#2)}
\newcommand{\suchthat}{\;\ifnum\currentgrouptype=16 \middle\fi|\;}

\newcommand{\ie}{i.e., }
\newcommand{\eg}{e.g., }
\newcommand{\etal}{\textit{et al.}\xspace}

\newcommand{\wrt}{w.r.t.\ }

\usepackage[textsize=small]{todonotes}

\def\qed{{$\hfill \Box$\\}}

\newcommand{\Q}{\mathbb{Q}\xspace}
\newcommand{\N}{\mathbb{N}\xspace}
\newcommand{\Z}{\mathbb{Z}\xspace}

\newcommand{\MPG}{MPG\xspace}
\newcommand{\EG}{EG\xspace}
\newcommand{\MCP}{MCP\xspace}

\def\C{{\cal C}}
\def\J{J}

\def\F{{\cal F}}
\def\W{{\cal W}}

\newcommand{\figref}[1]{Fig.~\ref{#1}}

\newtheorem{Thm}{Theorem}
\newtheorem{Exp}{Example}
\newtheorem{Lem}{Lemma}
\newtheorem{Prop}{Proposition}

\newtheorem{Def}{Definition}
\newtheorem{Rem}{Remark}

\usepackage[ruled,vlined,linesnumbered]{algorithm2e}
\newcounter{proccnt}
\newenvironment{algo-proc}[1][htb]
  {\refstepcounter{proccnt}\renewcommand*{\algorithmcfname}{SubProcedure} 
   \begin{algorithm}[#1]\renewcommand\thealgocf{\arabic{proccnt}}
}{\end{algorithm}\addtocounter{algocf}{-1}}
\usepackage{subfig}
\SetCommentSty{textsf}
\SetKwRepeat{DoWhile}{do}{while}
\SetAlFnt{\small} 

\makeatletter
\newcommand{\removelatexerror}{\let\@latex@error\@gobble}
\makeatother

\let\oldnl\nl
\newcommand{\nonl}{\renewcommand{\nl}{\let\nl\oldnl}}
\usepackage{tikz}
\usepackage{tikz-qtree}
\usepackage{pgfplots}
\usetikzlibrary{pgfplots.units}
\usetikzlibrary{calc,positioning,fit}
\usetikzlibrary{shapes,shapes.multipart,shapes.arrows}
\usetikzlibrary{decorations,decorations.pathmorphing,decorations.pathreplacing,decorations.markings,decorations.shapes}
\usetikzlibrary{arrows}
\usetikzlibrary{fit,backgrounds}
\usetikzlibrary{plotmarks}
\tikzstyle{node}=[circle,draw,inner sep=2pt,transform shape,minimum size=1.75em]
 
\newcommand*\sizedcircled[2]{\tikz[baseline=(char.base)]{ \node[shape=circle,draw,inner sep=2pt, scale=#1] (char) {\textbf{#2}}; }} 

\tikzstyle{every picture}=[>=latex]
\tikzstyle{every label}=[inner sep=2pt]

\usetikzlibrary{decorations.pathmorphing}

\journal{arXiv}

\begin{document}

\begin{frontmatter}

\title{Faster $O(|V|^2|E|W)$-Time Energy Algorithms \\ for Optimal Strategy Synthesis in Mean Payoff Games}

\author[tn,fr]{Carlo Comin}
\ead{carlo.comin@unitn.it}

\author[vr]{Romeo Rizzi}
\ead{romeo.rizzi@univr.it}

\address[tn]{Department of Mathematics, University of Trento, Trento, Italy}
\address[fr]{LIGM, Universit{\'e} Paris-Est, Marne-la-Vall{\'e}e, Paris, France}
\address[vr]{Department of Computer Science, University of Verona, Verona, Italy}

\date{Received: date / Accepted: date}

\begin{abstract}
This study strengthens the links between Mean Payoff Games (\MPG{s}) and Energy Games (EG{s}). 
Firstly, we offer a faster $O(|V|^2|E|W)$ pseudo-polynomial time and $\Theta(|V|+|E|)$ space 
deterministic algorithm for solving the Value Problem and Optimal Strategy Synthesis in \MPG{s}. 
This improves the best previously known estimates on the pseudo-polynomial time complexity to: 
\[ O(|E|\log |V|) + \Theta\Big(\sum_{v\in V}\texttt{deg}_{\Gamma}(v)\cdot\ell_{\Gamma}(v)\Big) = O(|V|^2|E|W), \] 
where $\ell_{\Gamma}(v)$ counts the number of times that a certain energy-lifting operator 
$\delta(\cdot, v)$ is applied to any $v\in V$, along a certain sequence of Value-Iterations on reweighted \EG{s}; 
and $\texttt{deg}_{\Gamma}(v)$ is the degree of $v$. 
This improves significantly over a previously known pseudo-polynomial time estimate, 
\ie $\Theta\big(|V|^2|E|W + \sum_{v\in V}\texttt{deg}_{\Gamma}(v)\cdot\ell_{\Gamma}(v)\big)$~\citep{CR15, CR16}, 
as the pseudo-polynomiality is now confined to depend solely on $\ell_\Gamma$. Secondly, 
we further explore on the relationship between Optimal Positional Strategies (OPSs) in \MPG{s} 
and Small Energy-Progress Measures (SEPMs) in reweighted \EG{s}. It is observed that the space of all OPSs, 
$\texttt{opt}_{\Gamma}\Sigma^M_0$, admits a unique complete decomposition in terms of extremal-SEPM{s} in reweighted EG{s}. 
This points out what we called the “Energy-Lattice $\mathcal{X}^*_{\Gamma}$ 
associated to $\texttt{opt}_{\Gamma}\Sigma^M_0$". Finally, it is offered a pseudo-polynomial total-time recursive procedure  
for enumerating (w/o repetitions) all the elements of $\mathcal{X}^*_{\Gamma}$, 
and for computing the corresponding partitioning of $\texttt{opt}_{\Gamma}\Sigma^M_0$.
\end{abstract}


\end{frontmatter}

\section{Introduction}\label{sect:introduction}  
A \emph{Mean Payoff Game} (\MPG) is a two-player infinite game $\Gamma=(V, E, w, \langle V_0, V_1 \rangle)$, 
which is played on a finite weighted directed graph, denoted $G^{\Gamma} \triangleq ( V, E, w )$, where $w:E\rightarrow \Z$,  
the vertices of which are partitioned into two classes, $V_0$ and $V_1$, according to the player to which they belong. 

At the beginning of the game a pebble is placed on some vertex $v_s\in V$, and then the two players, 
named Player~0 and Player~1, move it along the arcs ad infinitum. Assuming the pebble is currently on some $v\in V_0$, 
then Player~0 chooses an arc $(v,v')\in E$ going out of $v$ and moves the pebble to the destination vertex $v'$. 
Similarly, if the pebble is currently on some $v\in V_1$, it is Player~1 to choose an outgoing arc. 
The infinite sequence $v_s,v,v'\ldots$ of all the encountered vertices forms a \emph{play}. 
In order to play well, Player~0 wants to maximize the limit inferior of the long-run average weight 
of the traversed arcs, \ie to maximize $\liminf_{n\rightarrow\infty} \frac{1}{n}\sum_{i=0}^{n-1}w(v_i, v_{i+1})$, whereas 
Player~1 wants to minimize the $\limsup_{n\rightarrow\infty} \frac{1}{n} \sum_{i=0}^{n-1}w(v_i, v_{i+1})$. 
\cite{EhrenfeuchtMycielski:1979}~proved that each vertex $v$ admits a \emph{value}, denoted $\val{\Gamma}{v}$, 
which each player can secure by means of a \emph{memoryless} (or \emph{positional}) strategy, 
\ie one depending only on the current vertex position and not on the previous choices.

Solving an \MPG consists in computing the values of all vertices (\emph{Value Problem}) and, for each player, 
a positional strategy that secures such values to that player (\emph{Optimal Strategy Synthesis}).
The corresponding decision problem lies in $\NP\cap \coNP$~\citep{ZwickPaterson:1996} and it was later 
shown to be in $\UP\cap\coUP$~\citep{Jurdzinski1998}.

The problem of devising efficient algorithms for solving \MPG{s} has been studied extensively in the literature.
The first milestone was settled in \cite{GKK88}, in which it was offered an \emph{exponential} time algorithm 
for solving a slightly wider class of \MPG{s} called \emph{Cyclic Games}. 
Afterwards, \cite{ZwickPaterson:1996} devised the first deterministic procedure for computing values in \MPG{s}, 
and optimal strategies securing them, within a pseudo-polynomial time and polynomial space. 
In particular, it was established an $O(|V|^3 |E| W)$ upper bound for the time complexity of the Value Problem, 
as well as an upper bound of $O(|V|^4|E| W \log(|E|/|V|))$ for that of Optimal Strategy Synthesis~\citep{ZwickPaterson:1996}.

Several research efforts have been spent in studying quantitative extensions of infinite games for 
modeling quantitative aspects of reactive systems~\citep{Chakrabarti03,Bouyer08,brim2011faster}. 
In this context quantities may represent, \eg the power usage of an embedded component, or the buffer size of a networking element. 
These studies unveiled interesting connections with \MPG{s}, 
and recently they have led to the design of faster procedures for solving them.
In particular, \cite{brim2011faster} devised a faster deterministic algorithm for solving the Value Problem 
and Optimal Strategy Synthesis in \MPG{s} within $O\big(|V|^2 |E|\, W\log(|V|\,W)\big)$ pseudo-polynomial time and polynomial space. 
Essentially, a binary search is directed by the resolution of multiple reweighted \EG{s}. 
The determination of \EG{s} comes from repeated applications of energy-lifting operators $\delta(\cdot, v)$ for any $v\in V$; 
these are all monotone functions defined on a complete lattice (\ie the Energy-Lattice of a reweighted \EG). 
At this point the correct termination is ensured by the Knaster–Tarski fixed-point theorem.

Recently, the worst-cast time complexity of the Value Problem and Optimal Strategy Synthesis was given an 
improved pseudo-polynomial upper bound~\citep{CR15, CR16}. Those works focused on offering a simple proof 
of the improved time complexity bound. The algorithm there proposed, henceforth called Algorithm~0, had the 
advantage of being very simple; its existence made it possible to discover and analyze some of the underlying fundamental ideas, 
that ultimately led to the improved upper bound, more directly; it was shown appropriate to supersede 
(at least in the perspective of sharpened bounds) the above mentioned binary search by 
sort of a linear search that succeeds at amortizing all the energy-liftings throughout the computation. 
However, its running time turns out to be also $\Omega(|V|^2|E|W)$, the actual time complexity being 
$\Theta\big(|V|^2 |E|\, W + \sum_{v\in V}\texttt{deg}_{\Gamma}(v)\cdot\ell_{\Gamma}^0(v)\big)$, where 
$\ell_{\Gamma}^0(v)\leq (|V|-1)|V|W$ denotes the total number of times that the energy-lifting operator 
$\delta(\cdot, v)$ is applied to any $v\in V$ by Algorithm~0. 

After the appearance of those works, a way to overcome this issue was found.

\paragraph*{Contribution}
This study aims at strenghtening the relationship between \MPG{s}~and~EG{s}. 

	Our results are summarized as follows. 
\begin{itemize}
\item \emph{A Faster $O(|V|^2|E|W)$-Time Algorithm for \MPG{s} by Jumping through Reweighted \EG{s}.} 
\end{itemize}
We introduce a novel algorithmic scheme, named \emph{Jumping} (Algorithm~\ref{algo:solve_mpg}), 
which tackles on some further regularities of the problem, thus reducing the estimate on the pseudo-polynomial 
time complexity of \MPG{s} to: \[O(|E|\log |V|) + \Theta\Big(\sum_{v\in V}\texttt{deg}_{\Gamma}(v)\cdot\ell_{\Gamma}^1(v)\Big), \] 
where $\ell_{\Gamma}^1(v)$ is the total number of applications of $\delta(\cdot, v)$ to $v\in V$ that are made by Algorithm~1;
$\ell_{\Gamma}^1\leq (|V|-1)|V|W$ (worst-case; but experimentally, $\ell_{\Gamma}^1\ll\ell_{\Gamma}^0$), 
and the working space is $\Theta(|V|+|E|)$. Overall the worst-case complexity is still $O(|V|^2|E|W)$,  
but the pseudo-polynomiality is now confined to depend solely on the total number $\ell^1_{\Gamma}$ of required energy-liftings; 
this is not known to be $\Omega(|V|^2|E|W)$ generally, and future theoretical or practical 
advancements concerning the Value-Iteration framework for \EG{s} could help reducing this metric. 
Under this perspective, theoretically, the computational equivalence between \MPG{s} and \EG{s} seems 
now like a bit more unfolded and subtle. In practice, Algorithm~\ref{algo:solve_mpg} 
allows us to reduce the magnitude of $\ell_{\Gamma}$ considerably, \wrt \cite{CR15, CR16}, and 
therefore the actual running time of the algorithm; our experiments suggest that $\ell_{\Gamma}^1\ll\ell_{\Gamma}^0$ holds 
for wide families of \MPG{s} (see SubSection~\ref{subsect:experiments}). 

In summary, the present work offers a \emph{faster} pseudo-polynomial time algorithm; theoretically the 
pseudo-polynomiality now depends only on $\ell^1_{\Gamma}$, and in practice the actual performance also 
improves considerably \wrt \cite{CR15, CR16}. With hindsight, Algorithm~0 turned out to be a high-level description 
and the tip of a more technical and efficient underlying procedure. This is the first 
truly $O(|V|^2|E|W)$ time deterministic algorithm, for solving the Value Problem and Optimal Strategy Synthesis in \MPG{s}, 
that can be effectively applied in practice (optionally, in interleaving with some other known sub-exponential time algorithms). 

Indeed, a wide spectrum of different approaches have been investigated in the literature.
For instance, \cite{Andersson06fastalgorithms} provided a fast \emph{randomized} algorithm for solving \MPG{s} 
in \emph{sub-exponential} time $O\big(|V|^2 |E|\, \exp\big(2\, \sqrt{|V|\, \ln(|E| / \sqrt{|V|})}+O(\sqrt{|V|}+\ln|E|)\big)\big)$. 
\cite{LifshitsPavlov:2007} devised a deterministic $O(2^{|V|}\, |V|\, |E|\, \log W)$ \emph{singly-exponential} time 
procedure by considering a so called potential-theory of \MPG{s}, one that is akin to \EG{s}. 

Table~\ref{Table:Algorithms} offers a summary of past and present results in chronological order.

\begin{table}[!htb]
\caption{Time Complexity of the main Algorithms for solving \MPG{s}.}
\label{Table:Algorithms}
\centering
\bgroup
\def\arraystretch{1.4}
\begin{tabular}{| c  c  c  | }
\hline
\vtop{\hbox{\strut Algorithm }}&\vtop{\hbox{\strut Optimal Strategy Synthesis}} & \vtop{\hbox{\strut Value Problem}} \\
\hline
This work &  $ O(|E|\log |V|) + \Theta\big( \sum_{v\in V}\texttt{deg}_\Gamma(v) \cdot 
					\ell^1_{\Gamma}(v)\big)$  & \textit{\small same complexity} \\ 
CR15-16    &  $\Theta\big(|V|^2 |E|\, W + \sum_{v\in V}\texttt{deg}_\Gamma(v) \cdot 
					\ell^0_{\Gamma}(v)\big)$  & \textit{\small same complexity} \\
BC11 & $O(|V|^2 |E|\, W\log(|V|\, W))$ & \textit{\small same complexity} \\
LP07 &  n/a &    $O(|V| |E| 2^{|V|} \log W)$ \\
AV06 &  $O\Big(|V|^2 |E|\, e^{2\,\sqrt{|V|\, \ln\big(\frac{|E|}{\sqrt{|V|}}\big)} + 
		O(\sqrt{|V|}+\ln|E|)}\Big)$ & \textit{\small same complexity} \\
ZP96 &       $\Theta(|V|^4 |E|\, W \log\frac{|E|}{|V|})$    & $\Theta(|V|^3 |E|\, W)$  \\
\hline
\end{tabular}
\egroup
\end{table}
\begin{itemize}
	\item \emph{An Energy-Lattice Decomposition of the Space of all Optimal Positional Strategies.}
\end{itemize}
Let us denote by $\texttt{opt}_{\Gamma}\Sigma^M_0$ the space of all the optimal positional strategies in a given \MPG $\Gamma$.
What allows Algorithm~\ref{algo:solve_mpg} to compute at least one $\sigma^*_0\in 
\texttt{opt}_{\Gamma}\Sigma^M_0$ is a so called \emph{compatibility} relation, 
linking optimal arcs in \MPG{s} to arcs that are \emph{compatible} \wrt least-SEPM{s} in reweighted \EG{s}. 
The family $\mathcal{E}_\Gamma$ of all SEPM{s} of a given \EG $\Gamma$ forms a complete lattice, 
which we call the Energy-Lattice of the \EG $\Gamma$. Firstly, 
we observe that even though compatibility \wrt \emph{least}-SEPMs 
in reweighted \EG{s} implies optimality of positional strategies 
in \MPG{s} (see Theorem~\ref{Thm:pos_opt_strategy}), 
the converse doesn't hold generally (see Proposition~\ref{prop:counter_example}). 
Thus a natural question was whether compatibility \wrt SEPM{s} 
was really appropriate to capture (\eg to provide a recursive enumeration of) the 
whole $\texttt{opt}_{\Gamma}\Sigma^M_0$ and not just a proper subset of it. 
Partially motivated by this question we further explored on the relationship between $\texttt{opt}_{\Gamma}\Sigma^M_0$ 
and $\mathcal{E}_\Gamma$. In Theorem~\ref{thm:main_energystructure}, it is observed a unique complete decomposition of 
$\texttt{opt}_{\Gamma}\Sigma^M_0$ which is expressed in terms of so called \emph{extremal}-SEPM{s} in reweighted EG{s}. 
This points out what we called the “Energy-Lattice $\mathcal{X}^*_{\Gamma}$ associated to $\texttt{opt}_{\Gamma}\Sigma^M_0$", 
\ie the family of all the extremal-SEPM{s} of a given $\Gamma$. 
So, compatibility \wrt SEPM{s} actually turns out to be appropriate for constructing the whole $\texttt{opt}_{\Gamma}\Sigma^M_0$; 
but an entire lattice $\mathcal{X}^*_{\Gamma}$ of extremal-SEPM{s} then arises (and not only the least-SEPM, 
which accounts only for the ``join/top" component of $\texttt{opt}_{\Gamma}\Sigma^M_0$). 

\begin{itemize}
	\item \emph{A Recursive Enumeration of Extremal-SEPMs and of Optimal Positional Strategies.}
\end{itemize}
Finally, it is offered a pseudo-polynomial total time recursive procedure 
for enumerating (w/o repetitions) all the elements of $\mathcal{X}^*_{\Gamma}$, 
and for computing the associated partitioning of $\texttt{opt}_{\Gamma}\Sigma^M_0$. 
This shows that the above mentioned compatibility relation is appropriate so to extend our algorithms, 
recursively, in order to compute the whole $\texttt{opt}_{\Gamma}\Sigma^M_0$ and $\mathcal{X}^*_{\Gamma}$. 
It is observed that the corresponding recursion tree actually defines an additional lattice $\mathcal{B}^*_{\Gamma}$,  
whose elements are certain sub-games $\Gamma'\subseteq \Gamma$ that we call \emph{basic}. 
The extremal-SEPMs of a given $\Gamma$ coincide with the least-SEPMs of the basic sub-games of $\Gamma$;
so, $\mathcal{X}^*_{\Gamma}$ is the energy-lattice comprising all and only the \emph{least}-SEPMs of 
the \emph{basic} sub-games of $\Gamma$. The total time complexity of the proposed enumeration 
for both $\mathcal{B}^*_{\Gamma}$ and $\mathcal{X}^*_{\Gamma}$ is $O(|V|^3|E|W |\mathcal{B}^*_{\Gamma}|)$, 
it works in space $O(|V||E|)+\Theta\big(|E||\mathcal{B}^*_{\Gamma}|\big)$.

\paragraph*{Organization} 
The manuscript is organized as follows. In Section~\ref{sect:background}, we introduce some notation and provide 
the required background on infinite two-player pebble-games and related algorithmic results. 
In Section~\ref{section:values}, a suitable relation between values, optimal strategies, 
and certain reweighting operations is recalled from~\cite{CR15, CR16}. 
In Section~\ref{sect:algorithm}, it is offered an $O(|E|\log |V|) + 
	\Theta\big(\sum_{v\in V}\texttt{deg}_\Gamma(v) \cdot \ell^1_{\Gamma}(v)\big) = O(|V|^2 |E|\, W)$ 
pseudo-polynomial time and $\Theta(|V|+|E|)$ space deterministic algorithm for solving the 
Value Problem and Optimal Strategy Synthesis in \MPG{s}; SubSection~\ref{subsect:experiments} offers 
an experimental comparison between Algorithm~\ref{algo:solve_mpg} and Algorithm~0~\citep{CR15, CR16}. 
Section~\ref{sect:energy} offers a unique and complete energy-lattice decomposition of $\texttt{opt}_{\Gamma}\Sigma^M_0$. 
Finally, Section~\ref{sect:recursive_enumeration} provides a recursive enumeration 
of $\mathcal{X}^*_{\Gamma}$ and the corresponding partitioning of $\texttt{opt}_{\Gamma}\Sigma^M_0$. 
The manuscript concludes in Section~\ref{sect:conclusions}.

\section{Notation and Preliminaries}\label{sect:background}
We denote by $\N$, $\Z$, $\Q$ the set of natural, integer, and rational numbers. 
It will be sufficient to consider integral intervals, \eg $[a,b]\triangleq\{z\in\Z\mid a\leq z\leq b\}$ 
and $[a,b)\triangleq\{z\in\Z\mid a \leq z < b\}$ for any $a,b\in \Z$. 
If $(a,b),(a',b')\in\Z$, then $(a,b)<(a',b')$ \textit{iff} $a<a'$, or $a=a'$ and $b<b'$.
Our graphs are directed and weighted on the arcs; thus, if $G=(V, E, w)$ is a graph, 
then every arc $e\in E$ is a triplet $e=(u,v,w_e)$, where $w_e = w(u,v)\in\Z$. 
Let $W \triangleq \max_{e\in E} |w_e|$ be the maximum absolute weight. Given a vertex $u\in V$, the set of its successors is  
$N_\Gamma^{\text{out}}(u) \triangleq \{ v \in V \mid (u,v) \in E \}$, 
whereas the set of its predecessors is $N_\Gamma^{\text{in}}(u) \triangleq \{ v \in V \mid (v,u) \in E \}$. 
Let $\texttt{deg}_\Gamma(v)\triangleq |N_\Gamma^{\text{in}}(v)| + |N_\Gamma^{\text{out}}(v)|$.
A \emph{path} is a sequence $v_0v_1\ldots v_n\ldots$ such that $\forall^{i\in [n]}\, (v_i, v_{i+1}) \in E$. 
Let $V^*$ be the set of all (possibly empty) finite paths. 
A \emph{simple path} is a finite path $v_0v_1\ldots v_n$ having no repetitions, 
\ie for any $i,j\in [0,n]$ it holds $v_i \neq v_j$ if $i\neq j$. 
A \emph{cycle} is a path $v_0v_1\ldots v_{n-1}v_n$ such that $v_0\ldots v_{n-1}$ is simple and $v_n = v_0$.
The \emph{average weight} of a cycle $v_0\ldots v_n$ is $w(C)/|C|=\frac{1}{n} \sum_{i=0}^{n-1} w(v_i,v_{i+1})$. 
A cycle $C=v_0v_1\ldots v_n$ is \emph{reachable} from $v$ in $G$ if there is some path $p$ in $G$ such that $p\cap C\neq\emptyset$. 

An \emph{arena} is a tuple $\Gamma = (V, E, w, \langle V_0, V_1\rangle)$ 
where $G^{\Gamma} \triangleq (V, E, w)$ is a finite weighted directed graph and $(V_0, V_1)$ is a partition 
of $V$ into the set $V_0$ of vertices owned by Player~0, and $V_1$ owned by Player~$1$.
It is assumed that $G^{\Gamma}$ has no sink, \ie $\forall^{v\in V} N_{\Gamma}^{\text{out}}(v)\neq\emptyset$; 
we remark that $G^{\Gamma}$ is not required to be a bipartite graph on colour classes $V_0$ and $V_1$.
A {sub-arena} $\Gamma'$ (or \emph{sub-game}) of $\Gamma$ is any arena $\Gamma' = (V', E', w', \langle V'_0, V'_1\rangle )$ 
such that: $V'\subseteq V$, $\forall^{i\in\{0,1\}} V'_i=V'\cap V_i$, $E'\subseteq E$, and $\forall^{e\in E'} w'_e=w_e$. 
Given $S\subseteq V$, the sub-arena of $\Gamma$ induced by $S$ is denoted $\Gamma_{|_{S}}$, 
its vertex set is $S$ and its edge set is $E'=\{(u,v)\in E \mid u,v\in S\}$.
A game on $\Gamma$ is played for infinitely many rounds by two players moving a pebble along the arcs 
of $G^{\Gamma}$. At the beginning of the game the pebble is found on some vertex $v_s\in V$, 
which is called the \emph{starting position} of the game. At each turn, 
assuming the pebble is currently on a vertex $v\in V_i$ (for $i=0, 1$), 
Player~$i$ chooses an arc $(v,v')\in E$ and then the next turn starts with the pebble on $v'$. 
Below, \figref{fig:ex1_arena} depicts an example arena $\Gamma_{\text{ex}}$. 

\begin{figure}[!h]
\center
\begin{tikzpicture}[scale=.6,arrows={-triangle 45}, node distance=1.5 and 2]
 		\node[node, thick, color=red] (E) {$E$};
		\node[node, thick, color=blue, left=of E, fill=blue!20] (C) {$C$};
		\node[node, thick, color=red, above=of C, xshift=-8.5ex, yshift=-5ex] (B) {$B$};
		\node[node, thick, color=blue, left=of C, fill=blue!20] (A) {$A$};
		\node[node, thick, color=red, below=of C, xshift=-8.5ex, yshift=5ex] (D) {$D$}; 	
		\node[node, thick, color=blue, right=of E, fill=blue!20] (F) {$F$};
		\node[node, thick, color=red, right=of F] (G) {$G$};	
		\draw[] (E) to [bend left=0] node[above] {$0$} (C);
		\draw[] (E) to [bend left=22] node[above left, xshift=-4ex] {$0$} (A.south east);
		\draw[] (E) to [bend left=50] node[above left, xshift=-4ex, yshift=-1ex] {$0$} (G.north);
		\draw[] (E) to [bend left=0] node[above] {$0$} (F);
		\draw[] (A) to [bend left=40] node[left] {$+3$} (B);
		\draw[] (B) to [bend left=40] node[xshift=2ex, yshift=1ex] {$+3$} (C);
		\draw[] (C) to [bend left=40] node[xshift=2ex, yshift=0ex] {$-5$} (D);
		\draw[] (D) to [bend left=40] node[xshift=-2ex, yshift=-1ex] {$-5$} (A);
		\draw[] (F) to [bend left=40] node[xshift=-2.5ex, yshift=-.5ex, above] {$-5$} (G);
		\draw[] (G) to [bend left=40] node[below] {$+3$} (F);
\end{tikzpicture}
\caption{
An arena $\Gamma_{\text{ex}}=\langle V, \E, w, (V_0, V_1) \rangle$. Here, $V=\{A,B,C,D,E,F,G\}$ and  
$\E=\{(A,B,+3), (B,C,+3), (C,D,-5)$, $(D,A,-5), (E,A,0), (E,C,0), (E,F,0), (E,G,0), (F,G,-5), (G,F,+3)\}$.  
Also, $V_0=\{B,D,E,G\}$ is colored in red, while $V_1=\{A,C,F\}$ is filled in blue.
}\label{fig:ex1_arena}
\end{figure}
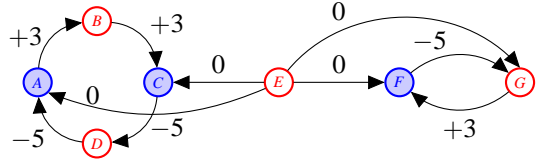

A \emph{play} is any infinite path $v_0v_1\ldots v_n\ldots\in V^\omega$ in $\Gamma$. 
For any $i\in \{0,1\}$, a strategy of Player~$i$ is any function $\sigma_i:V^*\times V_i\rightarrow V$ 
such that for every finite path $p'v$ in $G^{\Gamma}$, 
where $p'\in V^*$ and $v\in V_i$, it holds that $(v, \sigma_i(p', v))\in E$. 
A strategy $\sigma_i$ of Player $i$ is \emph{positional} (or \emph{memoryless}) 
if $\sigma_i(p, v_n) = \sigma_i(p', v'_m)$ for every finite paths $pv_n=v_0\ldots v_{n-1}v_n$ 
and $p'v'_m=v'_0\ldots v'_{m-1}v'_m$ in $G^{\Gamma}$ such that $v_n=v'_m\in V_i$. 
The set of all the positional strategies of Player~$i$ is denoted by $\Sigma^M_i$. 
A play $v_0v_1\ldots v_n\ldots $ is \emph{consistent} with a strategy 
$\sigma\in\Sigma_i$ if $v_{j+1} = \sigma(v_0v_1\ldots v_j)$ whenever $v_j\in V_i$. 

Given a starting position $v_s\in V$, the \emph{outcome} of two strategies $\sigma_0 \in\Sigma_0$ and $\sigma_1 \in\Sigma_1$, 
denoted $\texttt{outcome}^{\Gamma}(v_s, \sigma_0, \sigma_1)$, 
is the unique play that starts at $v_s$ and is consistent with both $\sigma_0$ and $\sigma_1$.

Given a memoryless strategy $\sigma_i\in\Sigma^M_i$ of Player~$i$ in $\Gamma$, 
then $G(\sigma_i, \Gamma)=(V, E_{\sigma_i}, w)$ is the graph obtained from $G^{\Gamma}$ 
by removing all the arcs $( v,v')\in E$ such that $v\in V_i$ and $v'\neq \sigma_i(v)$; 
we say that $G(\sigma_i, \Gamma)$ is obtained from $G^{\Gamma}$ \emph{by projection} \wrt $\sigma_i$. 

For any weight function $w':E\rightarrow \Z $, the \emph{reweighting} of $\Gamma=(V, E, w, \langle V_0, V_1\rangle )$ \wrt $w'$ 
is the arena $\Gamma^{w'}= (V, E, w', \langle V_0, V_1\rangle )$. Also, for $w:E\rightarrow \Z$ and any $\nu\in \Z$, 
we denote by $w+\nu$ the weight function $w'$ defined as $\forall^{e\in E} w'_e \triangleq w_e+\nu$. 
Indeed, we shall consider reweighted games of the form $\Gamma^{w-q}$, for some $q\in \Q$. 
Notice that the corresponding weight function $w':E\rightarrow\Q:e\mapsto w_e-q$ is rational, 
while we required the weights of the arcs to be always integers. 
To overcome this issue, it is sufficient to re-define $\Gamma^{w-q}$ by  scaling all weights by a factor equal 
to the denominator of $q\in \Q$; \ie when $q\in \Q$, 
say $q=N/D$ for $\gcd(N,D)=1$ we define $\Gamma^{w-q}\triangleq \Gamma^{D\cdot w-N}$. This rescaling operation doesn't change 
the winning regions of the corresponding games (we denote this equivalence as $\Gamma^{w-q}\cong \Gamma^{D\cdot w - N}$), 
and it has the significant advantage of allowing for a discussion (and an algorithmics) strictly based on integer weights. 

\subsection{Mean Payoff Games}
A \emph{Mean Payoff Game}~(\MPG)~\citep{brim2011faster, ZwickPaterson:1996, EhrenfeuchtMycielski:1979} 
is a game played on some arena $\Gamma$ for infinitely many rounds by two opponents,  
Player~$0$ gains a payoff defined as the long-run average weight of the play, 
whereas Player~$1$ loses that value. 
Formally, the Player~$0$'s \emph{payoff} of a play $v_0v_1\ldots v_n\ldots $ 
in $\Gamma$ is defined as follows: 
\[\texttt{MP}_0(v_0v_1\ldots v_n\ldots)\triangleq\liminf_{n\rightarrow\infty} 
	\frac{1}{n}\sum_{i=0}^{n-1}w(v_i, v_{i+1}).\]
The value \emph{secured} by a strategy $\sigma_0\in\Sigma_0$ in a vertex $v$ is defined as:
\[\texttt{val}^{\sigma_0}(v)\triangleq
\inf_{\sigma_1\in\Sigma_1}\texttt{MP}_0\big(\texttt{outcome}^{\Gamma}(v, \sigma_0, \sigma_1)\big),\]
Notice that payoffs and secured values can be defined symmetrically for the Player~$1$ 
(\ie by interchanging the symbol \emph{0} with \emph{1} and \emph{inf} with \emph{sup}).

Ehrenfeucht and Mycielski~\cite{EhrenfeuchtMycielski:1979} proved that each vertex 
$v\in V$ admits a unique \emph{value}, denoted $\val{\Gamma}{v}$, which each player can secure by means 
of a \emph{memoryless} (or \emph{positional}) strategy. Moreover, 
\emph{uniform} positional optimal strategies do exist for both players, 
in the sense that for each player there exist at least one 
positional strategy which can be used to secure all the optimal values, 
independently with respect to the starting position $v_s$.
Thus, for every \MPG $\Gamma$, there exists a strategy 
$\sigma_0\in\Sigma^M_0$ such that $\forall^{v\in V} \val{\sigma_0}{v}\geq \val{\Gamma}{v}$,  
and there exists a strategy $\sigma_1\in\Sigma^M_1$ such that $\forall^{v\in V} \val{\sigma_1}{v}\leq \val{\Gamma}{v}$.
The \emph{(optimal) value} of a vertex $v\in V$ in the \MPG $\Gamma$ is given by:
\[\val{\Gamma}{v} = \sup_{\sigma_0\in\Sigma_0}\val{\sigma_0}{v} = \inf_{\sigma_1\in\Sigma_1}\val{\sigma_1}{v}.\]
Thus, a strategy $\sigma_0\in\Sigma_0$ is \emph{optimal} if $\texttt{val}^{\sigma_0}(v)=\val{\Gamma}{v}$ for all $v\in V$.  
We denote $\text{opt}_{\Gamma}\Sigma^M_0\triangleq \big\{\sigma_0\in\Sigma^M_0(\Gamma) \mid \; 
	\forall^{v\in V}\, \texttt{val}^{\Gamma}_{\sigma_0}(v) = \val{\Gamma}{v}\big\}$.
A strategy $\sigma_0\in\Sigma_0$ is said to be \emph{winning} for Player~$0$ if $\forall^{v\in V}\texttt{val}^{\sigma_0}(v)\geq 0$,  
and $\sigma_1\in\Sigma_1$ is winning for Player~$1$ if $\texttt{val}^{\sigma_1}(v) < 0$.
Correspondingly, a vertex $v\in V$ is a \emph{winning starting position} for 
Player~$0$ if $\val{\Gamma}{v}\geq 0$, otherwise it is winning for Player~$1$. 
The set of all winning starting positions of Player~$i$ is denoted by $\W_i$ for $i\in \{0,1\}$.

A refined formulation of the determinacy theorem is offered in~\cite{Bjorklund04}.
\begin{Thm}[\cite{Bjorklund04}]\label{thm:ergodic_partition}
Let $\Gamma$ be an \MPG and let $\{C_i\}_{i=1}^m$ be a partition 
(called \emph{ergodic}) of its vertices into $m\geq 1$ classes each one having the same optimal value $\nu_i\in\Q$. 
Formally, $V=\bigsqcup_{i=1}^m C_i$ and $\forall^{i\in [m]}\forall^{v\in C_i} \val{\Gamma_{i}}{v}=\nu_i$, 
where $\Gamma_{i}\triangleq\Gamma_{|_{C_i}}$.

Then, Player~0 has no vertices with outgoing arcs leading from $C_i$ to $C_j$ whenever $\nu_i<\nu_j$, 
	and Player~1 has no vertices with outgoing arcs leading from $C_i$ to $C_j$ whenever $\nu_i>\nu_j$;

moreover, there exist $\sigma_0\in\Sigma^M_0$ and $\sigma_1\in\Sigma^M_1$ such that:

-- If the game starts from any vertex in $C_i$, 
	then $\sigma_0$ secures a gain at least $\nu_i$ to Player~0 and $\sigma_1$ secures a loss at most $\nu_i$ to Player~1;

-- Any play that starts from $C_i$ always stays in $C_i$, if it is consistent with both strategies $\sigma_0, \sigma_1$, 
	\ie if Player~0 plays according to $\sigma_0$, and Player~1 according to $\sigma_1$.
\end{Thm}

A finite variant of \MPG{s} is well-known in the literature 
\citep{EhrenfeuchtMycielski:1979, ZwickPaterson:1996, brim2011faster},  
where the game stops as soon as a cyclic sequence of vertices is traversed. 
It turns out that this is equivalent to the infinite game formulation~\citep{EhrenfeuchtMycielski:1979}, 
in the sense that the values of an \MPG are in a strong relationship 
with the average weights of its cycles, as in the next lemma.

\begin{Prop}[Brim,~\etal \cite{brim2011faster}]\label{prop:reachable_cycle}
Let $\Gamma$ be an \MPG. For all $\nu\in\Q$, for all $\sigma_0\in\Sigma^M_0$, 
and for all $v\in V$, the value $\texttt{val}^{\sigma_0}(v)$ is greater than $\nu$ \textit{iff} all cycles $C$ 
reachable from $v$ in the projection graph $G^{\Gamma}_{\sigma_0}$ have an average weight $w(C)/|C|$ greater than $\nu$.
\end{Prop} 
The proof of Proposition~\ref{prop:reachable_cycle} follows from the memoryless determinacy of \MPG{s}. 
We remark that a proposition which is symmetric to Proposition~\ref{prop:reachable_cycle} holds for Player~$1$ as well: 
for all $\nu\in\Q$, for all positional strategies $\sigma_1\in\Sigma^M_1$ of Player~$1$, and for all vertices $v\in V$, 
the value $\texttt{val}^{\sigma_1}(v)$ is less than $\nu$ \textit{iff} if 
all cycles reachable from $v$ in the projection graph $G^{\Gamma}_{\sigma_1}$
have an average weight less than $\nu$. Also, 
it is well-known~\citep{brim2011faster, EhrenfeuchtMycielski:1979} that each value $\val{\Gamma}{v}$ is contained within 
the following set of rational numbers: \[ S_{\Gamma}=\Big\{ N/D \suchthat D\in [1, |V|],\, N\in [-D\cdot W, D\cdot W] \Big\}.\] 
Notice, $|S_{\Gamma}|\leq |V|^2 W$. 

The present work tackles on the algorithmics of the following two classical problems:

-- \emph{Value Problem.} Compute for each vertex $v\in V$ the (rational) optimal value $\val{\Gamma}{v}$.

-- \emph{Optimal Strategy Synthesis.} Compute an optimal positional strategy for Player~0 in $\Gamma$. 

In Section~\ref{sect:recursive_enumeration} we shall consider the problem of computing the whole $\texttt{opt}_\Gamma\Sigma^M_0$.

-- \emph{Optimal Strategy Enumeration.} Provide a listing\footnote{The listing have to be exhaustive 
(\ie each element is listed, eventually) and without repetitions (\ie no element is listed twice).} 
of all the optimal positional strategies of Player~0 in the \MPG $\Gamma$.

\subsection{Energy Games and Small Energy-Progress Measures}
An \emph{Energy Game} (\EG) is a game that is played on an arena $\Gamma$ for infinitely many rounds by two opponents,
where the goal of Player~0 is to construct an infinite play $v_0v_1\ldots v_n\ldots$ 
such that for some initial \emph{credit} $c\in\N$ the following holds: 
$c + \sum_{i=0}^{j}w(v_i, v_{i+1})\geq 0\, \text{, for all } j \geq 0$. 
Given a credit $c\in\N$, a play $v_0v_1\ldots v_n\ldots$ is \emph{winning} for Player~0 if it satisfies (1), 
otherwise it is winning for Player~1. A vertex $v\in V$ is a winning 
starting position for Player~0 if there exists an initial credit $c\in\N$ 
and a strategy $\sigma_0\in\Sigma_0$ such that, for every strategy $\sigma_1\in\Sigma_1$, 
the play $\texttt{outcome}^{\Gamma}(v, \sigma_0, \sigma_1)$ is winning for Player~0. 
As in the case of \MPG{s}, the \EG{s} are memoryless determined \cite{brim2011faster}, 
\ie for every $v\in V$, either $v$ is winning for Player~$0$ or $v$ is winning for Player~$1$, 
and (uniform) memoryless strategies are sufficient to win the game.
In fact, as shown in the next lemma, the decision problems of \MPG{s} and \EG{s} are intimately related.
\begin{Prop}[\cite{brim2011faster}]\label{prop:relation_MPG_EG}
Let $\Gamma$ be an arena. For all threshold $\nu\in\Q$, for all vertices $v\in V$, 
Player~$0$ has a strategy in the \MPG $\Gamma$ that secures value at least $\nu$ from $v$ if and only if, 
for some initial credit $c\in\N$, Player~$0$ has a winning strategy from $v$ in the reweighted \EG $\Gamma^{w-\nu}$.
\end{Prop}

In this work we are especially interested in the \emph{Minimum Credit Problem} (\MCP) for \EG{s}: 
for each winning starting position $v$, compute the minimum initial credit $c^*=c^*(v)$ 
such that there exists a winning strategy $\sigma_0\in\Sigma^M_0$ for Player~$0$ starting from $v$.
A fast pseudo-polynomial time deterministic procedure for solving \MCP{s} comes from \cite{brim2011faster}.

\begin{Thm}[\cite{brim2011faster}]\label{Thm:VI}
There exists a deterministic algorithm for solving the MCP within $O(|V|\,|E|\,W)$ pseudo-polynomial time, 
on any input \EG $(V, E, w, \langle V_0, V_1\rangle)$. 
\end{Thm}
The algorithm mentioned in Theorem~\ref{Thm:VI} is 
the \emph{Value-Iteration} algorithm~\citep{brim2011faster}. 
Its rationale relies on the notion of \emph{Small Energy-Progress Measures} (SEPM{s}).

\subsection{Energy-Lattices of Small Energy-Progress Measures}
Small-Energy Progress Measures are bounded, non-negative and integer-valued 
functions that impose local conditions to ensure global properties on the
arena, in particular, witnessing that Player~0 has a way to enforce 
conservativity (\ie non-negativity of cycles) in the resulting game's graph.
Recovering standard notation, see e.g.~\cite{brim2011faster}, 
let us denote $\C_\Gamma=\{n\in\N\mid n\leq (|V|-1) W\}\cup\{\top\}$ and let $\preceq$ be the total order on 
$\C_\Gamma$ defined as: $x\preceq y$ \textit{iff} either $y=\top$ or $x,y\in\N$ and $x\leq y$.
In order to cast the minus operation to range over $\C_{\Gamma}$, 
let us consider an operator $\ominus:\C_\Gamma\times\Z\rightarrow \C_\Gamma$ defined as follows: 
\[
a\ominus b \triangleq \left\{ 
	\begin{array}{ll}
		\max(0, a-b) & ,  \text{ if } a\neq \top \text{ and } a-b\leq (|V|-1)W ; \\
		a\ominus b = \top & , \text{ otherwise.} \\
	\end{array}\right.
\]
Given an \EG $\Gamma$ on vertex set $V = V_0 \cup V_1$, a function $f:V\rightarrow\C_\Gamma$ is 
a \emph{Small Energy-Progress Measure} (SEPM) for $\Gamma$ \textit{iff} the following two conditions are met: 
\begin{enumerate}
\item if $v\in V_0$, then $f(v)\succeq f(v') \ominus w(v,v')$ for \emph{some} $(v, v')\in E$;
\item if $v\in V_1$, then $f(v)\succeq f(v') \ominus w(v,v')$ for \emph{all} $(v, v')\in E$.
\end{enumerate}

The values of a SEPM, \ie the elements of the image $f(V)$, are called the \emph{energy levels} of $f$.
It is worth to denote by $V_f=\{v\in V\mid f(v)\neq\top\}$ the set of vertices having finite energy.
Given a SEPM $f:V\rightarrow \C_\Gamma$ and a vertex $v\in V_0$, 
an arc $(v, v')\in E$ is said to be \emph{compatible} with $f$ whenever $f(v)\succeq f(v')\ominus w(v,v')$; 
otherwise $(v, v')$ is said to be \emph{incompatible} with $f$. Moreover, 
a positional strategy $\sigma_0\in\Sigma^M_0$ is said to be \emph{compatible} with $f$ whenever: $\forall^{v\in V_0}$  
if $\sigma_0(v)=v'$ then $(v,v')\in E$ is compatible with $f$; otherwise, $\sigma_0$ is \emph{incompatible} with $f$.

It is well-known that the family of all the SEPMs of a given $\Gamma$ forms a complete (finite) lattice, 
which we denote by $\mathcal{E}_\Gamma$ call it the \emph{Energy-Lattice} of $\Gamma$. Therefore, we shall consider: 
\[\mathcal{E}_{\Gamma}\triangleq \big(\{f:V\rightarrow \C_\Gamma \mid f \text{ is SEPM of } \Gamma\}, \sqsubseteq),\] 
where for any two SEPMs $f,g$ define $f \sqsubseteq g$ \textit{iff} $\forall{v\in V} f(v)\preceq g(v)$.
Notice that, whenever $f$ and $g$ are SEPM{s}, then so is the \emph{minimum function} defined as: 
$\forall^{v\in V} h(v)\triangleq \min\{f(v), g(v)\}$.
This fact allows one to consider the \emph{least} SEPM, namely, the unique SEPM $f^*:V\rightarrow \C_\Gamma$ such that, 
for any other SEPM $g:V\rightarrow \C_\Gamma$, the following holds: $\forall^{v\in V} f^*(v)\preceq g(v)$.
Thus, $\mathcal{E}_\Gamma$ is a complete lattice. So, $\mathcal{E}_\Gamma$ enjoys of \emph{Knaster–Tarski Theorem}, 
which states that the set of fixed-points of a monotone function on a complete lattice is again a complete lattice.

Also concerning SEPMs, we shall rely on the following lemmata. 
The first one relates SEPMs to the winning region $\W_0$ of Player~0 in \EG{s}.
\begin{Prop}[\cite{brim2011faster}]\label{prop:least_energy_prog_measure}
Let $\Gamma$ be an \EG. Then the following hold. 
\begin{enumerate}
\item If $f$ is any SEPM of the \EG $\Gamma$ and $v\in V_{f}$, 
then $v$ is a winning starting position for Player~$0$ in the \EG $\Gamma$.
Stated otherwise, $V_f\subseteq \W_0$;
\item If $f^*$ is the least SEPM of the \EG $\Gamma$, 
and $v$ is a winning starting position for Player~$0$ in the \EG $\Gamma$, then $v\in V_{f^*}$. 
Thus, $V_{f^*}=\W_0$.
\end{enumerate}
\end{Prop}
The following bound holds on the energy-levels of any SEPM (by definition of $\C_{\Gamma}$). 
\begin{Prop}\label{prop:lepm_equals_mincredit}
Let $\Gamma$ be an \EG. Let $f$ be any SEPM of $\Gamma$.

Then, for every $v\in V$ either $f(v)=\top$ or $0\leq f(v)\leq (|V|-1)W$.
\end{Prop}

\subsection{The Value-Iteration Algorithm for solving MCPs in EGs} 
In order to resolve MCPs in \EG{s}, the well-known  \emph{Value-Iteration}~\citep{brim2011faster} algorithm is employed.
Given an \EG $\Gamma$ as input, the Value-Iteration aims to compute the least SEPM $f^*$ of $\Gamma$.
This simple procedure basically relies on an \emph{energy-lifting operator} $\delta$.
Given $v\in V$, the energy-lifting operator $\delta(\cdot, v): 
[V\rightarrow \C_{\Gamma}]\rightarrow [V\rightarrow \C_\Gamma]$ is defined by $\delta(f,v)\triangleq g$, where:
\[
g(u)\triangleq \left\{
\begin{array}{ll}
f(u) & \text{ if } u\neq v \\
\min\{f(v') \ominus w(v,v') \mid v'\in N_{\Gamma}^{\text{out}}(v)\} & \text{ if } u=v\in V_0 \\
\max\{f(v') \ominus w(v,v') \mid v'\in N_{\Gamma}^{\text{out}}(v)\} & \text{ if } u=v\in V_1 \\
\end{array}
\right.
\]

We also need the following definition. Given a function $f:V\rightarrow \C_\Gamma$, we say that 
$f$ is \emph{inconsistent} in $v$ whenever one of the following two holds: 

1. $v\in V_0$ and  $\forall^{v'\in N_{\Gamma}^{\text{out}}(v)}\, f(v) \prec f(v') \ominus w(v,v')$;

2.  $v\in V_1$ and $\exists^{v'\in N_{\Gamma}^{\text{out}}(v)}\, f(v) \prec f(v') \ominus w(v,v')$.

In that case, we also say that $v$ is inconsistent \wrt $f$ in $\Gamma_{i,j}$.
 
To start with, the Value-Iteration algorithm initializes $f$ to the constant zero function, \ie $\forall^{v\in V} f(v)=0$. 
Furthermore, the procedure maintains a list $L^{\text{inc}}$ of vertices in order to witness the inconsistencies of $f$.
Initially, $v\in V_0\cap L^{\text{inc}}$ \textit{iff} all arcs going out of $v$ are negative, 
while $v\in V_1\cap L^{\text{inc}}$ if and only 
if $v$ is the source of at least one negative arc. Notice that checking the above conditions takes time $O(|E|)$.

While $L^{\text{inc}}$ is nonempty, the algorithm picks a vertex $v$ from $L^{\text{inc}}$ and performs the following:
\begin{enumerate}
\item Apply the lifting operator $\delta(f,v)$ to $f$ in order to resolve the inconsistency of $f$ in $v$;
\item Insert into $L^{\text{inc}}$ all vertices $u\in N_{\Gamma}^{\text{in}}(v)\setminus L^{\text{inc}}$ 
witnessing a new inconsistency due to the increase of $f(v)$. (Here, the same vertex can’t occur twice in $L^{\text{inc}}$.)
\end{enumerate}
The algorithm terminates when $L^{\text{inc}}$ is empty. This concludes the description of the Value-Iteration algorithm.
As shown in~\cite{brim2011faster}, 
the update of $L^{\text{inc}}$ following an application of the lifting operator $\delta(f,v)$ requires $O(|N_{\Gamma}^{\text{in}}(v)|)$ time. 
Moreover, a single application of the lifting operator $\delta(\cdot, v)$ takes $O(|N_{\Gamma}^{\text{out}}(v)|)$ time at most.
This implies that the algorithm can be implemented so that it will always halt within $O(|V||E|W)$ time  
(the reader is referred to \cite{brim2011faster} for all the details of the correctness and complexity proofs).

\begin{Rem} 
The Value-Iteration procedure lends itself to the following basic generalization, which is of a pivotal importance for us. 
Let $f^*$ be the least SEPM of the \EG $\Gamma$. Recall that, as a first step, 
the Value-Iteration algorithm initializes $f$ to be the constant zero function. Here, 
we remark that it is not necessary to do that really. 
Indeed, it is sufficient to initialize $f$ to be any function $f_0$ which bounds $f^*$ from below, that is to say, 
to initialize $f$ to any $f_0:V\rightarrow\C_\Gamma$ such that $\forall^{v\in V}\, f_0(v) \preceq f^*(v)$.
Soon after, $\L$ can be initialized in a natural way: just insert $v$ into $L^{\text{inc}}$ \textit{iff} $f_0$ is inconsistent at $v$.
This initialization still requires $O(|E|)$ time and it doesn't affect the correctness of the procedure.
\end{Rem}

\section{Values and Optimal Strategies from Reweightings}\label{section:values}
\paragraph*{Values and Farey sequences} 
Recall that each value $\val{\Gamma}{v}$ is contained within the following set of rational numbers:
$S_{\Gamma}=\left\{ N/D \suchthat D\in [1, |V|],\, N\in [-D\cdot W, D\cdot W] \right\}$. 
Let us introduce some notation in order to handle $S_{\Gamma}$ in a way that is suitable for our purposes.
Firstly, we write every $\nu\in S_{\Gamma}$ as $\nu=i+F$, 
where $i=i_\nu=\lfloor \nu \rfloor$ is the integral and $F=F_\nu=\{\nu\}=\nu-i$ is the fractional part. 
Notice that $i\in [-W, W]$ and that $F\in\Q$ is non-negative with denominator at most $|V|$.

Therefore, it is worthwhile to consider the \emph{Farey sequence} $\F_n$ of order $n=|V|$.
This is the increasing sequence of all irreducible fractions from the (rational) interval 
$[0,1]$ with denominators less than or equal to $n$. In the rest of this paper, 
$\F_n$ denotes the following sorted set:
\[ \F_n=\left\{ N/D \suchthat 0 \leq N \leq D\leq n, \gcd(N,D)=1 \right\}. \]
Farey sequences have numerous and interesting properties, in particular, many algorithms for 
generating the entire sequence $\F_n$ in time $O(n^2)$ are known in the literature~\cite{GKP:94}. 
The above mentioned quadratic running time is optimal, since $\F_n$ has $s(n) = \Theta(n^2)$ many elements.
We shall assume that $F_0, \ldots, F_{s-1}$ is an increasing ordering of $\F_n$, 
so that $\F_n=\{F_j\}_{j=0}^{s-1}$ and $F_j < F_{j+1}$ for every $j$. Notice that $F_0=0$ and $F_{s-1}=1$. 

For example, $\F_5=\{0, \frac{1}{5}, \frac{1}{4}, \frac{1}{3}, \frac{2}{5}, 
		\frac{1}{2}, \frac{3}{5}, \frac{2}{3}, \frac{3}{4}, \frac{4}{5}, 1 \}$. 

We will be interested in generating the sequence $F_0, \ldots, F_{s-1}$, one term after another, iteratively and efficiently.
As mentioned in \cite{PaPa09}, combining several properties satisfied by the Farey sequence, 
one can get a trivial iterative algorithm, which generates the next term $F_j=N_j/D_j$ of $\F_n$ based on the previous two:
\begin{align*}
	N_j & \leftarrow \left\lfloor \frac{D_{j-2}+n}{D_{j-1}} \right\rfloor \cdot N_{j-1} - N_{j-1}; &  
	D_j & \leftarrow \left\lfloor \frac{D_{j-2}+n}{D_{j-1}} \right\rfloor \cdot D_{j-1} - D_{j-1}. 
\end{align*}
Given $F_{j-2},F_{j-1}$, this computes $F_j$ in $O(1)$ time and space. It will perfectly fit our needs.

At this point, it is worth observing that $S_{\Gamma}$ can be represented as follows, this will be convenient in a while:
\[S_{\Gamma} = [-W, W) + \F_{|V|} = \left\{i+F_j \suchthat i\in [-W, W),\, j\in [0, s-1]\right\}. \]

\paragraph{A Characterization of Values in MPG{s} by reweighted EG{s}} 
It is now recalled a suitable characterization of optimal values in \MPG{s} in terms of winning regions.
\begin{Thm}[\cite{CR16}]\label{Thm:transition_opt_values}
Given an \MPG $\Gamma = ( V, E, w, \langle V_0, V_1 \rangle )$, let us consider the reweightings:
\[ \Gamma_{i,j}\cong\Gamma^{w-i-F_j} ,\,\text{ for any } i\in [-W, W] \text{ and } j\in [1, s-1], \]
where $s=|\F_{|V|}|$ and $F_j$ is the $j$-th term of the Farey sequence $\F_{|V|}$.

Then, the following holds: 
\[\val{\Gamma}{v} = i + F_{j-1} \textit{ iff } v\in \W_0(\Gamma_{i,j-1})\cap \W_1(\Gamma_{i, j}).\]
\end{Thm}
\begin{proof}
See the proof of [Theorem~3 in \cite{CR16}].
\end{proof}

\subsection{A Description of Optimal Positional Strategies in \MPG{s} from reweighted EG{s}}
We provide a sufficient condition, for a positional strategy to be optimal, which is expressed in terms of reweighted \EG{s} and their SEPM{s}.
\begin{Thm}[\cite{CR16}]\label{Thm:pos_opt_strategy} Let $\Gamma = (V , E, w, \langle V_0, V_1 \rangle)$ be an \MPG.
For each $u\in V$, consider the reweighted \EG $\Gamma_u \cong \Gamma^{w-\val{\Gamma}{u}}$. 
Let $f_{u}:V\rightarrow\C_{\Gamma_u}$ be any SEPM of $\Gamma_u$ such that $u\in V_{f_u}$ (\ie $f_u(u)\neq\top$). 
Moreover, we assume: $f_{u_1}=f_{u_2}$ whenever $\val{\Gamma}{u_1}=\val{\Gamma}{u_2}$.

When $u\in V_0$, let $v_{f_u}\in N_{\Gamma}^{\text{out}}(u)$ be any vertex such that $(u, v_{f_u})\in E$ is compatible with $f_u$ in \EG $\Gamma_u$,  
and consider the positional strategy $\sigma^*_0\in\Sigma^M_{0}$ defined as follows: $\forall^{u\in V_0}\, \sigma^*_0(u) \triangleq v_{f_u}$.

Then, $\sigma^*_0$ is an optimal positional strategy for Player~$0$ in the \MPG $\Gamma$.
\end{Thm}
\begin{proof}
See the proof of [Theorem~4 in \cite{CR16}].
\end{proof}

\begin{Rem}\label{Rem:pos_opt_strategy} Notice that Theorem~\ref{Thm:pos_opt_strategy} holds, in particular, 
when $f_u$ is the least SEPM $f^*_u$ of the reweighted \EG $\Gamma_u$. 
This is because $u\in V_{f^*_u}$ always holds for the least SEPM $f^*_u$ of the \EG $\Gamma_u$:
indeed, by Proposition~\ref{prop:relation_MPG_EG} and by definition of $\Gamma_u$, 
then $u$ is a winning starting position for Player~0 in the \EG $\Gamma_u$ (for some initial credit);
thus, by Proposition~\ref{prop:least_energy_prog_measure}, it follows that $u\in V_{f^*_u}$.
\end{Rem}

\section{A Faster $O(|V|^2 |E|\, W)$-Time Algorithm for \MPG{s} by Jumping through Reweighted EGs}\label{sect:algorithm}
This section offers an $O(|E|\log |V|) + 
	\Theta\big(\sum_{v\in V}\texttt{deg}_\Gamma(v) \cdot \ell_\Gamma^1(v)\big)=O(|V|^2 |E|\, W)$ 
time algorithm for solving the Value Problem and Optimal Strategy Synthesis in \MPG{s} 
$\Gamma = (V, E, w, \langle V_0, V_1\rangle)$, where $W\triangleq\max_{e\in E}|w_e|$; 
it works with $\Theta(|V|+|E|)$ space. Its name is Algorithm~\ref{algo:solve_mpg}.

In order to describe it in a suitable way, let us firstly recall some notation. 
Given an \MPG $\Gamma$, we shall consider the following reweightings:
\[ \Gamma_{i,j} \cong \Gamma^{w-i-F_j},\,\text{ for any } i\in [-W, W] \text{ and } j\in [1, s-1], \] 
where $s\triangleq |\F_{|V|}|$, and $F_j$ is the $j$-th term of $\F_{|V|}$.

Assuming $F_j=N_j/D_j$ for some (co-prime) $N_j,D_j\in \N$, we work with the following weights: 
\[
	w_{i,j}\triangleq w-i-F_j = w-i-N_j/D_j; \;\;\;\;\;\;\;\;\;\;\; 
	w'_{i,j}\triangleq D_j\, w_{i,j} = D_j\, (w-i) - N_j. 
\]
Recall $\Gamma_{i,j}\triangleq \Gamma^{w'_{i,j}}$ and $\forall^{e\in E} w'_{i,j}(e)\in \Z$. 
Notice, since $F_1 < \ldots < F_{s-1}$ is monotone increasing, 
$\{w_{i,j}\}_{i,j}$ can be ordered (inverse)-lexicographically \wrt $(i,j)$; 
\ie $w_{(i,j)}>w_{(i',j')}$ \textit{iff} either: $i<i'$, or $i=i'$ and $j<j'$; 
\eg $w_{W^-, 1} > w_{W^-, 2} > \ldots > w_{W^-, s-1} >\ldots > w_{W^+-1, s-1} > w_{W^+, s-1}$. 
Also, we denote the least-SEPM of the reweighted \EG $\Gamma_{i,j}$ by $f^*_{w'_{i,j}}:V\rightarrow \C_{\Gamma_{i,j}}$.
In addition, $f^*_{i,j}:V\rightarrow\Q$ denotes the \emph{rational-scaling} of $f^*_{w'_{i,j}}$, 
which is defined as: $\forall^{v\in V} f^*_{i,j}(v)\triangleq \frac{1}{D_j}\cdot f^*_{w'_{i,j}}(v)$. 
Finally, if $f$ is any SEPM of the \EG $\Gamma_{i,j}$, 
then $\text{Inc}(f,i,j)\triangleq\{v \in V\mid v \text{ is inconsistent \wrt } f \text{ in } \Gamma_{i,j}\}$.

\subsection{Description of Algorithm~\ref{algo:solve_mpg}} 
\textit{Outline.} Given an input arena $\Gamma=(V, E, w, \langle V_0, V_1 \rangle)$, 
Algorithm~\ref{algo:solve_mpg} aims at returning a tuple $(\W_0, \W_1, \nu, \sigma^*_0)$ where: 
$\W_0$ is the winning set of Player~$0$ in the \MPG $\Gamma$, and $\W_1$ is that of Player~$1$; 
$\nu:V\rightarrow S_\Gamma$ maps each starting position $v_s\in V$ to $\val{\Gamma}{v_s}$; 
finally, $\sigma^*_0:V_0\rightarrow V$ is an optimal positional strategy for Player $0$ in the \MPG $\Gamma$. 

Let $W^-\triangleq \min_{e\in E} w_e$ and $W^+\triangleq \max_{e\in E} w_e$. 
The first aspect underlying Algorithm~\ref{algo:solve_mpg} is that of ordering $[W^-, W^+]\times [1, s-1]$ lexicographically, 
by considering the above mentioned (decreasing) sequence of weights: 
\[ \rho  : [W^-, W^+]\times [1, s-1] \rightarrow \Z^E : (i,j) \mapsto w_{i,j}, \]
\[ \rho  : w_{W^-, 1} > w_{W^-, 2} > \ldots > w_{W^-, s-1} > w_{W^-+1, 1} > 
		w_{W^-+1, 2} > \ldots  > w_{W^+-1, s-1} > \ldots > w_{W^+, s-1}, \]
then, to rely on Theorem~\ref{Thm:transition_opt_values}, at each step of $\rho$, 
testing whether some \emph{transition of winning regions} occurs.  
At the generic $(i,j)$-th step of $\rho$, we run a Value-Iteration \citep{brim2011faster} 
in order to compute the least-SEPM of $\Gamma_{i,j}$,   
and then we check for every $v\in V$ whether $v$ is winning for Player~$1$ 
\wrt the \emph{current} weight $w_{i,j}$ (\ie \wrt $\Gamma_{i,j}$), 
meanwhile recalling whether $v$ was winning for Player~$0$ \wrt 
the (immediately, inverse-lex) \emph{previous} weight $w_{\texttt{prev}_{\rho}(i,j)}$ (\ie \wrt $\Gamma_{\texttt{prev}_{\rho}(i,j)}$). 
This step relies on Proposition~\ref{prop:least_energy_prog_measure}, 
as in fact $\W_0(\Gamma_{\texttt{prev}_{\rho}(i,j)})=V_{f^*_{\texttt{prev}_{\rho}(i,j)}}$ 
	and $\W_1(\Gamma_{i,j})=V\setminus V_{f^*_{i,j}}$.

If a transition occurs, say for some $\hat{v} \in \W_0(\Gamma_{\texttt{prev}_{\rho}(i,j)}) \cap \W_1(\Gamma_{i,j})$, 
then $\val{\Gamma}{\hat{v}}$ can be computed easily by relying on Theorem~\ref{Thm:transition_opt_values}, 
\ie $\nu(\hat{v})\leftarrow i + F[j-1]$; also, 
an optimal positional strategy can be extracted from $f^*_{\texttt{prev}_{\rho}(i,j)}$
thanks to Theorem~\ref{Thm:pos_opt_strategy} and Remark~\ref{Rem:pos_opt_strategy}, provided that $\hat{v}\in V_0$.

Each phase, in which one does a Value-Iteration and looks at transitions of winning regions, 
it is named \emph{Scan-Phase}. Remarkably, for every $i\in [W^-, W^+]$ and $j\in [1,s-1]$, 
the \emph{$(i,j)$-th} Scan-Phase performs a Value-Iteration~\citep{brim2011faster} on 
the reweighted \EG $\Gamma_{i,j}$ by initializing all the energy-levels to those computed by the previous 
Scan-Phase (subject to a suitable re-scaling and a rounding-up, \ie $\lceil D_j\cdot f^*_{\texttt{prev}_{\rho}(i,j)}\rceil$). 
As described in~\cite{CR16}, the main step of computation that is carried 
on at the $(i,j)$-th Scan-Phase goes therefore as follows: 
\[ f_{i,j} \leftarrow \frac{1}{D_j}\, 
	\texttt{Value-Iteration}\Big(\Gamma_{i,j}, \big\lceil D_j\cdot f^*_{\texttt{prev}_{\rho}(i,j)}\big\rceil\Big), \]
where $D_j$ is the denominator of $F_j$. Then, one can prove that $\forall{(i,j)}\, f_{i,j}=f^*_{i,j}$ 
[\cite{CR16}, Lemma~8, Item~4]. Indeed, 
	by Propositions~\ref{prop:relation_MPG_EG} and Proposition~\ref{prop:least_energy_prog_measure}, 
$\W_0(\Gamma_{\texttt{prev}_{\rho}(i,j)})=V_{f^*_{\texttt{prev}_{\rho}(i,j)}}$ and $\W_1(\Gamma_{i,j})=V\setminus V_{f^*_{i,j}}$.
And since $\rho$ is monotone decreasing, the sequence of 
	energy-levels $\psi_\rho : (i,j)\mapsto f^*_{i,j}$ is monotone non-decreasing [\cite{CR16}, Lemma~8, Item~1]:
\[ \psi_\rho : f^*_{W^-, 1} \preceq f^*_{W^-, 2} \preceq  \ldots \preceq f^*_{W^-, s-1} 
	\preceq f^*_{W^-+1, 1} \preceq f^*_{W^-+1, 2} \preceq \ldots 
	\preceq f^*_{W^+-1, s-1} \preceq \ldots \preceq f^*_{W^+, s-1}; \]
Our algorithm will succeed at \emph{amortizing} the cost of the corresponding 
	sequence of Value-Iterations for computing $\psi_\rho$. A similar amortization takes place already in Algorithm~0. 

However, Algorithm~0 performs exactly one Scan-Phase (\ie one Value-Iteration, 
	plus the tests $v \in^{?} \W_0(\Gamma_{\texttt{prev}_{\rho}(i,j)}) \cap \W_1(\Gamma_{i,j})$) 
for each term of $\rho$ --without making any \emph{Jump} in $\rho$--. So, 
	Algorithm~0 performs $\Theta(|V|^2 W)$ Scan-Phases overall, each one costing $\Omega(|E|)$ time 
(\ie the cost of initializing the Value-Iteration as in~\cite{brim2011faster}). 
	This brings an overall $\Omega(|V|^2|E|W)$ time complexity, which turns out to be also $O(|V|^2|E|W)$; 
leading us to an improved pseudo-polynomial time upper bound for solving \MPG{s}~\citep{CR15, CR16}.

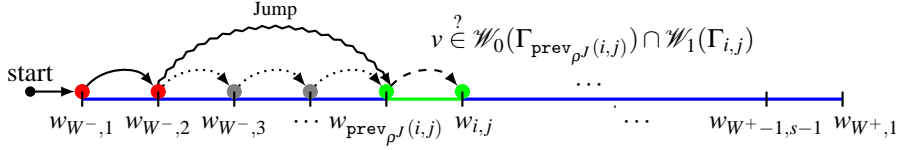
\begin{figure}[!t]
\begin{center}
\begin{tikzpicture}[xscale=1.0]
\draw[-][draw=blue, very thick] (0,0) -- (4,0);
\draw[-][draw=green, very thick] (4,0) -- (5,0);
\draw[-][draw=blue, very thick]  (5,0) -- (10,0);

\node[node, color=red, fill=red, scale=.3] (A) at (0,.1) {};
\node[node, left=of A, xshift=3ex, color=black, fill=black, scale=.2, label={above:start} ] (S) {};
\node[node, color=red, fill=red, scale=.3] (B) at (1,.1) {};
\node[node, color=gray, fill=gray, scale=.3] (C) at (2,.1) {};
\node[node, color=gray, fill=gray, scale=.3] (D) at (3,.1) {};
\node[node, color=green, fill=green, scale=.3] (E) at (4,.1) {};
\node[node, color=green, fill=green, scale=.3] (F) at (5,.1) {};
\node[-] (G) at (6,.1) {};

\draw[thick, arrows=->] (S) to [ bend left=0] node[below, xshift=-1.5ex, yshift=1ex] {} (A);
\draw[thick, arrows=->] (A) to [ bend left=50] node[below, xshift=-1.5ex, yshift=1ex] {} (B); 
\draw[thick, arrows=->, line join=round,
decorate, decoration={
    zigzag,
    segment length=5,
    amplitude=.5,post=lineto,
    post length=2pt
} ] (B.north) to [ bend left=60] node[above, xshift=0ex, yshift=0ex, scale=.75] {Jump} (E.north east); 
\draw[thick, dotted,arrows=->] (B) to [ bend left=50] node[below, xshift=-1.5ex, yshift=1ex] {} (C); 
\draw[thick, dotted, arrows=->] (C) to [ bend left=50] node[below, xshift=-1.5ex, yshift=1ex] {} (D); 
\draw[thick, dotted,arrows=->] (D) to [ bend left=50] node[below, xshift=-1.5ex, yshift=1ex] {} (E); 
\draw[thick, dashed, arrows=->] (E) to [bend left=50] node[below, xshift=-3ex, yshift=-1ex, 
label={above right:$v\overset{?}{\in} \W_0(\Gamma_{\texttt{prev}_{\rho^J}(i,j)}) \cap \W_1(\Gamma_{i,j})$}] {} (F); 

\draw [thick] (0,-.1) node[below]{$w_{W^-, 1}$} -- (0,0.1);
\draw [thick] (1,-.1) node[below]{$w_{W^-, 2}$} -- (1,0.1);
\draw [thick] (2,-.1) node[below]{$w_{W^-, 3}$} -- (2,0.1);
\draw [thick] (3,-.1) node[below, xshift=-.2ex]{$\cdots$} -- (3,0.1);
\draw [thick] (4,-.1) node[below]{$w_{\texttt{prev}_{\rho^J}(i,j)}$} -- (4,0.1);
\draw [thick] (5,-.1) node[below, xshift=1ex]{$w_{i, j}$} -- (5,0.1);
\draw [thick] (7,-.1) node[below, xshift=2ex]{$\cdots$} -- (7,-.1);
\draw [thick] (7,-.1) node[above, xshift=-2ex, yshift=.5ex]{$\cdots$} -- (7,-.1);
\draw [thick] (9,-.1) node[below]{$w_{W^+-1, s-1}$} -- (9,0.1);
\draw [thick] (10,-.1) node[below, xshift=2ex]{$w_{W^+, 1}$} -- (10,0.1);
\end{tikzpicture}
\end{center}
\caption{An illustration of Algorithm~\ref{algo:solve_mpg}.}\label{fig:algo_solve}
\end{figure}

The present work shows that it is instead possible, and actually very convenient, to perform many \emph{Jumps} in $\rho$; 
thus introducing ``gaps" between the weights that are considered along the sequence of Scan-Phases. 
The corresponding sequence of weights is denoted by $\rho^J$. This is Algorithm~\ref{algo:solve_mpg}. 
In \figref{fig:algo_solve}, a graphical intuition of Algorithm~\ref{algo:solve_mpg} and $\rho^J$ is given, 
in which a Jump is depicted with an arc going from $w_{W^-,2}$ to $w_{\texttt{prev}_{\rho^J}(i,j)}$, 
	\eg $w_{\texttt{prev}_{\rho^J}(\texttt{prev}_{\rho^J}(i,j))}=w_{W^-,2}$. 

Two distinct kinds of Jumps are employed: \emph{Energy-Increasing-Jumps (EI-Jumps)} and \emph{Unitary-Advance-Jumps (UA-Jumps)}. 
	Briefly, EI-Jumps allow us to satisfy a suitable invariant: 

[\emph{Inv-EI}] Whenever a Scan-Phase is executed (each time that a Value-Iteration is invoked), 
	an energy-level $f(v)$ strictly increases for at least one $v\in V$. There will be no \emph{vain} Scan-Phase 
(\ie such that all the energy-levels stand still); so, $\delta$ will be applied (successfully) at least once per Scan-Phase. 
Therefore, $\psi_{\rho^J}$ will be monotone increasing (except at the steps of backtracking introduced next, 
	but there will be at most $|V|$ of them). $\Box$ 

Indeed, the UA-Jumps are employed so to scroll through $\F_{|V|}$ only \emph{when} 
	and \emph{where} it is really necessary. Consider the following facts.

-- Suppose that Algorithm~\ref{algo:solve_mpg} came at the end of the $(i,s-1)$-th Scan-Phase, 
	for some $i\in [W^-,W^+]$; recall that $F_{s-1}=1$, so $w_{i,s-1}=w'_{i,s-1}$ is integral. 
Then, Algorithm~0 would scroll through $\F_{|V|}$ entirely, by invoking one Scan-Phase per each term, 
going from the $(i+1,1)$-th to the $(i+1, s-1)$-th, meanwhile testing whether a transition of winning regions occurs; 
notice, $w_{i+1, s-1}$ is integral again. Instead, 
to UA-Jump means to jump in advance (proactively) from $w_{i,s-1}$ to $w_{i+1,s-1}$, 
by making a Scan-Phase on input $\Gamma_{i+1,s-1}$, 
thus skipping all those from the $(i+1,1)$-th to the $(i+1, s-2)$-th one. After that, 
	Algorithm~\ref{algo:solve_mpg} needs to \emph{backtrack} to $w_{i,s-1}$, and to scroll through $\F_{|V|}$, 
if and only if $\W_0(\Gamma_{w_{i,s-1}}) \cap \W_1(\Gamma_{w_{i+1,s-1}})\neq\emptyset$. 
	Otherwise, it is safe to keep the search going on, from $w_{i+1,s-1}$ on out, 
making another UA-Jump to $w_{i+2,s-1}$. The backtracking step may happen at most $|V|$ times overall, 
because some value $\nu(v)$ is assigned to some $v\in V$ at each time. 
So, Algorithm~\ref{algo:solve_mpg} scrolls entirely through $\F_{|V|}$ at most $|V|$ times; 
\ie only \emph{when} it is really necessary.

-- Remarkably, when scrolling through $\F_{|V|}$, soon after the above mentioned backtracking step, 
	the corresponding sequence of Value-Iterations really need to lift-up again (more slowly) \emph{only} the 
energy-levels of the sub-arena of $\Gamma$ that is induced by 
	$S\triangleq \W_0(\Gamma_{w_{i,s-1}}) \cap \W_1(\Gamma_{w_{i+1,s-1}})$. 
All the energy-levels of the vertices in  
$V\setminus S$ can be confirmed and left unchanged during the UA-Jump's backtracking step; 
	and they will all stand still, during the 
forthcoming sequence of Value-Iterations (at least, until a new EI-Jump will occur), 
as they were computed just \emph{before} the occurence of the UA-Jump's backtracking step. 
This is why Algorithm~\ref{algo:solve_mpg} scrolls through $\F_{|V|}$ only \emph{where} it is really necessary. 

-- Also, Algorithm~\ref{algo:solve_mpg} succeeds at 
\emph{interleaving} EI-Jumps and UA-Jumps, thus making only one single pass through $\rho^J$ (plus the backtracking steps).

Altogether these facts are going to reduce the running time considerably.

\begin{Def}[$\ell_{\Gamma}^1$]
Given an input \MPG $\Gamma$, let $\ell_{\Gamma}^1(v)$ be the total number 
of times that the energy-lifting operator $\delta(\cdot, v)$ is applied to any $v\in V$ by Algorithm~\ref{algo:solve_mpg} 
(notice that it will be applied only at line~\ref{algo:jvalue:l3} of $\texttt{J-VI}()$, 
see SubProcedure~\ref{algo:j-value-iteration}). 
\end{Def}
Then, the following remark holds on Algorithm~\ref{algo:solve_mpg}.
\begin{Rem} 
Jumping is not heuristic, the theoretical running time of the procedure improves exactly, 
from $\Theta(|V|^2|E|W + \sum_{v\in V}\texttt{deg}_\Gamma(v)\cdot\ell_\Gamma^0(v))$~\citep{CR16} 
to $O(|E|\log |V|) + \Theta\big(\sum_{v\in V}\texttt{deg}_\Gamma(v)\cdot\ell_\Gamma^1(v)\big)$ (Algorithm~\ref{algo:solve_mpg}), 
where $\ell_\Gamma^1\leq (|V|-1)|V|W$; which is still $O(|V|^2 |E|\, W)$ in the worst-case, 
but it isn't known to be $\Omega(|V|^2|E|W)$ generally. 
In practice, this reduces the magnitude of $\ell_\Gamma$ significantly, 
\ie $\ell_\Gamma^1 \ll \ell_\Gamma^0$ is observed in our experiments (see SubSection~\ref{subsect:experiments}).
\end{Rem}

To achieve this, we have to overcome some subtle issues. Firstly, 
we show that it is unnecessary to re-initialize the Value-Iteration 
at each Scan-Phase (this would cost $\Omega(|E|)$ each time otherwise), even when making wide jumps in $\rho$. 
Instead, it will be sufficient to perform an initialization phase only at the beginning, 
by paying only $O(|E|\log |V|)$ total time and a linear space in pre-processing. 
For this, we will provide a suitable readjustment of the Value-Iteration; 
it is named $\texttt{J-VI}()$ (SubProcedure~\ref{algo:j-value-iteration}). 
Briefly, the Value-Iteration of \cite{brim2011faster} employs an array of counters, 
$\texttt{cnt}:V_0\rightarrow\N$, in order to check in time $O(|N_\Gamma^{\text{in}}(v)|)$ 
which vertices $u\in N_\Gamma^{\text{in}}(v)\cap V_0$ have become inconsistent (soon after that the energy-level 
$f(v)$ was increased by applying $\delta(f, v)$ to some $v\in V$),  
and should therefore be added to the list $L^{\text{inc}}$ of inconsistent vertices. 
One subtle issue here is that, when going from the $\texttt{prev}_{\rho^J}(i,j)$-th  
to the $(i,j)$-th Scan-Phase, the \emph{coherency} of $\texttt{cnt}$ can 
break (\ie $\texttt{cnt}$ may provide false-positives, thus classifying a vertex as consistent when it isn't really so). 
This may happen when $w_{\texttt{prev}_{\rho^J}(i,j)} > w_{(i,j)}$ (which is always the case, 
except for the UA-Jump's backtracking steps). This is even amplified by the EI-Jumps, as they may lead to wide jumps in $\rho$. 
The algorithm in~\cite{CR16} recalculates $\texttt{cnt}$ from scratch, at the beginning of each Scan-Phase, 
thus paying $\Omega(|E|)$ time per each. In this work, 
we show how to keep $\texttt{cnt}$ coherent throughout the Jumping Scan-Phases, efficiently. 
Actually, even in Algorithm~\ref{algo:solve_mpg} the coherency of $\texttt{cnt}$ can possibly break, 
but Algorithm~\ref{algo:solve_mpg} succeeds at \emph{repairing} all the incoherencies that may happen 
during the whole computation in $\Theta(|E|)$ \emph{total} time -- just by paying $O(|E|\log|V|)$ time in pre-processing. 
This is a very convenient trade-off. At this point we should begin entering into the details of Algorithm~\ref{algo:solve_mpg}.

\textit{Jumper.} We employ a container data-structure, denoted by $\J$. It comprises a bunch of arrays, 
maps, plus an integer variable $\J.i$. Concerning maps, the key universe is $V$ or $E$; 
\ie keys are restricted to a narrow range of integers ($[1,|V|]$ or $[1,|E|]$, depending on the particular case). 

We suggest direct addressing: the value binded to a key $v\in V$ (or $(u,v)\in E$) is stored at $A[v]$ (resp., $A[(u,v)]$); 
if there is no binding for key $v$ (resp., $(u,v)$), the cell stores a sentinel, \ie $A[v]=\bot$. 
Also, we would need to iterate efficiently through $A$ (\ie without having to scroll entirely through $A$). 
This is easy to implement by handling pointers in a suitable way; 
one may also keep a list $L_A$ associated to $A$, explicitly, 
storing one element for each $(k,v)\neq\bot$ of $A$; every time that an item is added to or removed from $A$, 
then $L_A$ is updated accordingly, in time $O(1)$ (by handling pointers). 
The \emph{last} entry inserted into $A$ (the key of which isn't already binded at insertion time) goes in \emph{front} of $L_A$. 
We say that $L\triangleq (A,L_A)$ is an \emph{array-list}, 
and we dispose of the following operations: $\texttt{insert}((k, \texttt{val}), L)$, 
which binds $\texttt{val}$ to $k$ by inserting $(k, \texttt{val})$ into $L$ 
(if any $(k, \texttt{val}')$ is already in $L$, then $\texttt{val}'$ gets overwritten by $\texttt{val}$); 
$\texttt{remove}(k, L)$ deletes an entry $(k, \texttt{val})$ from $L$; $\texttt{pop\_front}(L)$, 
removes from $L$ the \emph{last} $(k, \texttt{val})$ that was inserted (and 
whose key was not already binded at the time of the insertion, 
\ie the \emph{front}) also returning it; $\texttt{for\_each}\big((k,\texttt{val})\in L\big)$ 
iterates through the entries of $L$ efficiently (\ie skipping the sentinels). 
Notice, any sequence of $\texttt{insert}$ and \texttt{pop\_front} on $L$ implements a LIFO policy. 

So, $\J$ comprises: an integer variable $\J.i$; an array $\J.f : V \rightarrow \Q$; 
an array $\J.\texttt{cnt}:V_0\rightarrow \N$; 
an array $\J.\texttt{cmp} : \{(u,v)\in E\mid u\in V_0\}\rightarrow \{\texttt{T}, \texttt{F}\}$; a bunch of array-lists, 
$L_f : V\rightarrow \N$, and $L^{\text{inc}}, L^{\text{inc}}_{\text{nxt}}, 
L^{\text{inc}}_{\text{cpy}}, L_\top : V\rightarrow \{\ast\}$;
finally, a special array-list $L_{\omega}$ indexed by $\{w_e\mid e\in E\}$, 
whose values are in turn (classical, linked) lists of arcs, denoted $L_{\alpha}$; $L_{\omega}$ 
is filled in pre-processing as follows: 
$(\hat{w}, L_{\alpha})\in L_{\omega}$ \textit{iff} $L_{\alpha}=\{e\in E\mid w_e=\hat{w}\}$.  
The subprocedure $\texttt{init\_jumper}()$ (SubProcedure~\ref{proc:jumper_init}) takes care of initializing $\J$.

\begin{wrapfigure}[13]{r}{6cm}
\raisebox{0pt}[\dimexpr\height-2\baselineskip\relax]{
\begin{algo-proc}[H]
\caption{Init Jumper $\J$}\label{proc:jumper_init}
\scriptsize
\DontPrintSemicolon
\nonl \SetKwProg{SubFn}{SubProcedure}{}{}
\SubFn{$\texttt{init\_jumper}(\J,\Gamma)$}{
\SetKwInOut{Input}{input}
\SetKwInOut{Output}{output}
\Input{Jumper $\J$, an \MPG $\Gamma$.}
$L_f, L^{\text{inc}}, L^{\text{inc}}_{\text{nxt}}, 
	L^{\text{inc}}_{\text{cpy}}, L_{\top}, L_{\omega} \leftarrow\emptyset$; \label{proc:jumper_init:l1}\; 
\ForEach{$v\in V$}{ \label{proc:jumper_init:l2}
	$\J.f[v]\leftarrow 0$; \label{proc:jumper_init:l3} \; 
	\If{$v\in V_0$ \label{proc:jumper_init:l4} }{
		$\J.\texttt{cnt}[v]\leftarrow |N^{\text{out}}_{\Gamma}(v)|$; \label{proc:jumper_init:l5}
	}
} 
\ForEach{$(u,v,w)\in E$}{ \label{proc:jumper_init:l6}
	\If{$v\in V_0$ \label{proc:jumper_init:l7} }{
		$\J.\texttt{cmp}[(u,v)]\leftarrow \texttt{T}$; \label{proc:jumper_init:l8} 
	}
	\If{$L_{\omega}[w]=\bot$}{ \label{proc:jumper_init:l9}
		$\texttt{insert}\big((w,\emptyset), L_{\omega})\big)$; \label{proc:jumper_init:l10}\;
	}
	$\texttt{insert}\big((u,v), L_{\omega}[w]\big)$; \label{proc:jumper_init:l11}\;
}
}
Sort $L_{\omega}$ in increasing order \wrt the keys $w$; \label{proc:jumper_init:l12} 
\end{algo-proc} }
\end{wrapfigure}

At the beginning, all array-lists are empty (line~\ref{proc:jumper_init:l1}). 
For every $v\in V$ (line~\ref{proc:jumper_init:l2}), we set $\J.f[u] = 0$ and, if $v\in V_0$, then 
$\J.\texttt{cnt}[v] \leftarrow |N^{\text{out}}_{\Gamma}(v)|$ (lines~\ref{proc:jumper_init:l3}-\ref{proc:jumper_init:l5}). 
Then, each arc $(u,v,w)\in E$ is flagged as \emph{compatible}, 
\ie $\J.\texttt{cmp}[(u,v)]\leftarrow \texttt{T}$ (lines~\ref{proc:jumper_init:l6}-\ref{proc:jumper_init:l8}); 
also, if $L_{\omega}$ doesn't contain an entry already binded to $w(u,v)$, 
then an empty list of arcs is inserted into $L_{\omega}$ as an entry 
$(w,\emptyset)$ (lines~\ref{proc:jumper_init:l9}-\ref{proc:jumper_init:l10}); then, in any case, 
the arc $(u,v)$ is added to the unique $L_{\alpha}$ which is binded to $w=w(u,v)$ in $L_\omega$ (line~\ref{proc:jumper_init:l11}). 
At the end (line~\ref{proc:jumper_init:l12}), all the elements of $L_{\omega}$ are sorted in increasing order \wrt their weight keys, 
$w_e$ for $e\in E$ (\eg $(W^{-}, L_{\alpha})$ goes in front of $L_{\omega}$). 
This concludes the initialization of $\J$; it takes $O(|E|\log|V|)$ time and $\Theta(|V|+|E|)$ space.

\textit{Main Procedure: $\texttt{solve\_\MPG}()$.} The main procedure of Algorithm~\ref{algo:solve_mpg} is organized as follows.
Firstly, the algorithm performs an initialization phase; which includes $\texttt{init\_jumper}(\J,\Gamma)$. 

The variables $\W_0,\W_1,\nu,\sigma^*_0$ are initially empty (line~\ref{algo:solve:l1}). 
Also, $W^-\leftarrow \min_{e\in E} w_e$, $W^+\leftarrow \max_{e\in E} w_e$ (line~\ref{algo:solve:l2}). 
And $F$ is a reference to the Farey's terms, 
say $\{F[j]\mid j\in [0, s-1]\}=\F_{|V|}$, and $s\leftarrow |\F_{|V|}|$ (line~\ref{algo:solve:l3}). 
At line~\ref{algo:solve:l4}, $\J$ is initialized by $\texttt{init\_jumper}(\J,\Gamma)$ (SubProcedure~\ref{proc:jumper_init}). 

Then the Scan-Phases start. 

\begin{wrapfigure}[14]{r}{7cm}
\raisebox{0pt}[\dimexpr\height-1.9\baselineskip\relax]{
\begin{algorithm}[H]
\caption{Main Procedure}\label{algo:solve_mpg}
\scriptsize
\DontPrintSemicolon
\nonl \SetKwProg{Fn}{Procedure}{}{}
\Fn{$\texttt{solve\_MPG}(\Gamma)$}{
    \SetKwInOut{Input}{input}
    \SetKwInOut{Output}{output}
\Input{An \MPG $\Gamma= ( V, E, w, \langle V_0, V_1 \rangle )$.}
\Output{$(\W_0, \W_1, \nu, \sigma^*_0)$.}  
$\W_0, \W_1, \nu, \sigma^*_0\leftarrow\emptyset$; \textit{/*Init Phase*/} \label{algo:solve:l1} \; 
$W^-\leftarrow \min_{e\in E} w_e$; $W^+\leftarrow \max_{e\in E} w_e$; \label{algo:solve:l2} \;
$F\leftarrow $ reference to $\F_{|V|}$; $s\leftarrow |\F_{|V|}|$; \label{algo:solve:l3} \; 
$\texttt{init\_jumper}(\J,\Gamma)$; \label{algo:solve:l4} \; 
$i\leftarrow W^--1$; $j\leftarrow 1$; \label{algo:solve:l5} \textit{/*Jumping Scan-Phases*/}\;
\While{\texttt{T}}{ \label{algo:solve:l6} 
	\If{$\texttt{ei-jump}(i,\J)$ \label{algo:solve:l7}}{ 
		\lIf{$L^{\text{inc}}=\emptyset$}{
			\Return{$(\W_0, \W_1, \nu, \sigma^*_0)$;} \label{algo:solve:l8}
		}
		$(i, S)\leftarrow \texttt{ua-jumps}(\J.i, s, F, \J, \Gamma)$; \label{algo:solve:l9} \; 
		$j\leftarrow 1$; \label{algo:solve:l10} \; 
	} 
	$\texttt{J-VI}(i, j, F, \J, \Gamma[S]);$ \label{algo:solve:l11} \;	
	$\texttt{set\_vars}(\W_0, \W_1, \nu, \sigma^*_0, i, j, F, \J, \Gamma[S])$; \label{algo:solve:l12} \; 
	$\texttt{scl\_back\_}f(j,F,\J)$; \label{algo:solve:l13} \; 
	$j\leftarrow j+1$; \label{algo:solve:l14} \; 
}
}
\end{algorithm}}
\end{wrapfigure}
After setting $i\leftarrow W^--1$, $j\leftarrow 1$ (line~\ref{algo:solve:l5}), 
Algorithm~\ref{algo:solve_mpg} enters into a \texttt{while} loop (line~\ref{algo:solve:l6}), 
which lasts until both $\texttt{ei-jump}(i,\J)=\texttt{T}$ at line~\ref{algo:solve:l7}, 
\emph{and} $L^{\text{inc}}=\emptyset$ at line~\ref{algo:solve:l8}, hold; 
in which case $(\W_0,\W_1,\nu,\sigma^*_0)$ is returned (line~\ref{algo:solve:l8}) and Algorithm~\ref{algo:solve_mpg} halts. 
Inside the \texttt{while} loop, $\texttt{ei-jump}(i,\J)$ (SubProcedure~\ref{proc:energy_jump}) 
is invoked (line~\ref{algo:solve:l7}). This checks whether or not to make an EI-Jump; 
if so, the ending point of the EI-Jump (the new value of $i$) is stored into $\J.i$.  
This will be the starting point for making a sequence of UA-Jumps, 
which begins by invoking $\texttt{ua-jumps}(\J.i, s, F, \J, \Gamma)$ at line~\ref{algo:solve:l9}.
When the $\texttt{ua-jumps}()$ halts, it returns $(\hat{i},S)$, where: $\hat{i}$ is the new value of $i$ (line~\ref{algo:solve:l9}), 
for some $\hat{i}\geq \J.i$; and $S$ is a set of vertices such 
that $S = \W_0(\Gamma_{w_{\hat{i}-1,s-1}}) \cap \W_1(\Gamma_{w_{\hat{i},s-1}})$. 
Next, $j\leftarrow 1$ is set (line~\ref{algo:solve:l10}), 
as Algorithm~\ref{algo:solve_mpg} is now completing the backtracking from $w_{\hat{i},s-1}$ to $w_{\hat{i},1}$, 
in order to begin scrolling through $\F_{|V|}$ by running a sequence of $\texttt{J-VI}()$ at line~\ref{algo:solve:l11}. 
Such a sequence of $\texttt{J-VI}()${s} will last until the occurence of another EI-Jump at line~\ref{algo:solve:l7}, 
that in turn will lead to another sequence of UA-Jumps at line~\ref{algo:solve:l9}, 
and so on. So, a $\texttt{J-VI}()$ (SubProcedure~\ref{algo:j-value-iteration}) 
is executed on input $(\hat{i},j,F,\J,\Gamma[S])$ at line~\ref{algo:solve:l11}. 
We remark that, during the $\texttt{J-VI}(i,j,F,\J,\Gamma[S])$, 
the energy-levels are scaled up, from $\Q$ to $\N$; actually, from $\J.f$ to $\lceil D_j\cdot \J.f \rceil$, 
where $D_j$ is the denominator of $F_j$. Also, $\texttt{J-VI}(i,j,F,\J,\Gamma[S])$ (SubProcedure~\ref{algo:j-value-iteration}) 
is designed so that, when it halts, $L_\top = \W_0(\Gamma_{\texttt{prev}_{\rho^{\J}}(i,j)}) \cap \W_1(\Gamma_{i,j})$. 
Then, $\texttt{set\_vars}()$ is invoked on input $(\W_0, \W_1, \nu, \sigma^*_0, i, j, F, \J, \Gamma[S])$ (line~\ref{algo:solve:l12}): 
this checks whether some value and optimal strategy needs to be assigned to $\nu$ and $\sigma^*_0$ (respectively). 
Next, all of the energy-levels are scaled back, from $\N$ to $\Q$, and stored back into $\J.f$: 
this is done by invoking $\texttt{scl\_back\_}f(j,F,\J)$ (line~\ref{algo:solve:l13}). 
Finally, $j\leftarrow j+1$ (line~\ref{algo:solve:l14}) is assigned (to step through the sequence $\F_{|V|}$ during the \texttt{while} loop at line~\ref{algo:solve:l7}). 
This concludes $\texttt{solve\_MPG}()$, which is the main procedure of Algorithm~\ref{algo:solve_mpg}.

\emph{Set Values and Optimal Strategy.} Let us provide the details 
of $\texttt{set\_vars}()$ (SubProcedure~\ref{proc:set_values_and_strategies}). 
It takes $(\W_0, \W_1, \nu, \sigma^*_0, i, j, F, \Gamma)$ in input, 
where $i\in [W^-, W^+]$ and $j\in [1,s-1]$. At line~\ref{proc:set:l0}, $D=D_{j-1}$ is the denominator of $F_{j-1}$. 
Then, all of the following operations are repeated while $L_{\top}\neq\emptyset$ (line~\ref{proc:set:l1}). 
Firstly, the front element $u$ of $L_{\top}$ is popped (line~\ref{proc:set:l2}); 
recall, it will turn out that $u\in \W_0(\Gamma_{i,j-1}) \cap \W_1(\Gamma_{i,j})$, 
thanks to the specs of $\texttt{J-VI}()$ (SubProcedure~\ref{algo:j-value-iteration}). 
For this reason, the optimal value of $u$ in the \MPG $\Gamma$ is set to $\nu(u)\leftarrow i+F[j-1]$ (line~\ref{proc:set:l3}); 
and, if $\nu(u)\geq 0$, $u$ is added to the winning region $\W_0$; else, to $\W_1$ (line~\ref{proc:set:l4}). 
The correctness of lines~\ref{proc:set:l3}-\ref{proc:set:l4} relies on Theorem~\ref{Thm:transition_opt_values}.
If $u\in V_0$ (line~\ref{proc:set:l5}), it is searched an arc $(u,v)\in E$ 
that is compatible \wrt $D_{j-1}\cdot\J.f$ in $\Gamma_{i,j-1}$ (line~\ref{proc:set:l7}), 
\ie it is searched some $v\in N_\Gamma^{\text{out}}(u)$ such that: 
\[ (D\cdot\J.f[u]) \succeq (D\cdot\J.f[v]) \ominus \texttt{get\_scl\_}w\big(w(u, v), i, j-1, F\big) \; \text{ (line~\ref{proc:set:l7})}; \] 
By Theorem~\ref{Thm:pos_opt_strategy}, setting $\sigma^*_0(u)\leftarrow v$ (line~\ref{proc:set:l8}) 
brings an optimal positional strategy for Player~0 in the \MPG $\Gamma$. 
Here, $\texttt{get\_scl\_}w\big(w,i,j-1,F\big)$ simply returns $D_{j-1}\cdot(w(u,v)-i)-N_{j-1}$, 
where: $N_{j-1}$ is the numerator of $F_{j-1}$, and $D_{j-1}$ is its denominator. 
Thanks to how $\texttt{J-VI}()$ (SubProcedure~\ref{algo:j-value-iteration}) is designed, 
at this point $\J.f$ still stores the energy-levels as they were just \emph{before} the last invocation 
of $\texttt{J-VI}()$ made at line~\ref{algo:solve:l11} of Algorithm~\ref{algo:solve_mpg}; 
instead, the new energy-levels, those lifted-up during that same $\texttt{J-VI}()$, are stored into $L_f$.
So, at this point, it will turn out that $\forall^{u\in V}\, \J.f[u]=f^*_{i,j-1}(u)$. 

\begin{algo-proc}[H]
\caption{Set Values and Optimal Strategy}\label{proc:set_values_and_strategies}
\scriptsize
\DontPrintSemicolon
\nonl \SetKwProg{Fn}{Procedure}{}{}
\Fn{$\texttt{set\_vars}(\W_0, \W_1, \nu, \sigma^*_0, i, j, F, J, \Gamma)$}{
    \SetKwInOut{Input}{input}
    \SetKwInOut{Output}{output}
\Input{Winning sets $\W_0, \W_1$, values $\nu$, opt. strategy $\sigma^*_0$, 
$i\in [W^-, W^+]$, $j\in [1, s-1]$, ref. $F$ to $\F_{|V|}$, \MPG $\Gamma$}
	$D\leftarrow$ denominator of $F[j-1]$; \label{proc:set:l0} \;
	\While{$L_{\top}\neq\emptyset$ \label{proc:set:l1}}{
		$u\leftarrow\texttt{pop\_front}(L_{\top})$; \label{proc:set:l2}\;
		$\nu(u)\leftarrow i + F[j-1]$; \label{proc:set:l3}\; 
		\lIf{$\nu(u)\geq 0$ }{ 
			$\W_0\leftarrow \W_0\cup \{v\}$; \texttt{\bf else} $\W_1\leftarrow \W_1\cup \{v\}$; \label{proc:set:l4}
		}
		\If{$u\in V_0$ \label{proc:set:l5}}{ 
			\For{$v\in N^{\text{out}}_{\Gamma}(u)$ \label{proc:set:l6}}{ 	
				\If{$(D\cdot\J.f[u])\succeq (D\cdot\J.f[v]) \ominus 
					\texttt{get\_scl\_}w\big(w(u, v), i, j-1, F\big)$\label{proc:set:l7} }{ 
						$\sigma^*_0(u)\leftarrow v$; \textbf{break}; \label{proc:set:l8} }}}
	}
}
\end{algo-proc}

This actually concludes the description of $\texttt{set\_vars}()$ (SubProcedure~\ref{proc:set_values_and_strategies}).

Indeed, the role of $L_f$ is precisely that to allow the $\texttt{J-VI}()$ 
to lift-up the energy-levels during the $(i,j)$-th Scan-Phase,  
meanwhile preserving (inside $\J.f$) those computed at the $(i,j-1)$-th one (because $\texttt{set\_vars}()$ needs them in order to rely on Theorem~\ref{Thm:pos_opt_strategy}).
As mentioned, when $\texttt{set\_vars}()$ halts, all the energy-levels are scaled back, 
from $\N$ to $\Q$, and stored back from $L_f$ into $\J.f$ 
(at line~\ref{algo:solve:l13} of Algorithm~\ref{algo:solve_mpg}, 
see $\texttt{scl\_back\_}f()$ in SubProcedure~\ref{subproc:scaling_energy}).

We remark at this point that all the arithmetics of Algorithm~\ref{algo:solve_mpg} can be done in $\Z$. 

Now, let us detail the remaining subprocedures, 
those governing the Jumps and those concerning the energy-levels and the $\texttt{J-VI}()$.
Since the details of the former rely significantly on those of the latter two, we proceed by discussing firstly 
how the energy-levels are handled by the 
$\texttt{J-VI}()$ (see SubProcedure~\ref{algo:j-value-iteration} and~\ref{subproc:scaling_energy}).

\textit{J-Value-Iteration.} $\texttt{J-VI}()$ is similar to the Value-Iteration of \cite{brim2011faster}. 
Still, there are some distinctive features. 
The $\texttt{J-VI}()$ takes in input two indices $i\in [W^{-}, W^{+}]$ and $j\in [1, s-1]$, 
a reference $F$ to $\F_{|V|}$, (a reference to) the Jumper $\J$, (a reference to) the input arena $\Gamma$. 
Basically, $\texttt{J-VI}(i,j,F,\J,\Gamma)$ aims at computing the least-SEPM of the reweighted \EG $\Gamma_{i,j}$. 
For this, it relies on a (slightly revisited) \emph{energy-lifting} operator $\delta: 
[V\rightarrow \C_{\Gamma}]\times V\rightarrow [V\rightarrow \C_\Gamma]$. 
The array-list employed to keep track of the inconsistent vertices is $L^{\text{inc}}$. 
It is assumed, as a pre-condition, that $L^{\text{inc}}$ is already initialized when $\texttt{J-VI}()$ starts.
We will show that this pre-condition holds thanks to how $L^{\text{inc}}_{\text{nxt}}$ is managed. 
Recall, Algorithm~\ref{algo:solve_mpg} is going to perform a sequence of invocations to $\texttt{J-VI}()$. 
During the execution of any such invocation of $\texttt{J-VI}()$, 
the role of $L^{\text{inc}}_{\text{nxt}}$ is precisely that of collecting, 
in advance, the initial list of inconsistent vertices for the 
\emph{next\footnote{\ie the subsequent invocation (in the above mentioned sequence of $\texttt{J-VI}()$)
that will be performed, either at line~\ref{algo:solve:l12} of $\texttt{solve\_MPG}()$ (Algorithm~\ref{algo:solve_mpg}), 
or at line~\ref{subproc:ua-jumps:l3} of $\texttt{ua-jumps}()$ (SubProcedure~\ref{subproc:ua-jumps}).}} $\texttt{J-VI}()$. 
Rephrasing, the $k$-th invocation of $\texttt{J-VI}()$ takes care of 
initializing $L^{\text{inc}}$ for the $k+1$-th invocation of $\texttt{J-VI}()$, 
and this is done thanks to $L^{\text{inc}}_{\text{nxt}}$. 

Also, the energy-levels are managed in a special way. The \emph{inital} energy-levels are stored inside $\J.f$ (as a pre-condition).
Again, the $k$-th invocation of $\texttt{J-VI}()$ takes care of initializing the initial energy-levels for the $k+1$-th one: 
actually, those computed at the end of the $k$-th $\texttt{J-VI}()$ will become the initial energy-levels 
for the $k+1$-th one (subject to a rescaling). In this way, 
Algorithm~\ref{algo:solve_mpg} will succeed at amortizing the cost of all invocations of $\texttt{J-VI}()$. 
As mentioned, since $\J.f$ stores rational-scalings, and $\Gamma_{i,j}$ is weighted in $\Z$, 
the $\texttt{J-VI}()$ needs to scale everything up, from $\Q$ to $\N$, when it reads the energy-levels out from $\J.f$. 
So, $\J.f$ is accessed \emph{read-only} during the $\texttt{J-VI}()$: we want to update the energy-levels by applying $\delta$, 
but still we need a back-up copy of the initial energy-levels 
(because they are needed at line~\ref{proc:set:l7} of $\texttt{set\_vars}()$, SubProcedure~\ref{proc:set_values_and_strategies}). 
Therefore, a special subprocedure is employed for accessing energy-levels during $\texttt{J-VI}()$, 
it is named $\texttt{get\_scl\_}f()$ (SubProcedure~\ref{subproc:scaling_energy}); moreover, an array-list $L_f$ is employed, 
whose aim is that to store the current energy-levels, those lifted-up during the $\texttt{J-VI}()$. 
SubProcedure~\ref{subproc:scaling_energy} shows $\texttt{get\_scl\_}f()$, it takes: 
$u\in V$, some $j\in [1, s-1]$, a reference $F$ to $\F_{|V|}$, and (a reference to) $\J$. 

$\texttt{get\_scl\_}f()$ goes as follows. If $L_{f}[u]=\bot$ (line~\ref{subproc:get_scaled_energy:l1}), 
the denominator $D$ of $F_j$ is taken (line~\ref{subproc:get_scaled_energy:l2}), 
and $f\leftarrow \left\lceil D\cdot\J.f[u]\right\rceil$ is computed (line~\ref{subproc:get_scaled_energy:l3}); 
a (new) entry $(v,f)$ is inserted into $L_f$ (line~\ref{subproc:get_scaled_energy:l4}). 
Finally, in any case, $L_f[v]$ is returned (line~\ref{subproc:get_scaled_energy:l5}). 

\begin{wrapfigure}[13]{r}{6.5cm}
\raisebox{0pt}[\dimexpr\height-2\baselineskip\relax]{
\begin{algo-proc}[H]
\caption{Energy-Levels}\label{subproc:scaling_energy}
\scriptsize
\DontPrintSemicolon
\nonl \SetKwProg{Fn}{SubProcedure}{}{}
\Fn{$\texttt{get\_scl\_}f(v,j,F,\J)$}{
\SetKwInOut{Input}{input}
\Input{$v\in V$, $j\in [1, s-1]$, \\$F$ is a ref. to $\F_{|V|}$, $\J$ is Jumper.}
\If{$L_f[v]=\bot$}{ \label{subproc:get_scaled_energy:l1}
	$D\leftarrow$ denominator of $F[j]$; \label{subproc:get_scaled_energy:l2} \; 
	$f\leftarrow \left\lceil D\cdot\J.f[v]\right\rceil$; \label{subproc:get_scaled_energy:l3} \;
	$\texttt{insert}\big((v,f), L_f\big)$; \label{subproc:get_scaled_energy:l4} \;
}
\Return{$L_f[v]$;} \label{subproc:get_scaled_energy:l5}
}
\setcounter{AlgoLine}{0}
\nonl \Fn{$\texttt{scl\_back\_}f(j,F,\J)$}{
\SetKwInOut{Input}{input}
\Input{$j\in [0, s-1]$, $F$ is a ref. to Farey's terms, \\ $\J$ is Jumper.}
$D\leftarrow$ denominator of $F[j]$; \label{subproc:scale_back:l1}\; 
\While{$L_f\neq\emptyset$}{ \label{subproc:scale_back:l2}
	$(v,f)\leftarrow \texttt{pop\_front}(L_f)$; \label{subproc:scale_back:l3}\;
	$\J.f[v]\leftarrow f/D $; \label{subproc:scale_back:l4}\;
}
}
\end{algo-proc}}
\end{wrapfigure}

As mentioned, at line~\ref{algo:solve:l13} of Algorithm~\ref{algo:solve_mpg}, 
$\J.f$ will be overwritten by scaling back the values that are stored in $L_f$. 
This is done by $\texttt{scl\_back\_}f()$ (SubProcedure~\ref{subproc:scaling_energy}): 
at line~\ref{subproc:scale_back:l1}, $D$ is the denominator of $F_j$; 
then, $L_f$ is emptied, one element at a time (line~\ref{subproc:scale_back:l2}); 
for each $(v,f)\in L_f$ (line~\ref{subproc:scale_back:l3}), 
the rational $f/D$ is stored back to $\J.f[v]$ (line~\ref{subproc:scale_back:l4}). 
This concludes $\texttt{scl\_back\_}f()$.

Next, $\texttt{J-VI}()$ takes in input: $i\in [W^-, W^+]$, $j\in [1, s-1]$, a reference $F$ to $\F_{|V|}$, 
(a reference to) the Jumper $\J$, and (a reference to) the input \MPG $\Gamma$.
At line~\ref{algo:jvalue:l1}, $\texttt{J-VI}()$ enters into a \texttt{while} loop which lasts while $L^{\text{inc}}\neq\emptyset$. 
The front vertex $v\leftarrow\texttt{pop\_front}(L^{\text{inc}})$ is popped from $L^{\text{inc}}$ (line~\ref{algo:jvalue:l2}).
Next, the energy-lifting operator $\delta$ is applied to $v$ 
by invoking $\texttt{apply\_}\delta(v,i,j,F,\J,\Gamma)$ (line~\ref{algo:jvalue:l3}).

There inside (at line~\ref{subproc:energy-lift:l1} of $\texttt{apply\_}\delta()$), the energy-level of $v$ is lifted-up as follows:
\[
f_v\leftarrow \left\{
	\begin{array}{ll}
		\min\big\{\texttt{get\_scl\_}f\big(v', j, F, \J\big) \ominus 
	\texttt{get\_scl\_}w\big(w(v,v'),i,j,F\big) \mid v'\in N^{\text{out}}_{\Gamma}(v)\big\}, & \text{ if } v\in V_0; \\ 
		\max\big\{\texttt{get\_scl\_}f\big(v', j, F, \J\big) \ominus 
	\texttt{get\_scl\_}w\big(w(v,v'),i,j,F\big) \mid v'\in N^{\text{out}}_{\Gamma}(v)\big\}, & \text{ if } v\in V_1.
	\end{array} \right.
\]

Then, $f_v$ is stored inside $L_f$ (notice, not in $\J.f$), where it is binded to the key $v$  (line~\ref{subproc:energy-lift:l2}). 
The control turns back to $\texttt{J-VI}()$. 
The current energy-level of $v$ is retrieved by $f_v\leftarrow \texttt{get\_scl\_}f(v,j,F,\J)$ (line~\ref{algo:jvalue:l4}).
If $f_v\neq \top$ (line~\ref{algo:jvalue:l5}), 
then $v$ is inserted into $L^{\text{inc}}_{\text{nxt}}$ (if it isn't already in there) (line~\ref{algo:jvalue:l6}); 
moreover, if $v\in V_0$, then $\J.\texttt{cnt}[v]$ and 
$\{\J.\texttt{cmp}[(v,v')]\mid v'\in N_\Gamma^{\text{out}}(v)\}$ are recalculated from scratch, 
by invoking $\texttt{init\_cnt\_cmp}(v,i,j,F,\J,\Gamma)$ (line~\ref{algo:jvalue:l7}, see SubProcedure~\ref{subproc:counters_inc-arcs}). 
Else, if $f_v=\top$ (line~\ref{algo:jvalue:l8}), then $v$ is stored into $L_{\top}$ (line~\ref{algo:jvalue:l9}); 
and if $L^{\text{inc}}_{\text{nxt}}[v]\neq\bot$ in addition, 
then $v$ is removed from $L^{\text{inc}}_{\text{nxt}}$ (line~\ref{algo:jvalue:l10}).  

At this point it is worth introducing the following notation concerning energy-levels.
\begin{Def}
For any step of execution $\iota$ and for any variable $x$ of Algorithm~\ref{algo:solve_mpg}, 
the state of $x$ at step $\iota$ is denoted by $x^\iota$. 
Then, the \emph{current energy-levels} at step $\iota$ are defined as follows: 
\[
\forall^{v\in V} f^{\texttt{c}:\iota}(v)\triangleq
\left\{\begin{array}{ll} 
	L^{\iota}_f[v], & \text{ if } L^{\iota}_f[v]\neq\bot; \\ 
	\big\lceil D_{j^\iota}\cdot \J.f^\iota[v]\big\rceil, & \text{ otherwise}.
\end{array}\right.  
\]
where $D_{j^\iota}$ is the denominator of $F_{j^\iota}$. If $\iota$ is implicit, 
the \emph{current energy-levels} are denoted by $f^{\texttt{c}}$. 
\end{Def}

\begin{algo-proc}[H]
\caption{J-Value-Iteration}\label{algo:j-value-iteration}
\scriptsize
\DontPrintSemicolon
\nonl \SetKwProg{Fn}{Procedure}{}{}
\Fn{$\texttt{J-VI}(i,j,F,\J,\Gamma)$}{
    \SetKwInOut{Input}{input}
    \SetKwInOut{Output}{output}
\Input{$i\in [W^-, W^+]$ and $j\in [1, s-1]$, $F$ is a ref. to Farey's terms, $\J$ is Jumper, $\Gamma$ is an \MPG.}
\While{$L^{\text{inc}}\neq\emptyset$}{ \label{algo:jvalue:l1}
	$v\leftarrow \texttt{pop\_front}(L^{\text{inc}})$; \label{algo:jvalue:l2} \;
		$\texttt{apply\_}\delta(v,i,j,F,\J,\Gamma)$; \label{algo:jvalue:l3} \;	
		$f_v\leftarrow \texttt{get\_scl\_}f(v,j,F,\J)$; \label{algo:jvalue:l4} \;	
		\If{$f_v\neq\top$}{  \label{algo:jvalue:l5}
			\lIf{$L^{\text{inc}}_{\text{nxt}}[v]=\bot$}{ 
				$\texttt{insert}(v,L^{\text{inc}}_{\text{nxt}})$; \label{algo:jvalue:l6}
			}
			\lIf{$v\in V_0$}{$\texttt{init\_cnt\_cmp}(v,i,j,F,\J,\Gamma)$; \label{algo:jvalue:l7}}
		}\Else{ \label{algo:jvalue:l8}
			$\texttt{insert}(v, L_{\top})$; \label{algo:jvalue:l9}\;
			\lIf{$L^{\text{inc}}_{\text{nxt}}[v]\neq\bot$}{
				$\texttt{remove}(v,L^{\text{inc}}_{\text{nxt}});$ \label{algo:jvalue:l10} 
			}
		}
		\ForEach{$u\in N^{\text{in}}_{\Gamma}(v)$}{ \label{algo:jvalue:l11}
			$f_u\leftarrow \texttt{get\_scl\_}f(u,j,F,\J)$; \label{algo:jvalue:l12}\;
			$\Delta_{u,v}\leftarrow f_v\ominus \texttt{get\_scl\_}w(w(u,v),i,j,F)$; \label{algo:jvalue:l13}\;
			\If{ $L^{\text{inc}}[u]=\bot$ \texttt{\bf\, and\,} $f_u < \Delta_{u,v}$ }{ \label{algo:jvalue:l14}
				\If{$u\in V_0$ \texttt{\bf and} $\J.\texttt{cmp}[(u,v)]=\texttt{T}$}{ \label{algo:jvalue:l15}
					$\J.\texttt{cnt}[u]\leftarrow \J.\texttt{cnt}[u]-1$; \label{algo:jvalue:l16}\;
					$\J.\texttt{cmp}[(u,v)]\leftarrow\texttt{F}$; \label{algo:jvalue:l17}\; 
				}
				\lIf{$u\in V_1$ \texttt{\bf OR} 	
				$J.\texttt{cnt}[u]=0$}{$\texttt{insert}(u,L^{\text{inc}})$; \label{algo:jvalue:l18}}	
			}
		}
}
$\texttt{swap}(L^{\text{inc}}, L^{\text{inc}}_{\text{nxt}})$; \label{algo:jvalue:l19}
}
\setcounter{AlgoLine}{0}
\nonl \SetKwProg{SubFn}{SubProcedure}{}{}
\SubFn{$\texttt{apply\_}\delta(v,i,j,F,\J,\Gamma)$}{
    \SetKwInOut{Input}{input}
    \SetKwInOut{Output}{output}
\Input{$v\in V$, $i\in [W^-, W^+]$, $j\in [1, s-1]$, 
	$F$ is a ref. to Farey, $\J$ is Jumper, $\Gamma$ is an \MPG.}
$f_v\leftarrow \left\{
\begin{array}{ll}
	\min\big\{\texttt{get\_scl\_}f\big(v', j, F, \J\big) \ominus \texttt{get\_scl\_}w\big(w(v,v'),i,j,F\big) 
							\mid v'\in N^{\text{out}}_{\Gamma}(v)\big\}, & \text{ if } v\in V_0; \\
	\max\big\{\texttt{get\_scl\_}f\big(v', j, F, \J\big) \ominus \texttt{get\_scl\_}w\big(w(v,v'),i,j,F\big) 
							\mid v'\in N^{\text{out}}_{\Gamma}(v)\big\}, & \text{ if } v\in V_1. 
\end{array}
\right.$ \label{subproc:energy-lift:l1}\;
$\texttt{insert}\big((v, f_v), L_f\big)$; \label{subproc:energy-lift:l2}\;
}
\end{algo-proc}

\begin{Rem} 
Recall, the role of $L^{\text{inc}}_{\text{nxt}}$ and that of the $\texttt{swap}()$ (line~\ref{algo:jvalue:l19}) 
is precisely that of initializing, 
in advance, the list of inconsistent vertices $L^{\text{inc}}$ for the \emph{next} $\texttt{J-VI}()$; 
because the $\texttt{J-VI}()$ assumes a correct initialization of $L^{\text{inc}}$ as a pre-condition. 

We argue in Proposition~\ref{prop:jvalue_init} and Lemma~\ref{lem:jvalue_init_correct} that, 
when $\texttt{J-VI}()$ halts --say at step $h$-- it is necessary to 
initialize $\J.{L^{\text{inc}}}$ for the \emph{next} $\texttt{J-VI}()$ by including (at least) 
all the $v\in V$ such that: $0<f^{\text{c}:h}(v)\neq\top$.  
\end{Rem}
Notice, if $L^{\text{inc}}_{\text{nxt}}=\emptyset$ holds just before the $\texttt{swap}()$ at line~\ref{algo:jvalue:l19}, 
then $L^{\text{inc}}=\emptyset$ holds soon after;
therefore, in that case yet another \emph{EI-Jump} will occur (at line~\ref{algo:solve:l7} of Algorithm~\ref{algo:solve_mpg}) 
and eventually some other vertices will be inserted into 
$L^{\text{inc}}$ (see the details of SubProcedure~\ref{subproc:energy_jumps}).

We shall provide the details of $\texttt{init\_cnt\_inc}(v,i,j,F,\J)$ (line~\ref{algo:jvalue:l7}) very soon hereafter.

But let us first discuss the role that is played by $\J.\texttt{cnt}$ and $\J.\texttt{cmp}$ during the $\texttt{J-VI}()$.

\begin{wrapfigure}[11]{r}{7.5cm}
\raisebox{0pt}[\dimexpr\height-1.3\baselineskip\relax]{
\begin{algo-proc}[H]
\caption{Counters and Cmp Flags}\label{subproc:counters_inc-arcs}
\scriptsize
\DontPrintSemicolon
\nonl \SetKwProg{SubFn}{SubProcedure}{}{}
\SubFn{$\texttt{init\_cnt\_cmp}(u,i,j,F,\J,\Gamma)$}{
    \SetKwInOut{Input}{input}
    \SetKwInOut{Output}{output}
\Input{$u\in V_0$, $i\in [W^-, W^+]$, $j\in [1, s-1]$, $F$ is a ref. to Farey, $\J$ is Jumper, $\Gamma$ is an \MPG.}
	$c_u\leftarrow 0$; \label{subproc:cnt_inc-arcs:l1} \;
	\ForEach{$v\in N^{\text{out}}_{\Gamma}(u)$}{ \label{subproc:cnt_inc-arcs:l2} 
		$f_u\leftarrow \texttt{get\_scl\_}f(u,j,F,\J)$; \label{subproc:cnt_inc-arcs:l3} \; 
		$f_v\leftarrow\texttt{get\_scl\_}f(v,j,F,\J)$; \label{subproc:cnt_inc-arcs:l4} \;
		\If{$f_u \succeq f_{v} \ominus \texttt{get\_scl\_}w\big(w(u,v),i,j,F\big)$ \label{subproc:cnt_inc-arcs:l5} }{ 
			$c_u\leftarrow c_u+1$; \label{subproc:cnt_inc-arcs:l6} \; 
			$\J.\texttt{cmp}[(u,v)]\leftarrow \texttt{T}$; \label{subproc:cnt_inc-arcs:l7}
		}\lElse{
			$\J.\texttt{cmp}[(u,v)]\leftarrow \texttt{F}$; \label{subproc:cnt_inc-arcs:l8}
		}
	}
	$\J.\texttt{cnt}[u]\leftarrow c_u$; \label{subproc:cnt_inc-arcs:l9}\;
}
\end{algo-proc}}
\end{wrapfigure}

From line~\ref{algo:jvalue:l11} to line~\ref{algo:jvalue:l18}, 
$\texttt{J-VI}()$ explores $N^{\text{in}}_{\Gamma}(v)$ in order to find all 
the $u\in N^{\text{in}}_{\Gamma}(v)$ that may have become inconsistent 
soon after the energy-lifting $\delta$ that was applied to $v$ (before, at line~\ref{algo:jvalue:l3}). 
For each $u\in N^{\text{in}}_{\Gamma}(v)$ (line~\ref{algo:jvalue:l11}), 
the energy-level $f_u\leftarrow \texttt{get\_scl\_}f(u,j,F,\J)$ is considered (line~\ref{algo:jvalue:l12}), 
also, $\Delta_{u,v}\leftarrow f_v\ominus w'_{i,j}(u,v)$ is computed (line~\ref{algo:jvalue:l13}), 
where $f_v\leftarrow \texttt{get\_scl\_}f(v,j,F,\J)$; 
if $f_u<\Delta_{u,v}$ (\ie in case $(u,v)$ is now incompatible 
\wrt $f^{\text{c}}$ in $\Gamma_{i,j}$) \emph{and} $L^{\text{inc}}[u]=\bot$ holds (line~\ref{algo:jvalue:l14}), then: 

-- If $u\in V_0$ \textit{and} $(u,v)$ was not already incompatible 
\emph{before} (\ie if $\J.\texttt{cmp}[(u,v)]=\texttt{T}$ at line~\ref{algo:jvalue:l15}, 
then: $\J.\texttt{cnt}[u]$ is decremented (line~\ref{algo:jvalue:l16}), 
and $\J.\texttt{cmp}[(u,v)]\leftarrow \texttt{F}$ is assigned (line~\ref{algo:jvalue:l17}). 
(This is the role of the $\J.\texttt{cnt}$ and $\J.\texttt{cmp}$ flags).

-- After that, if $u\in V_1$ \textit{or} $\J.\texttt{cnt}[u]=0$, 
then $u$ is inserted into $L^{\text{inc}}$ (line~\ref{algo:jvalue:l18}). 

When the \texttt{while} loop (at line~\ref{algo:jvalue:l1}) ends, 
the (references to) $L^{\text{inc}}$ and $L^{\text{inc}}_{\text{nxt}}$ 
are \emph{swapped} (line~\ref{algo:jvalue:l19}) (one is assigned to reference the other and vice-versa, 
in $O(1)$ time by interchanging pointers). 

The details of $\texttt{init\_cnt\_cmp}(u,i,j,F,\J,\Gamma)$ (line~\ref{algo:jvalue:l7}), 
where $u\in V_0$, are given in SubProcedure~\ref{subproc:counters_inc-arcs}. 
At line~\ref{subproc:cnt_inc-arcs:l1}, $c_u\leftarrow 0$ is initialized. 
For each $v\in N^{\text{out}}_{\Gamma}(u)$ (line~\ref{subproc:cnt_inc-arcs:l2}), 
it is checked whether $(u,v)$ is compatible with respect to the current energy-levels; 
\ie whether or not $f_u \succeq f_{v} \ominus w'_{i,j}(u,v)$, 
holds for $f_u\leftarrow\texttt{get\_scl\_}f(u,j,F,\J)=f^{\text{c}}(u)$ and 
$f_{v}\leftarrow\texttt{get\_scl\_}f(v,j,F,\J)=f^{\text{c}}(v)$ (lines~\ref{subproc:cnt_inc-arcs:l3}-\ref{subproc:cnt_inc-arcs:l5}); 
if $(u,v)$ is found to be compatible, then $c_u$ is 
incremented (line~\ref{subproc:cnt_inc-arcs:l6}) and $\J.\texttt{cmp}[(u,v)]\leftarrow \texttt{T}$ 
is assigned (line~\ref{subproc:cnt_inc-arcs:l7}); 
otherwise, ($c_u$ stands still and) it is set $\J.\texttt{cmp}[(u,v)]\leftarrow \texttt{F}$ (line~\ref{subproc:cnt_inc-arcs:l8}). 
At the very end, it is finally set $\J.\texttt{cnt}[u]\leftarrow c_u$ (line~\ref{subproc:cnt_inc-arcs:l9}).

Concerning $\J.\texttt{cmp}$ and $\J.\texttt{cnt}$, it is now worth defining a formal notion of \emph{coherency}.
\begin{Def}\label{def:coherency}
Let $\iota$ be any step of execution of Algorithm~\ref{algo:solve_mpg}. 
Let $i\in [W^-,W^+]$, $j\in [0,s-1]$, $u\in V_0$ and $v\in N_\Gamma^{\text{out}}(u)$. 
We say that $\J.\texttt{cmp}^{\iota}[(u,v)]$ is \emph{coherent} \wrt $f^{\text{c}:\iota}$ in $\Gamma_{i,j}$ when it holds:
\[
\J.\texttt{cmp}^{\iota}[(u,v)]=\texttt{T} \;\textit{ iff }\;
					f^{\text{c}:\iota}(u)\succeq f^{\text{c}:\iota}(v)\ominus w'_{i,j}(u,v).
\]
Also, we say that $\J.\texttt{cnt}^{\iota}[u]$ is \emph{coherent} \wrt $f^{\text{c}:\iota}$ in $\Gamma_{i,j}$ when:
\[
\J.\texttt{cnt}^{\iota}[u] = \big|\big\{(u,v)\in E\mid f^{\text{c}:\iota}(u) 
					\succeq f^{\text{c}:\iota}(v)\ominus w'_{i,j}(u,v)\big\}\big|.
\]
We say that $\J.\texttt{cmp}^\iota$ is \emph{coherent} when $\forall\, (u\in V_0\setminus {L^{\text{inc}}}^{\iota})\; 
	\forall\, (v\in N_\Gamma^{\text{out}}(u))$ $\J.\texttt{cmp}^{\iota}[(u,v)]$ \text{is coherent};

and we say that $\J.\texttt{cnt}^{\iota}$ is \emph{coherent} 
when $\forall (u\in V_0\setminus {L^{\text{inc}}}^{\iota})$ $\J.\texttt{cnt}^{\iota}[u]$ is coherent. 

Finally, when something is not coherent, it is \emph{incoherent}. Remark: the step $\iota$ can be implicit.
\end{Def}
\begin{Rem}
In the Value-Iteration~\citep{brim2011faster}, the consistency checking of $(u,v)\in E$ (line~\ref{algo:jvalue:l14}) is 
explicit: an inequality like ``$f(u)\succeq f(v)\ominus w(u,v)$" is tested; 
thus, neither the \texttt{cmp} flags nor an explicit notion of coherency are needed. 
So, why we introduced \texttt{cmp} flags and coherency? Observe, at line~\ref{algo:jvalue:l14} of $\texttt{J-VI}()$, 
it doesn't make much sense to check ``$f(u)\succeq f(v)\ominus w(u,v)$" in our setting. 
Consider the following facts: (1) of course the values of $w'_{i,j}$ depend on the index $(i,j)$ of the current Scan-Phase; 
(2) therefore, going from one Scan-Phase to the next one (possibly, by Jumping), 
some counters may become incoherent, because $w_{i',j'} < w_{i,j}$ if $(i',j') > (i,j)$; 
but in the Value-Iteration~\citep{brim2011faster} the only possible source of incoherency was the application of $\delta(\cdot, v)$; 
in Algorithm~\ref{algo:solve_mpg}, going from one Scan-Phase to the next, we have an additional source of incoherency. 
(3) still, $\texttt{J-VI}()$ can't afford to re-initialize $\texttt{cnt}: V\rightarrow\N$ each time that it is 
needed, as this would cost $\Omega(|E|)$. So, if $(u,v)\in E$ is found incompatible (at line~\ref{algo:jvalue:l14} of $\texttt{J-VI}()$) 
after the application of $\delta(\cdot, v)$ (line~\ref{algo:jvalue:l3}), 
how do we know whether or not $(u,v)$ was already incompatible \emph{before} the (last) application of $\delta(\cdot, v)$? 
We suggest to adopt the \texttt{cmp} flags, one bit per arc is enough.
\end{Rem}

To show correctness and complexity, we firstly assume that whenever $\texttt{J-VI}(i,j,F,\J,\Gamma)$ is invoked
the following three pre-conditions are satisfied:

\begin{itemize}
\item[\emph{(PC-1)}] $L_f=\emptyset$ and $\forall^{v\in V}\, f^{\texttt{c}}(v) \preceq f^*_{w'_{i,j}}(v)$;

\item[\emph{(PC-2)}] $L^{\text{inc}}= \text{Inc}(f^{\text{c}}, i, j)$;

\item[\emph{(PC-3)}] $\J.\texttt{cnt}$ and $\J.\texttt{cmp}$ are \textit{coherent} \wrt $f^{\texttt{c}}$ in $\Gamma_{i,j}$.
\end{itemize}

After having described the internals of the EI-Jumps, 
we'll show how to ensure (a slightly weaker, but still sufficient formulation of) (PC-1), (PC-2), (PC-3).

Assuming the pre-conditions, similar arguments as in [\cite{brim2011faster}, Theorem~4] show that $\texttt{J-VI}()$ 
computes the least-SEPM of the \EG $\Gamma_{i,j}$ in time $O(|V|^2|E|W)$ and linear space.
\begin{Prop}\label{prop:correctness}
Assume that $\texttt{J-VI}()$ is invoked on input $(i,j,F,\J,\Gamma)$, and that (PC-1), (PC-2), (PC-3) hold at invocation time. 
Then, $\texttt{J-VI}()$ halts within the following time bound:
\[
	\Theta\Big(\sum_{v\in V} \texttt{deg}_\Gamma(v)\cdot\ell^1_{\Gamma_{i,j}}(v)\Big)= O\big(|V|^2 |E| W\big), 
\]
where $0\leq \ell^1_{\Gamma_{i,j}}(v)\leq (|V|-1)|V| W$ is the number of times 
that the energy-lifting operator $\delta$ is applied to any $v\in V$, 
at line~\ref{algo:jvalue:l3} of $\texttt{J-VI}()$ on input $(i,j,F,\J,\Gamma)$. The working space is $\Theta(|V|+|E|)$. 

When $\texttt{J-VI}()$ halts, $f^{\texttt{c}}$ coincides with the \emph{least}-SEPM of the reweighted \EG $\Gamma_{i,j}$.
\end{Prop}
\begin{proof} 
The argument is very similar to that of [\cite{brim2011faster}, Theorem~4], 
but there are some subtle differences between the $\texttt{J-VI}()$ and the Value-Iteration of Brim, \etal:

(1) $\texttt{J-VI}()$ employs $\J.f$ and $L_f$ to manage the energy-levels; however, 
one can safely argue by always referring to the current energy-levels $f^{\texttt{c}}$. 

(2) $\texttt{J-VI}()$ has no initialization phase; however, 
notice that the pre-conditions (PC-1), (PC-2), (PC-3) ensure a correct initialization of it. 

(3) $\texttt{J-VI}()$ employs $\J.\texttt{cmp}$ in order to test 
the consistency state of the arcs (see line~\ref{algo:jvalue:l15} and \ref{algo:jvalue:l17} of $\texttt{J-VI}()$); 
but it is easy to see that, assuming (PC-3), this is a correct way to go.

Let us provide a sketch of the proof of correctness. As already observed in [\cite{brim2011faster}, Lemma~7], 
the energy-lifting operator $\delta$ is 
$\sqsubseteq$-\emph{monotone} (\ie $\delta(f,v) \sqsubseteq \delta(g,v)$ for all $f \sqsubseteq g$).  
Next, the following invariant is maintained by $\texttt{J-VI}()$ 
	(Subprocedure~\ref{algo:j-value-iteration}) at line~\ref{algo:jvalue:l1}. 

\emph{Inv-JVI.} $\forall(\text{iteration } \iota 
\text{ of line}~\ref{algo:jvalue:l1}$ of $\texttt{J-VI}(i,j,F,\J,\Gamma))$ $\forall (u\in V\setminus 
	\J.{L^{\text{inc}}}^\iota)$ $\forall (v\in N_\Gamma^{\text{out}}(u))$:  

	(\textit{i}) $\delta(f^{\texttt{c}:\iota},u)=f^{\texttt{c}:\iota}$; 

	(\textit{ii}) if $u\in V_0\setminus \J.{L^{\text{inc}}}^\iota$, 
then $\J.\texttt{cnt}^{\iota}[u]$ and $\J.\texttt{cmp}^{\iota}[(u,v)]$ are 
both coherent \wrt $f^{\text{c}:\iota}$ in $\Gamma_{i,j}$. 

It is not difficult to prove that \emph{Inv-JVI} holds. 
The argument is almost the same as in [\cite{brim2011faster}, Lemma~8]; 
the only noticeable variations are: 
(a) the $\texttt{J-VI}()$ employs $\J.\texttt{cmp}$ in order to flag the compatibility status of the arcs; 
(b) the reference energy-level is $f^{\texttt{c}}$; 
(c) at the first iteration of line~\ref{algo:jvalue:l1} of $\texttt{J-VI}()$, 
the \emph{Inv-JVI} holds thanks to (PC-2) and (PC-3).

Termination is enforced by three facts: (i) every application of the energy-lifting operator (line~\ref{algo:jvalue:l3}) 
strictly increases the energy-level of one vertex $v$; (ii) the co-domain of SEPMs is finite. 

Correctness follows by applying the Knaster-Tarski's Fixed-Point Theorem.
Indeed, at halting time, since $\delta$ is $\sqsubseteq$-monotone, 
and since (PC-1) and \emph{Inv-JVI} hold, then we can apply Knaster-Tarski's Fixed-Point Theorem to conclude that, 
when $\texttt{J-VI}()$ halts at step $h$ (say), 
then $f^{\texttt{c}:h}$ is the unique least fixpoint of (simultaneously) all operators $\delta(\cdot,v)$ 
for all $v \in V$, \ie $f^{\texttt{c}:h}$ is the least-SEPM of the \EG $\Gamma_{i,j}$. 

So, when $\texttt{J-VI}()$ halts, it holds that $\forall^{ v\in V}\, f^{\texttt{c}:h}(v)=f^*_{w'_{i,j}}(v)$.

Concerning the time and space complexity, $\delta(\cdot,v)$ can 
be computed in time $O(|N_\Gamma^{\text{out}}(v)|)$ (line~\ref{algo:jvalue:l3}) 
(see $\texttt{apply\_}\delta()$ in SubProcedure~\ref{algo:j-value-iteration}); 
the updating of $\J.\texttt{cnt}$ and $\J.\texttt{cmp}$, 
which is performed by $\texttt{init\_cnt\_cmp}()$ (line~\ref{algo:jvalue:l7}), also takes $O(|N_\Gamma^{\text{out}}(v)|)$ time. 
Soon after that $\delta(\cdot, v)$ has been applied to $v\in V$ (line~\ref{algo:jvalue:l3}), 
the whole $N^{\text{in}}_\Gamma(v)$ is explored for repairing incoherencies and for finding new inconsistent vertices 
(which is done from line~\ref{algo:jvalue:l11} to line~\ref{algo:jvalue:l18}): this process takes $O(|N_\Gamma^{\text{in}}(v)|)$ time.
Therefore, if $\delta(\cdot, v)$ is applied $\ell^1_{\Gamma_{i,j}}(v)$ 
times to (any) $v\in V$ during the $\texttt{J-VI}(i,j,F,\J,\Gamma)$, 
the total time is $\Theta\big(\sum_{v\in V} \texttt{deg}_\Gamma(v)\cdot\ell^1_{\Gamma_{i,j}}(v)\big)$. 
The codomain of any SEPM of $\Gamma_{i,j}$ is at most $(|V|-1)W'$, for $W'=D_jW\leq |V|W$, 
where the additional factor $D_j\leq |V|$ comes from the scaled weights of $\Gamma_{i,j}$; 
thus, $\forall^{v\in V} 0\leq \ell^1_{\Gamma_{i,j}}(v)\leq (|V|-1)D_j W\leq (|V|-1)|V|W$. 
As already mentioned in Section~\ref{section:values}, 
the Farey's term $F[j]$ can be computed at the beginning of $\texttt{J-VI}()$ in $O(1)$ time and space, 
from $F[j-1]$ and $F[j-2]$.
Since $\sum_{v\in V}\texttt{deg}_\Gamma(v)=2|E|$, the running time is also $O(|V|^2|E|W)$. 
We check that $\texttt{J-VI}()$ works with $\Theta(|V|+|E|)$ space: $L^{\text{inc}}$, $L^{\text{inc}}_{\text{nxt}}$, $L_f$, 
and $L_\top$ contain no duplicates, so they take $\Theta(|V|)$ space; 
the size of $\J.f$ and $\J.\texttt{cnt}$ is $|V|$, that of $\J.\texttt{cmp}$ plus $L_{\omega}$ is $\Theta(|E|)$.  
\end{proof}

Indeed, the $\texttt{J-VI}()$ keeps track of two additional array-lists, $L^{\text{inc}}_{\text{nxt}}$ and $L_{\bot}$.
The role of $L^{\text{inc}}_{\text{nxt}}$ is to ensure (a slightly weaker formulation of) (PC-2): 
during the execution of Algorithm~\ref{algo:solve_mpg}, the $\texttt{prev}_{\rho^J}(i,j)$-th invocation 
of $\texttt{J-VI}()$ handles $L^{\text{inc}}_{\text{nxt}}$ so to ensure that (a slightly weaker, 
but still sufficient form of) (PC-2) holds for the $(i,j)$-th invocation. 
However, the way in which this happens also relies on the internals of the EI-Jumps. 
Also, the EI-Jumps take care of repairing $\J.\texttt{cnt}$ and $\J.\texttt{cmp}$ so to ensure (a weaker) (PC-3). 
The weaker formulation of (PC-2), (PC-3) is discussed in SubSection~\ref{subsect:correctness}. 
From this perspective, the functioning of $\texttt{J-VI}()$ and that of the EI-Jumps is quite braided. 
In order to detail these aspects, we need to observe the following fact.
\begin{Prop}\label{prop:jvalue_init} Let $i\in [W^-,W^+]$ and $j\in [1,s-1]$.
Assume that $\texttt{J-VI}(i,j,F,\J,\Gamma)$ is invoked at some step $\iota$, 
suppose that $\J.{L^{\text{inc}}_{\text{nxt}}}^{\iota}=\emptyset$, and that (PC-1), (PC-2), (PC-3) hold at step $\iota$. 

Then, the following two facts hold:
\begin{enumerate}
\item At each step $\hat\iota\geq \iota$ of $\texttt{J-VI}()$, 
done \emph{before} the $\texttt{swap}()$ at line~\ref{algo:jvalue:l19}: 
$\J.{L^{\text{inc}}}^{\hat\iota}\subseteq \text{Inc}(f^{\text{c}:\hat\iota}, i, j)$.
\item When $\texttt{J-VI()}$ halts, \emph{after} the $\texttt{swap}()$ at line~\ref{algo:jvalue:l19}, say at step $h$, then: 
\[ \J.{L^{\text{inc}}}^{h} = \{ v\in V \mid 0<f^{\text{c}:h}(v) \neq \top \}. \]
\end{enumerate}
\end{Prop}
\emph{Proof of (1)} When $\texttt{J-VI()}$ is invoked, Item~1 holds by (PC-2). 
Then, $\texttt{J-VI}()$ can insert any $u\in V$ into $L^{\text{inc}}$ only at line~\ref{algo:jvalue:l18}, 
when exploring $N_\Gamma^{\text{in}}(v)$ (from line~\ref{algo:jvalue:l11} to line~\ref{algo:jvalue:l18}), for some $v\in V$. 
At line~\ref{algo:jvalue:l18}, $u\in V$ is inserted into 
$L^{\text{inc}}$ \textit{iff} $f_u < \Delta_{u,v}$ (line~\ref{algo:jvalue:l14}) and either 
$u\in V_1$ or $\J.\texttt{cnt}[u]=0$; \ie \textit{iff} $u$ is inconsistent \wrt $f^{\text{c}}$ in $\Gamma_{i,j}$ 
(indeed, $\J.\texttt{cnt}$ is coherent by (PC-3) and the fact that 
lines~\ref{algo:jvalue:l15}-\ref{algo:jvalue:l17} of $\texttt{J-VI}()$ preserve coherency). 
As $f^{\text{c}}(u)$ stands still while $u$ is inside $L^{\text{inc}}$, 
and $f^{\text{c}}(v)$ for any $v\in N_\Gamma^{\text{out}}(u)$ can only increase during the $\texttt{J-VI}()$, 
then Item~1 holds. \qed

\emph{Proof of (2)} Let us focus on the state of $L^{\text{inc}}_{\text{nxt}}$. 
Initially, $L^{\text{inc}}_{\text{nxt}}=\emptyset$ by hypothesis. 
During the $\texttt{J-VI()}$, $L^{\text{inc}}_{\text{nxt}}$ is 
modified only at line~\ref{algo:jvalue:l6} or \ref{algo:jvalue:l10}: 
some $v\in V$ is inserted into $L^{\text{inc}}_{\text{nxt}}$, 
say at step $\hat\iota$, (line~\ref{algo:jvalue:l6}) \textit{iff} $f_v\neq \top$ 
(where $f_v$ is the energy-level of $v$ at the time of the insertion $\hat\iota$). 
We argue that $f_v>0$ holds at $\hat\iota$ (line~\ref{algo:jvalue:l6}): 
since $v$ was extracted from $L^{\text{inc}}$ (line~\ref{algo:jvalue:l2}), 
and since all vertices in $L^{\text{inc}}$ are inconsistent \wrt $f^{\text{c}:\hat\iota}$ in $\Gamma_{i,j}$ by Item~1, 
then $\delta(\cdot, v)$ had really increased $f^{\text{c}}(v)$ (at line~\ref{algo:jvalue:l3}); 
thus, it really holds $f_v>0$ at $\hat\iota$. 
After the insertion, in case $f^{\text{c}}(v)$ becomes $\top$ at some subsequent execution of line~\ref{algo:jvalue:l3}, 
$v$ is removed from $L^{\text{inc}}_{\text{nxt}}$ (and inserted into $L_{\top}$), 
see lines~\ref{algo:jvalue:l8}-\ref{algo:jvalue:l10}. 
Finally, at line~\ref{algo:jvalue:l19} of $\texttt{J-VI()}$, 
$L^{\text{inc}}_{\text{nxt}}$ and $L^{\text{inc}}$ are swapped (line~\ref{algo:jvalue:l19}). 
Therefore, at that point, Item~2 holds. \qed

When $\texttt{J-VI}()$ halts, it is necessary to initialize $L^{\text{inc}}$ for the \emph{next} 
$\texttt{J-VI}()$ by including all the $v\in V$ such that $0<f^{\text{c}}(v)\neq\top$, because they are all inconsistent; 
this is shown by Lemma~\ref{lem:jvalue_init_correct}. 
\begin{Lem}\label{lem:jvalue_init_correct} Let $i\in [W^{-}, W^{+}]$ and $j\in [1, s-1]$, where $s\triangleq |\F_{|V|}|$. 
Assume that $\texttt{J-VI}()$ is invoked on input $(i,j,F,\J,\Gamma)$, 
and that all the pre-conditions (PC-1), (PC-2), (PC-3) are satisfied. 
Assume that $\texttt{J-VI}(i,j,F,\J,\Gamma)$ halts at step $h$.
Let $i'\in [W^{-}, W^{+}]$ and $j'\in [1, s-1]$ be any two indices such that $(i',j') > (i,j)$. If $v\in V$ satisfies 
$0 < f^{\text{c}:h}(v) \neq \top$, then $v\in \text{Inc}(f^{\text{c}:h}, i', j')$.  
\end{Lem}
\begin{proof}
Let $\hat{v}\in V$ be any vertex such that $0 < f^{\text{c}:h}(\hat{v}) \neq \top$.
By Proposition~\ref{prop:correctness}, $\forall^{v\in V}\, f^{\text{c}:h}(v)=f^*_{w'_{i,j}}(v)$ . 
Since $f^*_{w'_{i,j}}(\hat{v})$ is the least-SEPM of $\Gamma_{i,j}$, 
then it is the unique least fixed-point of 
simultaneously all operators $\{\delta(\cdot, v)\}_{v\in V}$ by Knaster-Tarski; 
therefore, the following holds:
\[ f^{\text{c}:h}(\hat{v})=f^*_{w'_{i,j}}(\hat{v})=
	\left\{
	 \begin{array}{ll}
		\min\{f^*_{w'_{i,j}}(v')\ominus w'_{i,j}(\hat{v},v')\mid 
				v'\in N_\Gamma^{\text{out}}(\hat{v}) \} , & \text{ if } \hat{v} \in V_0 \\ 
		\max\{f^*_{w'_{i,j}}(v')\ominus w'_{i,j}(\hat{v},v')\mid 
				v'\in N_\Gamma^{\text{out}}(\hat{v}) \} , & \text{ if } \hat{v} \in V_1
	 \end{array} 
	\right. 
\]
Since $0<f^{\text{c}:h}(\hat{v})\neq \top$, it is safe to discard the $\ominus$ operator in the equality above. 
Moreover, since $(i',j') > (i,j)$, then $w'_{i,j} > w'_{i',j'}$. Therefore, the following inequality holds:
\begin{align*}
f^{\text{c}:h}(\hat{v}) & = 
\left\{
 \begin{array}{ll}
	\min\{f^{\text{c}:h}(v') - w'_{i,j}(\hat{v},v')\mid v'\in N_\Gamma^{\text{out}}(\hat{v}) \} , & \text{ if } \hat{v} \in V_0 \\ 
	\max\{f^{\text{c}:h}(v') - w'_{i,j}(\hat{v},v')\mid v'\in N_\Gamma^{\text{out}}(\hat{v}) \} , & \text{ if } \hat{v} \in V_1
 \end{array} 
\right. \\ 
 & < \left\{
 \begin{array}{ll}
	\min\{f^{\text{c}:h}(v') - w'_{i',j'}(\hat{v},v')\mid v'\in N_\Gamma^{\text{out}}(\hat{v}) \} , & \text{ if } \hat{v} \in V_0 \\ 
	\max\{f^{\text{c}:h}(v') - w'_{i',j'}(\hat{v},v')\mid v'\in N_\Gamma^{\text{out}}(\hat{v}) \} , & \text{ if } \hat{v} \in V_1
 \end{array} 
\right. 
\end{align*}
So, restoring $\ominus$, we have $f^{\text{c}:h}(\hat{v})\prec \left\{
 \begin{array}{ll}
	f^{\text{c}:h}(v') \ominus w'_{i',j'}(\hat{v},v') \text{ for all }  
							v'\in N_\Gamma^{\text{out}}(\hat{v}) , & \text{ if } \hat{v} \in V_0 \\ 
	f^{\text{c}:h}(v') \ominus w'_{i',j'}(\hat{v},v') \text{ for some }  
							v'\in N_\Gamma^{\text{out}}(\hat{v}) , & \text{ if } \hat{v} \in V_1
 \end{array} 
\right.$ 

Therefore,  $v\in \text{Inc}(f^{\text{c}:h}, i', j')$.
\end{proof}
Although, when the $\texttt{prev}_{\rho^J}(i,j)$-th $\texttt{J-VI}()$ halts, 
it is correct --and necessary-- to initialize $L^{\text{inc}}$ for the $(i,j)$-th $\texttt{J-VI}()$ by including 
all those $v\in V$ such that $0<f^{\text{c}}(v)\neq\top$ (because they are all inconsistent \wrt 
to $f^{\text{c}}$ in $\Gamma_{i,j}$ by Lemma~\ref{lem:jvalue_init_correct}), 
still, we observe that this is not sufficient. Consider the following two facts (I-1) and (I-2):

(I-1) It may be that, when the $\texttt{prev}_{\rho^J}(i,j)$-th $\texttt{J-VI}()$ halts, 
it holds for all $v\in V$ that either $f^{\text{c}}(v)=0$ or $f^{\text{c}}(v)=\top$. 
In that case, $L^{\text{inc}}$ would be empty (if nothing more than 
what prescribed by Proposition~\ref{prop:jvalue_init} is done). 
We need to prevent this from happening, so to avoid \emph{vain} Scan-Phases.
	
(I-2) When going, say, from the $(i-1,j)$-th to the $(i,1)$-th Scan-Phase, there might be some $(u,v)\in E$ such that: 
$f^{\text{c}}(u)=0=f^{\text{c}}(v)$ and $w(u,v)=i$;  
those $(u,v)$ may become incompatible \wrt $f^{\text{c}}$ in $\Gamma_{i,1}$ (because $i-1$ had been increased to $i$), 
possibly breaking the compatibility (and thus the coherency) of $(u,v)$. 
These incompatible arcs are not taken into account by Proposition~\ref{prop:jvalue_init}, 
nor by Lemma~\ref{lem:jvalue_init_correct}. 
Thus a special care is needed in order to handle them.

\emph{Energy-Increasing-Jumps.} To resolve the issues raised in I-1 and I-2, the \emph{EI-Jumps} will come into play.
The pseudocode of the EI-Jumps is provided in SubProcedure~\ref{subproc:energy_jumps}. 
The $\texttt{ei-jump}(i,\J)$ really makes a jump only when $L^{\text{inc}}=\emptyset$ holds invocation. 
Basically, if $L^{\text{inc}}=\emptyset$ (line~\ref{subproc:energy_jump:l1}) 
we aim at avoiding \emph{vain} Scan-Phases, \ie (I-1); 
still, we need to take care of some additional (possibly) incompatible arcs, \ie (I-2). Recall, 
$L^{\text{inc}}$ is initialized by the $\texttt{J-VI}()$ itself according to Proposition~\ref{prop:jvalue_init}. 
Therefore, at line~\ref{subproc:energy_jump:l1}, $L^{\text{inc}}=\emptyset$ \textit{iff} 
either $\J.f(v)=0$ \textit{or} $\J.f(v)=\top$ for every $v\in V$.

\begin{wrapfigure}[24]{r}{7cm}
\raisebox{0pt}[\dimexpr\height-1.1\baselineskip\relax]{
\begin{algo-proc}[H]
\caption{EI-Jump}\label{subproc:energy_jumps}
\scriptsize
\DontPrintSemicolon
\nonl \SetKwProg{Fn}{Procedure}{}{}
\label{proc:energy_jump}
\Fn{$\texttt{ei-jump}(i,\J)$}{
    \SetKwInOut{Input}{input}
    \SetKwInOut{Output}{output}
\Input{Jumper $\J$.}
\Output{\texttt{T} if an EI-Jump occurs; else, \texttt{F}.}
\If{$L^{\text{inc}}=\emptyset$ \label{subproc:energy_jump:l1}}{
	$L^{\text{inc}}\leftarrow L^{\text{inc}}_{\text{cpy}}$;  
	$L^{\text{inc}}_{\text{cpy}}\leftarrow\emptyset$ \label{subproc:energy_jump:l2};\;
	$\J.i\leftarrow i+1$; \label{subproc:energy_jump:l3} \;	
	\If{$L_{\omega}\neq\emptyset$}{ \label{subproc:energy_jump:l4}
	   $(w, L_{\alpha})\leftarrow \texttt{read\_front}(L_{\omega})$; \label{subproc:energy_jump:l5}\;
		\If{$w=\J.i$ \label{subproc:energy_jump:l6}}{
			$\texttt{pop\_front}(L_\omega)$; \label{subproc:energy_jump:l7}\; 
			$\texttt{repair}(L_{\alpha}, \J)$;\label{subproc:energy_jump:l8}
		}
	}
	\While{$L^{\text{inc}}=\emptyset$ \texttt{\bf and }$L_{\omega}\neq\emptyset$ \label{subproc:energy_jump:l9}}{
		$(w, L_{\alpha})\leftarrow \texttt{pop\_front}(L_\omega)$; \label{subproc:energy_jump:l10}\;
		$\J.i\leftarrow w;$ \label{subproc:energy_jump:l11}\;
		$\texttt{repair}(L_{\alpha}, \J)$; \label{subproc:energy_jump:l12}\;
	}
	\Return{\texttt{T}}; \label{subproc:energy_jump:l13}\;
}
\lElse{\Return{\texttt{F}; \label{subproc:energy_jump:l14}}}
}
\setcounter{AlgoLine}{0}
\SetKwProg{SubFn}{SubProcedure}{}{}
\nonl\SubFn{$\texttt{repair}(L_{\alpha}, \J)$}{
    \SetKwInOut{Input}{input}
    \SetKwInOut{Output}{output}
\Input{A list of arcs $L_{\alpha}$, reference to Jumper $\J$.}
\ForEach{$(u,v)\in L_{\alpha}$ \label{subproc:repair:l1}}{
	\If{$\J.f[u]=0$ \texttt{\bf and }$\J.f[v]=0$ \texttt{\bf and} $L^{\text{inc}}[u]=\bot$ \label{subproc:repair:l2}}{ 
			\If{$u\in V_0$ \label{subproc:repair:l3}}{
				$\J.\texttt{cnt}[u]\leftarrow \J.\texttt{cnt}[u]-1;$ \label{subproc:repair:l4}\;	
				$\J.\texttt{cmp}\big[(u,v)\big]\leftarrow \texttt{F}$; \label{subproc:repair:l5}\;
				\If{$\J.\texttt{cnt}[u]=0$ \label{subproc:repair:l6}}{ 
					$\texttt{insert}(u, L^{\text{inc}})$; \label{subproc:repair:l7}		
				}		
			}\lIf{$u\in V_1$ \label{subproc:repair:l8}}{  
				$\texttt{insert}(u, L^{\text{inc}})$; \label{subproc:repair:l9}
			}	
		}
	} 
}
\end{algo-proc}} 
\end{wrapfigure} 

To begin with, if $L^{\text{inc}}=\emptyset$ (line~\ref{subproc:energy_jump:l1}), 
copy $L^{\text{inc}}\leftarrow L^{\text{inc}}_{\text{cpy}}$, then, 
erase $L^{\text{inc}}_{\text{cpy}}\leftarrow\emptyset$ (line~\ref{subproc:energy_jump:l2}): 
this is related to the steps of backtracking that are performed by the UA-Jumps, 
we will give more details on this later on. Next, we increment $i$ to $\J.i\leftarrow i+1$ (line~\ref{subproc:energy_jump:l3}). 
Then, if $L_{\omega}\neq\emptyset$ at line~\ref{subproc:energy_jump:l4}, 
we read (read-only) the front entry $(\hat{w},L_{\hat{\alpha}})$ of $L_{\omega}$ (line~\ref{subproc:energy_jump:l5}); 
only if $\hat{w}=\J.i$ (line~\ref{subproc:energy_jump:l6}), 
we pop $(\hat{w}, L_{\hat{\alpha}})$ out of $L_\omega$ (line~\ref{subproc:energy_jump:l7}), 
and we invoke $\texttt{repair}(L_{\hat{\alpha}}, \J)$ (line~\ref{subproc:energy_jump:l8}) to 
repair the coherency state of all those arcs (\ie all and only those in $L_{\hat{\alpha}}$) that we mentioned in (I-2). 
We will detail $\texttt{repair}()$ shortly, now let us proceed with $\texttt{ei-jump}()$. 
At line~\ref{subproc:energy_jump:l9}, \textit{while} $L^{\text{inc}}=\emptyset$ \textit{and} $L_{\omega}\neq\emptyset$: 
the front $(\bar{w},L_{\bar{\alpha}})$ is popped from $L_\omega$ (line~\ref{subproc:energy_jump:l10}) 
and $\J.i\leftarrow \bar{w}$ is assigned (line~\ref{subproc:energy_jump:l11}). 
The ending-point of the EI-Jump will now reach $\bar{w}$ (at least). 
A moment's reflection reveals that, jumping up to $\bar{w}$, 
some arcs $(u,v)\in E$ such that $f^{\text{c}}(u)=0=f^{\text{c}}(v)$ (which were compatible \wrt the $(i,j)$-th Scan-Phase, 
just \emph{before} the jump) may become incompatible for the $(\bar{w}, 1)$-th Scan-Phase (which is now candidate to happen), 
because $\bar{w}>i$. What are these new incompatible arcs? Since $L_{\omega}$ was sorted in increasing order, 
they're all \emph{and only} those of weight $w(u,v)=\bar{w}=\J.i$; 
\ie those in the $L_{\bar{\alpha}}$ that is binded to $\bar{w}$ in $L_{\omega}$. 
To repair coherency, $\texttt{repair}(L_{\bar{\alpha}}, \J)$ (line~\ref{subproc:energy_jump:l12}) is invoked. 
This repeats until $L^{\text{inc}}\neq\emptyset$ \textit{or} $L_{\omega}=\emptyset$. 
Then, $\texttt{ei-jump}()$ returns $\texttt{T}$ (at line~\ref{subproc:energy_jump:l13}).

If $L^{\text{inc}}\neq\emptyset$ at line~\ref{subproc:energy_jump:l1}, 
then $\texttt{F}$ is returned (line~\ref{subproc:energy_jump:l14}); so, in that case, \emph{no} EI-Jump will occur.  

Let us detail the $\texttt{repair}(L_\alpha, \J)$. On input $(L_{\alpha}, \J)$, 
for each arc $(u,v)\in L_{\alpha}$ (line~\ref{subproc:repair:l1}), if $\J.f[u]=0=\J.f[v]$ and 
$L^{\text{inc}}[u]=\bot$ (line~\ref{subproc:repair:l2}), the following happens. 
If $u\in V_1$, then $u$ is promptly inserted (in front of) $L^{\text{inc}}$ (line~\ref{subproc:repair:l8}); 
else, if $u\in V_0$, $\J.\texttt{cnt}[u]$ is decremented by one unit (line~\ref{subproc:repair:l4}); 
also, it is flagged $\J.\texttt{cmp}[(u,v)]\leftarrow\texttt{F}$ (line~\ref{subproc:repair:l5}). 
After that, if $\J.\texttt{cnt}[u]=0$ (line~\ref{subproc:repair:l6}), 
then $u$ is inserted in front of $L^{\text{inc}}$ (line~\ref{subproc:repair:l7}). 
The following proposition holds for the $\texttt{ei-jump}()$ (SubProcedure~\ref{subproc:energy_jumps}).
\begin{Prop}\label{prop:ei-jump_halts} 
The $\texttt{ei-jump}()$ (SubProcedure~\ref{subproc:energy_jumps}) halts in finite time. 
The \emph{total} time spent for all invocations of $\texttt{ei-jump}()$ (that are made, at line~\ref{algo:solve:l7}, 
during the main \texttt{while} loop of Algorithm~\ref{algo:solve_mpg}) is $\Theta(t_{\ell_{\ref{algo:solve:l7}}}+|E|)$, 
where $t_{\ell_{\ref{algo:solve:l7}}}$ is the total number of iterations 
of line~\ref{algo:solve:l7} that are made by Algorithm~\ref{algo:solve_mpg}. 
The $\texttt{ei-jump}()$ works with $\Theta(|V|+|E|)$ space.
\end{Prop}
\begin{proof} 
The \texttt{for-each} loop in $\texttt{repair}()$ is bounded: each arc $(u,v)$ of $L_{\alpha}$ is visited exactly once, 
spending $O(1)$ time per each. The \texttt{while} loop in 
$\texttt{ei-jump}()$ (lines~\ref{subproc:energy_jump:l9}-\ref{subproc:energy_jump:l12}) is also bounded: 
it consumes the elements $(w,L_{\alpha})$ of $L_{\omega}$, 
spending $O(|L_{\alpha}|)$ time per cycle. There are no other loops in $\texttt{ei-jump}()$, so it halts in finite time. 
Now, consider the following three facts: 
(i) $\texttt{ei-jump}()$ is invoked by $\texttt{solve\_MPG()}$ (Algorithm~\ref{algo:solve_mpg})  
once per each iteration of the main \texttt{while} loop at line~\ref{algo:solve:l7}. 
Assume there are $t_{\ell_{\ref{algo:solve:l7}}}$ such iterations overall. 
(ii) either $\texttt{ei-jump}()$ returns immediately or it visits $k$ 
arcs $(u,v)\in E$ in time $\Theta(k)$, for some $1\leq k\leq |E|$; 
(iii) each arc $(u,v)\in E$ is visited by $\texttt{ei-jump}()$ at most once 
during the whole execution of Algorithm~\ref{algo:solve_mpg}, 
because the elements of $L_{\omega}$ are consumed and there are no duplicates in there. 
Altogether, (i), (ii) and (iii) imply the $\Theta(t_{\ell_{\ref{algo:solve:l7}}}+|E|)$ total running time. 
Moreover, $\texttt{ei-jump}()$ works with $\Theta(|V|+|E|)$ space. 
Indeed $L^{\text{inc}}$ contains no duplicated vertices, so: $|L^{\text{inc}}|\leq |V|$, $|L_{\omega}|=|E|$, 
the size of $\J.f$ and that of $\J.\texttt{cnt}$ is $|V|$, and the size of $\J.\texttt{cmp}$ is $|E|$.
\end{proof}

The description of Algorithm~\ref{algo:solve_mpg} ends by detailing the UA-Jumps.

\emph{Unitary-Advance-Jumps.} Recall, UA-Jumps are adopted so to scroll 
through $\F_{|V|}$ only \emph{when} (and \emph{where}) it is really necessary; 
that is only $|V|$ times at most, because each time at least one vertex will take a value. 
The pseudocode is shown in \figref{subproc:ua-jumps}. 

The UA-Jumps begin soon after that $\texttt{ei-jump()}$ returns $\texttt{T}$ 
at line~\ref{algo:solve:l8} of Algorithm~\ref{algo:solve_mpg}. 
The starting point of the UA-Jumps (\ie the initial value of $i$) is provided 
by $\texttt{ei-jump()}$ (line~\ref{algo:solve:l7} of Algorithm~\ref{algo:solve_mpg}): 
it is stored into $\J.i$ and passed in input to 
$\texttt{ua-jumps}(\J.i, s, F, \J, \Gamma)$ (at line~\ref{algo:solve:l9} of Algorithm~\ref{algo:solve_mpg}). 
Starting from $i=\J.i$, basically the $\texttt{ua-jumps}()$ repeats a sequence of invocations to $\texttt{J-VI}()$, 
on input $(i, s-1), (i+1, s-1), (i+2, s-1), \cdots, (\hat{i}, s-1)$; 
until $L_{\top}\triangleq \W_0(\Gamma_{\hat{i}-1, s-1})\cap\W_1(\Gamma_{\hat{i}, s-1})\neq\emptyset$ 
holds for some $\hat{i}\geq i$. When $L_{\top}\neq\emptyset$, 
the $\texttt{ua-jumps}()$ \emph{backtracks} the Scan-Phases from the $(\hat{i}, s-1)$-th to the $(\hat{i},1)$-th one, 
by invoking $\texttt{backtrack\_ua-jump}(i,s,F,\J,\Gamma)$, and then it halts; soon after, Algorithm~\ref{algo:solve_mpg} 
will begin scrolling through $\F_{|V|}$ by invoking another sequence 
of $\texttt{J-VI}()$ (this time at line~\ref{algo:solve:l11} of Algorithm~\ref{algo:solve_mpg}) 
on input $(\hat{i}, 1), (\hat{i}, 2), (\hat{i}, 3), \ldots$ (which 
is controlled by the \texttt{while} loop at line~\ref{algo:solve:l6} of Algorithm~\ref{algo:solve_mpg}). 
More details concerning the UA-Jumps now follow. 

So, $\texttt{ua-jumps}()$ (SubProcedure~\ref{subproc:ua-jumps}) performs a sequence of UA-Jumps (actually, at least one). 
The invocation to $\texttt{J-VI}(\hat{i},s-1,F,\J,\Gamma)$ repeats for $\hat{i}\geq i$ 
(lines~\ref{subproc:ua-jumps:l1}-\ref{subproc:ua-jumps:l2}), until $L_{\top}\neq\emptyset$ (line~6). 
There, $L_{\top}$ contains all and only those $v\in V$ whose 
energy-level became $f(v)=\top$ during the last performed $\texttt{J-VI}()$ (line~\ref{subproc:ua-jumps:l2}); 
so, at line~\ref{subproc:ua-jumps:l3}, it is $L_{\top}=\W_0(\Gamma_{\hat{i}-1,s-1}) \cap \W_1(\Gamma_{\hat{i},s-1})$. 
At this point, if $L_{\top}=\emptyset$ (line~\ref{subproc:ua-jumps:l3}), 
the procedure prepares itself to make another UA-Jump: $\hat{i}\leftarrow \hat{i}+1$ is set (line~\ref{subproc:ua-jumps:l4}), 
and then $\texttt{rejoin\_ua-jump}(\hat{i},s,F,\J)$ is invoked (line~\ref{subproc:ua-jumps:l5}). 
Else, if $L_{\top}\neq\emptyset$ (line~6), 
it is invoked $\texttt{backtrack\_ua-jump}(\hat{i}, s, F, \J, \Gamma)$ (line~\ref{subproc:ua-jumps:l7}), 
and then $(i,S)$ is returned (line~\ref{subproc:ua-jumps:l8}), 
where $S\triangleq L_\top = \W_0(\Gamma_{\hat{i}-1,s-1}) \cap \W_1(\Gamma_{\hat{i},s-1})$ 
was assigned at line~\ref{subproc:ua-jumps_backtrack:l4} of $\texttt{backtrack\_ua-jump}()$. 

The $\texttt{rejoin\_ua-jump}(i,s,F,\J)$ firstly copies the energy-levels stored 
in $L_{f}$ back to $\J.f$, by invoking $\texttt{scl\_back\_}f(s-1,F,\J)$ (line~\ref{subproc:unitary_join:l1}).
Secondly, at lines~\ref{subproc:unitary_join:l4}-\ref{subproc:unitary_join:l6}, by operating in the same way as $\texttt{ei-jump}()$ 
does (see lines~\ref{subproc:energy_jump:l4}-\ref{subproc:energy_jump:l8} 
of $\texttt{ei-jump}()$, SubProcedure~\ref{subproc:energy_jumps}), 
it repairs the coherency state of $\J.\texttt{cnt}$ and $\J.\texttt{cmp}$ \wrt 
all those arcs $(u,v)\in E$ such that $w(u,v)=i$ and $\J.f[u]=0=\J.f[v]$. 

Let us detail the $\texttt{backtrack\_ua-jump}()$. Basically, 
it aims at preparing a correct state so to allow Algorithm~\ref{algo:solve_mpg} to step through $\F_{|V|}$. 
Stepping through $\F_{|V|}$ essentially means to execute a sequence 
of $\texttt{J-VI}()$ at line~\ref{algo:solve:l11} of Algorithm~\ref{algo:solve_mpg}, until $L^{\text{inc}}=\emptyset$. 
A moment's reflection reveals that this sequence of $\texttt{J-VI}()$ can run just on the sub-arena of $\Gamma$ that is induced by 
$S\triangleq L_\top=\W_0(\Gamma_{i-1,s-1}) \cap \W_1(\Gamma_{i,s-1})$ (see 
line~\ref{subproc:ua-jumps_backtrack:l4} of $\texttt{backtrack\_ua-jump}()$); 
there is no real need to lift-up again (actually, slowly than before) 
all the energy-levels of the component induced by $V\setminus L_\top$: 
those energy-levels can all be confirmed now that the UA-Jumps are finishing, 
and they can all stand still while Algorithm~\ref{algo:solve_mpg} is stepping through $\F_{|V|}$ at line~\ref{algo:solve:l11}, 
until another EI-Jump occurs. 

For this reason, $\texttt{backtrack\_ua-jump}(\hat{i}, s, F, \J, \Gamma)$ works as follows. 

\begin{wrapfigure}[32]{r}{6.5cm}
\raisebox{0pt}[\dimexpr\height-1\baselineskip\relax]{
\begin{algo-proc}[H]
\caption{UA-Jumps}\label{subproc:ua-jumps}
\scriptsize
\DontPrintSemicolon
\nonl \SetKwProg{Fn}{SubProcedure}{}{}
\Fn{$\texttt{ua-jumps}(i,s,F,\J,\Gamma)$}{
    \SetKwInOut{Input}{input}
    \SetKwInOut{Output}{output}
	\Input{$i\in [W^-,W^+]$, $s=|\F_{|V|}|$, $F$ is a ref. to $\F_{|V|}$, Jump $\J$, input MPG $\Gamma$.}
	\Repeat{$L_{\top}\neq\emptyset$}{ \label{subproc:ua-jumps:l1}
		$\texttt{J-VI}(i, s-1, F, \J, \Gamma);$ \textit{/* UA-Jump */} \label{subproc:ua-jumps:l2} \; 
		\If{$L_{\top}=\emptyset$}{ \label{subproc:ua-jumps:l3}
			$i\leftarrow i+1$; \label{subproc:ua-jumps:l4} \; 
			$\texttt{rejoin\_ua-jump}(i,s,F,\J)$; \label{subproc:ua-jumps:l5} \; 
		} 
	} 
	$S\leftarrow \texttt{backtrack\_ua-jump}(i,s,F,\J,\Gamma)$; \label{subproc:ua-jumps:l7} \; 
	\Return{$(i, S)$}; \label{subproc:ua-jumps:l8}
}
\setcounter{AlgoLine}{0}
\nonl\Fn{$\texttt{rejoin\_ua-jump}(i,s,F,\J)$}{
    \SetKwInOut{Input}{input}
    \SetKwInOut{Output}{output}
\Input{$i\in [W^-,W^+]$, $F$ is a ref. to $\F_{|V|}$, Jump $\J$.}
$\texttt{scl\_back\_}f(s-1,F,\J)$; \label{subproc:unitary_join:l1} \;
\If{$L_{\omega}\neq\emptyset$}{ \label{subproc:unitary_join:l2}
			$(w, L_{\alpha})\leftarrow \texttt{read\_front}(L_{\omega})$; \label{subproc:unitary_join:l3}\;
			\If{$w=i$ \label{subproc:unitary_join:l4}}{
				$\texttt{pop\_front}(L_\omega)$; \label{subproc:unitary_join:l5}\; 
				$\texttt{repair}(L_{\alpha}, \J)$; 
					\texttt{// see SubProc.~\ref{subproc:energy_jumps}} \label{subproc:unitary_join:l6}
			}
	}
}
\setcounter{AlgoLine}{0}
\nonl\Fn{$\texttt{backtrack\_ua-jump}(i,s,F,\J,\Gamma)$}{
    \SetKwInOut{Input}{input}
    \SetKwInOut{Output}{output}
\Input{$i\in [W^-,W^+]$, $s=|\F_{|V|}|$, Jump $\J$, \MPG $\Gamma$.}
$L^{\text{inc}}_{\text{cpy}}\leftarrow L^{\text{inc}}$; 
	$L^{\text{inc}}\leftarrow\emptyset$; \label{subproc:ua-jumps_backtrack:l1}\;
$L_f[u]\leftarrow \left\{
	\begin{array}{ll} 
		\bot &, \text{ if } u \in L_\top; \\
		L_f[u] &, \text{ if } u \in V\setminus L_\top.
	\end{array}\right. $ \label{subproc:ua-jumps_backtrack:l2}\; 
$\texttt{scl\_back\_}f(s-1,F,\J)$;\label{subproc:ua-jumps_backtrack:l3}\;
$S\leftarrow L_\top$; \label{subproc:ua-jumps_backtrack:l4}\;
\While{$L_\top\neq\emptyset$\label{subproc:ua-jumps_backtrack:l5}}{
	$u\leftarrow\texttt{pop\_front}(L_\top)$\label{subproc:ua-jumps_backtrack:l6}\;
	\If{$u\in V_0$\label{subproc:ua-jumps_backtrack:l7}}{
		$\texttt{init\_cnt\_cmp}(u, i, 1, F, J, \Gamma[S])$;\label{subproc:ua-jumps_backtrack:l8}\;
		\If{$J.\texttt{cnt}[u]=0$\label{subproc:ua-jumps_backtrack:l9}}{
			$\texttt{insert}(u, L^{\text{inc}})$;\label{subproc:ua-jumps_backtrack:l10}\;
		}
	}\If{$u\in V_1$\label{subproc:ua-jumps_backtrack:l11}}{
		\ForEach{$v\in N_{\Gamma[S]}^{\text{out}}(u)$\label{subproc:ua-jumps_backtrack:l12}}{
			$f_u\leftarrow \texttt{get\_scl\_}f(u,1,F,J)$;\label{subproc:ua-jumps_backtrack:l13}\;
			$f_v\leftarrow \texttt{get\_scl\_}f(v,1,F,J)$;\label{subproc:ua-jumps_backtrack:l14}\;
			$w'\leftarrow \texttt{get\_scl\_}w(w(u,v),i,1,F)$;\label{subproc:ua-jumps_backtrack:l15}\;
			\If{$f_u \prec f_v \ominus w'$\label{subproc:ua-jumps_backtrack:l16}}{
				$\texttt{insert}(u, L^{\text{inc}})$; \textbf{break};\label{subproc:ua-jumps_backtrack:l17}\;
			}
		}	
	}
}
\Return{$S$} \label{subproc:ua-jumps_backtrack:l18};
}
\end{algo-proc}}
\end{wrapfigure}

Firstly, we copy $L^{\text{inc}}_{\text{cpy}}\leftarrow L^{\text{inc}}$, 
then we erase $L^{\text{inc}}\leftarrow\emptyset$ (line~\ref{subproc:ua-jumps_backtrack:l1}). 
This is sort of a back-up copy, notice that $L^{\text{inc}}_{\text{cpy}}$ will be restored back to $L^{\text{inc}}$ at 
line~\ref{subproc:energy_jump:l2} of $\texttt{ei-jump}()$ (SubProcedure~\ref{subproc:energy_jumps}): 
when Algorithm~\ref{algo:solve_mpg} will finish to step through $\F_{|V|}$, 
it will hold $L^{\text{inc}}=\emptyset$ at line~\ref{subproc:energy_jump:l1} 
of $\texttt{ei-jump}()$ (SubProcedure~\ref{subproc:energy_jumps}), so at that point the state 
of $L^{\text{inc}}$ will need to be restored by including (at least) all those vertices that are now assigned 
to $L^{\text{inc}}_{\text{cpy}}$ at line~\ref{subproc:ua-jumps_backtrack:l1} of $\texttt{backtrack\_ua-jump}()$. 
Next, all the energy-levels of $V\setminus L_\top$ are confirmed and saved back to $\J.f$; this is done: (i) by setting,  

$
L_f[u]\leftarrow \left\{
	\begin{array}{ll} 
		\bot &, \text{ if } u \in L_\top; \\
		L_f[u] &, \text{ if } u \in V\setminus L_\top.
	\end{array}\right. \text{ (line~\ref{subproc:ua-jumps_backtrack:l2}) }
$ 

and (ii) by invoking $\texttt{scl\_back\_}f(s-1,F,\J)$ (line~\ref{subproc:ua-jumps_backtrack:l3}). 
The energy-levels of all $v\in L_\top$ are thus restored as they 
were at the end of the $(\hat{i}-1, s-1)$-th invocation of 
$\texttt{J-VI}()$ at line~\ref{subproc:ua-jumps:l2} of $\texttt{ua-jumps}()$. 
Next, it is assigned $S\leftarrow L_\top$ at line~\ref{subproc:ua-jumps:l4}.
Then, $\texttt{backtrack\_ua-jump}()$ takes care of preparing a correct state of 
$L^{\text{inc}}$, $\J.\texttt{cnt}$, $\J.\texttt{cmp}$ for letting Algorithm~\ref{algo:solve_mpg} stepping through $\F_{|V|}$. 

While $L_\top\neq\emptyset$ (line~\ref{subproc:ua-jumps_backtrack:l5}), we pop the front element of $L_\top$, 
	\ie $u\leftarrow \texttt{pop\_front}(L_\top)$ (line~\ref{subproc:ua-jumps_backtrack:l6}): 

-- If $u\in V_0$ (line~\ref{subproc:ua-jumps_backtrack:l7}), then we compute $\J.\texttt{cnt}[u]$ and we also compute for every 
$v\in N_{\Gamma[S]}^{\text{out}}(u)$ a coherent $\J.\texttt{cmp}[(u,v)]$ \wrt $f^{\text{c}}$ in ${\Gamma[S]}_{\hat{i}, 1}$, 
by $\texttt{init\_cnt\_cmp}(u, \hat{i}, 1, F, J, \Gamma[S])$ (line~\ref{subproc:ua-jumps_backtrack:l8}); 
finally, if $\J.\texttt{cnt}[u]=0$ (line~\ref{subproc:ua-jumps_backtrack:l9}), 
we insert $u$ into $L^{\text{inc}}$ (line~\ref{subproc:ua-jumps_backtrack:l10}).

-- Else, if $u\in V_1$ (line~\ref{subproc:ua-jumps_backtrack:l11}), 
we explore $N_{\Gamma[S]}^{\text{out}}(u)$ looking for some incompatible 
arc (lines~\ref{subproc:ua-jumps_backtrack:l12}-\ref{subproc:ua-jumps_backtrack:l17}). 
For each $v\in N_{\Gamma[S]}^{\text{out}}(u)$ (line~\ref{subproc:ua-jumps_backtrack:l12}), 
if $f_u \prec f_v \ominus w'_{\hat{i}, 1}(u,v)$ (\ie if $(u,v)$ is incompatible \wrt $f^\text{c}$ in $\Gamma_{\hat{i}, 1}$), 
where $f_u\leftarrow \texttt{get\_scl\_}f(u,1,F,J)$ and $f_v\leftarrow \texttt{get\_scl\_}f(v,1,F,J)$, 
then, we insert $u$ into $L^{\text{inc}}$ at line~\ref{subproc:ua-jumps_backtrack:l17} (also breaking the \texttt{for-each} cycle).

This concludes the description of the UA-Jumps. Algorithm~\ref{algo:solve_mpg} is completed. 

\subsection{Correctness of Algorithm~\ref{algo:solve_mpg}}\label{subsect:correctness}
This subsection presents the proof of correctness for Algorithm~\ref{algo:solve_mpg}. It is organized as follows. 
Firstly, we show that $\texttt{J-VI}()$ (SubProcedure~\ref{algo:j-value-iteration}) works fine even when assuming 
a relaxed form of the pre-conditions (PC-2) and (PC-3). Secondly, we identify an additional set of pre-conditions under which 
the $\texttt{ei-jump}()$ (SubProcedure~\ref{subproc:energy_jumps}) is correct. Thirdly, we prove that under these pre-conditions 
$\texttt{ua-jumps}()$ (SubProcedure~\ref{subproc:ua-jumps}) is also correct. Finally, we show that these pre-conditions 
are all satisfied during the execution of Algorithm~\ref{algo:solve_mpg}, and that the latter is thus correct. 

\subsubsection{Correctness of $\texttt{J-VI}()$ (SubProcedure~\ref{algo:j-value-iteration})}
To prove the correctness of $\texttt{J-VI}()$, the (PC-1), (PC-2), (PC-3) have been assumed in Lemma~\ref{lem:jvalue_init_correct}. 
It would be fine if they were met whenever Algoritm~\ref{algo:solve_mpg} invokes $\texttt{J-VI()}$. 
Unfortunately, (PC-2) and (PC-3) may not hold. Still, we shall observe that a weaker formulation of them, 
denoted by (w-PC-2) and (w-PC-3), really hold; and these will turn out to be enough for proving correctness.

\begin{Def}
Let $i\in [W^-, W^+]$ and $j\in [1, s-1]$. Fix some step of execution $\iota$ of Algorithm~\ref{algo:solve_mpg}. 

The pre-conditions (w-PC-2) and (w-PC-3) are defined at step $\iota$ as follows.
\begin{itemize}
\item[(w-PC-2)] ${L^{\text{inc}}}^\iota\subseteq \text{Inc}(f^{\text{c}:\iota}, i, j)$.
\item[(w-PC-3)] $\forall (u\in V\setminus 
{L^{\texttt{inc}}}^\iota)$ $\forall (v\in N_{\Gamma}^{\text{out}}(u))$:
\begin{itemize}
\item[If $u\in V_0$,] the following three properties hold on $\J.\texttt{cnt}^{\iota}$ and $\J.\texttt{cmp}^{\iota}$: 
\begin{enumerate}
	\item If $\J.\texttt{cmp}^{\iota}[(u,v)]=\texttt{F}$, then $(u,v)$ is incompatible \wrt $f^\text{c}$ in $\Gamma_{i,j}$; 
		
	\item If $\J.\texttt{cmp}^{\iota}[(u,v)]=\texttt{T}\,$ 
			and $(u,v)$ is incompatible \wrt $f^\text{c}$ in $\Gamma_{i,j}$, then $v\in {L^{\text{inc}}}^\iota$.
	\item $\J.\texttt{cnt}^{\iota}[u] = \big| \big\{ v \in N^{\text{out}}_{\Gamma}(u) \mid \J.\texttt{cmp}^{\iota}[(u,v)] = 
								\texttt{T} \big\}\big|$ and $\J.\texttt{cnt}^{\iota}[u]>0$.
\end{enumerate}
\item[If $u\in V_1$,] and $(u,v)\in E$ is incompatible \wrt $f^\text{c}$ in $\Gamma_{i,j}$, then $v\in {L^{\text{inc}}}^\iota$.
\end{itemize}
If (w-PC-3) holds on $\J.\texttt{cnt}^{\iota}$ and $\J.\texttt{cmp}^{\iota}$, 
	they are said \emph{weak-coherent} \wrt $f^{\text{c}}$ in $\Gamma_{i,j}$.  
\end{itemize}
\end{Def}

We will also need the following Lemma~\ref{lem:monotone_energy}, it asserts 
that $\psi_{\rho}:(i,j)\rightarrow f^*_{i,j}$ is monotone non-decreasing; the proof already appears in [\cite{CR16}, Lemma~8, Item~1].
\begin{Lem}\label{lem:monotone_energy}
Let $i,i'\in [W^-, W^+]$ and $j,j'\in [1, s-1]$ be any two indices such that $(i,j) < (i',j')$. 

Then, $\forall^{v\in V} f^*_{i,j}(v)\preceq f^*_{i',j'}(v)$. 
\end{Lem}

Proposition~\ref{prop:saint_weak_coherency} shows that (PC-1), (w-PC-2), (w-PC-3) suffices for the correctness of $\texttt{J-VI}()$. 
\begin{Prop}\label{prop:saint_weak_coherency}
The $\texttt{J-VI}()$ (SubProcedure~\ref{algo:j-value-iteration}) is correct (\ie Propositions~\ref{prop:correctness} 
and \ref{prop:jvalue_init} still hold) even if (PC-1), (w-PC-2), (w-PC-3) are assumed instead of (PC-1), (PC-2), (PC-3).

In particular, suppose that $\texttt{J-VI}()$ is invoked on input $(i,j,F,\J,\Gamma)$, 
say at step $\iota$, and that all of the pre-conditions (PC-1), (w-PC-2), (w-PC-3) hold at $\iota$. 
When $\texttt{J-VI}(i,j,F,\J,\Gamma)$ halts, say at step $h$, then all of the following four propositions hold:
\begin{enumerate} 
	\item $f^{\text{c}:h}$ is the least-SEPM of the \EG $\Gamma_{i,j}$; 
	\item $\J.\texttt{cnt}^{h}$, $\J.\texttt{cmp}^{h}$ are both coherent \wrt $f^{\text{c}:h}$ in $\Gamma_{i,j}$;
	\item ${L^{\text{inc}}}^{h}=\{ v \in V \mid 0 < f^{\text{c}:h}(v) \neq \top \}$;
	\item $L^h_\top=V_{f^{\text{c}:\iota}}\cap V\setminus V_{f^{\text{c}:h}}$. 
\end{enumerate}
\end{Prop}
\begin{proof} 
Basically, we want to prove that Propositions~\ref{prop:correctness} and \ref{prop:jvalue_init} still hold.

Suppose ${L^{\text{inc}}}^\iota=\emptyset$. 
Let $u\in V_0$. By (w-PC-3) and ${L^{\text{inc}}}^\iota=\emptyset$, 
for every $v\in N_\Gamma^{\text{out}}(u)$, $\J.\texttt{cmp}^{\iota}[(u,v)]$ 
is coherent \wrt $f^{\iota}$ in $\Gamma_{i,j}$; thus, $\J.\texttt{cnt}^{\iota}[u]$ is also coherent 
\wrt $f^{\text{c}:\iota}$ in $\Gamma_{i,j}$. Therefore, (PC-3) holds. 
Now, let $u\in V_1$. By (w-PC-3) and ${L^{\text{inc}}}^\iota=\emptyset$, for every $v\in N^{\text{out}}_\Gamma(u)$ it holds that 
$(u,v)$ is compatible \wrt $f^{c:\iota}$ in $\Gamma_{i,j}$; thus, $u$ is consistent \wrt $f^{\text{c}:\iota}$ in $\Gamma_{i,j}$.
In addition, by (w-PC-3) again, $\J.\texttt{cnt}^{\iota}[u]>0$ holds for every $u\in V_0$. 
Therefore, every $u\in V$ is consistent \wrt $f^{\text{c}:\iota}$ in $\Gamma_{i,j}$; so, (PC-2) holds.
Since (PC-1,2,3) hold, then Propositions~\ref{prop:correctness} and \ref{prop:jvalue_init} hold.

Now, suppose ${L^{\text{inc}}}^\iota\neq\emptyset$. Since $\J.\texttt{cnt}^{\iota}$ and $\J.\texttt{cmp}^{\iota}$ 
may be incoherent -- \emph{at time} $\iota$ --, there might be some $\hat{u}\in V\setminus {L^{\text{inc}}}^\iota$ 
which is already inconsistent \wrt $f^{\text{c}:\iota}$ in $\Gamma_{i,j}$ (\ie even if $u\not\in {L^{\text{inc}}}^\iota$).

Still, we claim that, during $\texttt{J-VI}()$'s execution (say at some steps $\iota', \iota''$, \ie eventually), 
for every $u\in V_0$ and $v\in N_{\Gamma}^{\text{out}}(u)$, both $\J.\texttt{cmp}^{\iota'}[(u,v)]$ 
and $\J.\texttt{cnt}^{\iota''}[u]$ will \emph{become} coherent (at $\iota'$, $\iota''$ respectively); 
and we also claim that any $u\in V_1$ which was inconsistent at $\iota$ 
will be (eventually, say at step $\iota'''$) inserted into ${L^{\text{inc}}}$. 
Indeed, at that point (say, at $\hat{\iota}=\max\{\iota', \iota'', \iota'''\}$), 
\emph{all} (and only those) $\hat{u}\in V$ that were already inconsistent at invocation time $\iota$, 
or that became inconsistent during $\texttt{J-VI}()$'s execution (until step $\hat{\iota}$), 
they will be really inserted into $L^{\text{inc}}$.

To prove it, let $\hat{u}\in V\setminus {L^{\text{inc}}}^\iota$ and 
$\hat{v}\in N_{\Gamma}^{\text{out}}(\hat{u})$ be any two (fixed) vertices such that either: 
\begin{itemize}
\item[$\hat{u}\in V_0$] and $\J.\texttt{cmp}^{\iota}[(\hat{u},\hat{v})]=\texttt{F}$: 
Then, by (w-PC-3), $(\hat{u},\hat{v})$ is incompatible \wrt $f^{\text{c}:\iota}$ in $\Gamma_{i,j}$.
\item[$\hat{u}\in V_0$] and $\J.\texttt{cmp}^{\iota}[(\hat{u},\hat{v})]=\texttt{T}$ 
but $(\hat{u},\hat{v})$ is incompatible \wrt $f^{\text{c}:\iota}$ in $\Gamma_{i,j}$: 

Then, by (w-PC-3), $\hat{v}\in {L^{\text{inc}}}^\iota$. Since $\texttt{J-VI}()$ aims precisely at emptying $L^{\text{inc}}$, 
$\hat{v}$ is popped from ${L^{\text{inc}}}^{\iota'}$ (line~\ref{algo:jvalue:l2} 
of SubProcedure~\ref{algo:j-value-iteration}) -- say at some step $\iota'$ of $\texttt{J-VI}()$'s execution. 
Soon after that, $N^{\text{in}}_{\Gamma}(\hat{v})$ is explored (lines~\ref{algo:jvalue:l11}-\ref{algo:jvalue:l18} 
of SubProcedure~\ref{algo:j-value-iteration}); so $\hat{u}$ is visited, then $(\hat{u},\hat{v})$ is found incompatible 
(\ie $f_{\hat{u}} < \Delta_{\hat{u},\hat{v}}$ at line~\ref{algo:jvalue:l14}, after $\iota'$).
Since $\hat{u}\in V_0\setminus {L^{\text{inc}}}^{\iota'}$, and $\J.\texttt{cmp}^{\iota'}[(\hat{u},\hat{v})]=\texttt{T}$,  
then at some step $\iota''>\iota'$ the counter $\J.\texttt{cnt}^{\iota''}$ is decremented by one unit 
and therefore $\J.\texttt{cmp}^{\iota''}[(u,v)]\leftarrow\texttt{F}$ is assigned 
(at lines~\ref{algo:jvalue:l16}-\ref{algo:jvalue:l17}). 
This proves that $\J.\texttt{cmp}[(\hat{u},\hat{v})]$ becomes coherent eventually (\ie at $\iota''$). 
Now, given $\hat{u}$, the same argument holds for any other $v\in N^{\text{out}}_{\Gamma}(\hat{u})$; 
therefore, when $\J.\texttt{cmp}[(\hat{u},v)]$ will finally become coherent for every $v\in N^{\text{out}}_{\Gamma}(\hat{u})$, 
then $\J.\texttt{cnt}[\hat{u}]$ will be coherent as well by (w-PC-3). 
Thus, by (w-PC-3), coherency of both $\J.\texttt{cnt}$ and $\J.\texttt{cmp}$ holds eventually, say at $\hat{\iota}$. 
At that point, all $u\in V_0$ that were inconsistent at $\iota$, 
or that have become inconsistent during the execution (up to $\hat{\iota}$), they necessarily have had to be inserted into 
$L^{\text{inc}}$ (at line~\ref{algo:jvalue:l18} of $\texttt{J-VI}()$, SubProcedure~\ref{algo:j-value-iteration}), 
because their (coherent) counter $\J.\texttt{cnt}[u]$ must reach $0$ (at $\hat{\iota}$), 
which allows $\texttt{J-VI}()$ to recognize $u$ as inconsistent at lines~\ref{algo:jvalue:l14}-\ref{algo:jvalue:l18}. 
Notice that the coherency of $\J.\texttt{cnt}$ and $\J.\texttt{cmp}$ is kept satisfied from $\hat{\iota}$ onwards: 
when some $v\in V$ is popped out of $L^{\text{inc}}$ (line~\ref{algo:jvalue:l2}), 
then $\J.\texttt{cnt}$ and $\J.\texttt{cmp}$ are recalculated from scratch (line~\ref{algo:jvalue:l7}), 
and it is easy to check that $\texttt{init\_cnt\_cmp}()$ (SubProcedure~\ref{subproc:counters_inc-arcs}) is correct; 
then $\J.\texttt{cnt}$, $\J.\texttt{cmp}$ may be modified subsequently, 
at lines~\ref{algo:jvalue:l16}-\ref{algo:jvalue:l17} (SubProcedure~\ref{algo:j-value-iteration}), 
but it's easy to check that lines~\ref{algo:jvalue:l14}-\ref{algo:jvalue:l17} preserve coherency;
so, coherency will be preserved until $\texttt{J-VI}()$ halts.

\item[$\hat{u}\in V_1$] and $(\hat{u},\hat{v})$ is incompatible \wrt $f^{\text{c}:\iota}$ in $\Gamma_{i,j}$: 

Then, by (w-PC-3), $\hat{v}\in {L^{\text{inc}}}^{\iota}$. 
As before, since the $\texttt{J-VI}()$ aims precisely at emptying $L^{\text{inc}}$, 
$\hat{v}$ is popped from $L^{\text{inc}}$ (line~\ref{algo:jvalue:l2} of SubProcedure~\ref{algo:j-value-iteration}); 
at some step of $\texttt{J-VI}()$'s execution. Soon after that, 
$N^{\text{in}}_{\Gamma}(\hat{v})$ is explored (lines~\ref{algo:jvalue:l11}-\ref{algo:jvalue:l18} 
of SubProcedure~\ref{algo:j-value-iteration}). As soon as $\hat{u}$ is visited, $(\hat{u},\hat{v})$ is found incompatible 
(\ie $f_{\hat{u}} < \Delta_{\hat{u},\hat{v}}$ at line~\ref{algo:jvalue:l14}).
Since $\hat{u}\in V_1\setminus L^{\text{inc}}$, then $\hat{u}$ is promptly inserted into 
$L^{\text{inc}}$ (line~\ref{algo:jvalue:l18}). In this way, all those $u\in V_1$ that were inconsistent 
at the time of $\texttt{J-VI}()$'s invocation, or that become inconsistent during the execution, 
they necessarily have had to be inserted into $L^{\text{inc}}$ (line~\ref{algo:jvalue:l18} of SubProcedure~\ref{algo:j-value-iteration}).
\end{itemize}
This analysis is already sufficient for asserting that Proposition~\ref{prop:correctness} holds, 
even assuming only (PC-1), (w-PC-3): indeed, the \texttt{Inv-JVI} invariant mentioned in its proof will hold, eventually, 
and then the Knaster-Tarski's Fixed-Point Theorem applies. This also proves Items (1) and (2).

Moreover, by (w-PC-2) and by arguments above, at each step $\bar\iota$ of $\texttt{J-VI}()$, 
if $v\in {L^{\text{inc}}}^{\bar\iota}$ then $v$ is really inconsistent \wrt $f^{\texttt{c}:\bar\iota}$ in $\Gamma_{i,j}$, \ie ${L^{\text{inc}}}^{\bar\iota}\subseteq \text{Inc}(f^{\text{c}:\bar\iota},i,j)$. 
Thus, every time that some $v$ is popped from $L^{\text{inc}}$ at line~\ref{algo:jvalue:l2}, 
then $\delta(f^{\text{c}},v)$ really increases $f^{\text{c}}(v)$ at line~\ref{algo:jvalue:l3}; therefore,  
$f^{\text{c}}(v)>0$ holds whenever $v$ is inserted into $L^{\text{inc}}_{\text{nxt}}$ at line~\ref{algo:jvalue:l6} of $\texttt{J-VI}()$ (SubProcedure~\ref{algo:j-value-iteration}); 
this implies that Proposition~\ref{prop:jvalue_init} holds, assuming (PC-1), (w-PC-2), (w-PC-3), and proves Item (3).
To conclude, we show Item (4). Notice, $L_\top$ is modified only at line~\ref{algo:jvalue:l9} of $\texttt{J-VI}()$ (SubProcedure~\ref{algo:j-value-iteration});
in particular, some $v\in V$ is inserted into $L_\top$ at line~\ref{algo:jvalue:l9}, say at step $\hat\iota$, if and only if 
$f^{\text{c}:\hat\iota}(v)=\top$. Since the energy-levels can only increase during the execution of $\texttt{J-VI}()$, 
then $L^h_\top\subseteq V\setminus V_{f^{\text{c}:h}}=\{ u\in V\mid f^{\text{c}:h}(u)=\top \}$. 
Since at each step $\bar\iota$ of $\texttt{J-VI}()$ it holds ${L^{\text{inc}}}^{\bar\iota}\subseteq \text{Inc}(f^{\text{c}:\bar\iota},i,j)$, 
then whenever some $v\in V$ is inserted into $L_\top$ at line~\ref{algo:jvalue:l9},  
it must be that $f^{\text{c}:\iota}(v) < \top$ where $\iota$ is the invocation time (otherwise, $v$ would not have been inconsistent at step $\bar\iota$); 
thus, $L^h_\top\subseteq V_{f^{\text{c}:\iota}}=\{v\in V \mid f^{\text{c}:\iota}(v)<\top\}$.
Therefore, $L^h_\top\subseteq V_{f^{\text{c}:\iota}}\cap V\setminus V_{f^{\text{c}:h}}$. 
Vice versa, let $v\in V_{f^{\text{c}:\iota}}\cap V\setminus V_{f^{\text{c}:h}}$;  
the only way in which $\texttt{J-VI}()$ can increase the energy-level of $v$ from step $\iota$ 
to step $h$ is by applying $\delta(f^{\text{c}}, v)$ at line~\ref{algo:jvalue:l3}; as soon as $f^{\text{c}}(v)=\top$ (and this will happen, eventually, 
since $v\in V_{f^{\text{c}:\iota}}\cap V\setminus V_{f^{\text{c}:h}}$), then $v$ is inserted into $L_\top$ at line~\ref{algo:jvalue:l9}.
Thus, $V_{f^{\text{c}:\iota}}\cap V\setminus V_{f^{\text{c}:h}}\subseteq L^h_\top$. 
Therefore, $L^h_\top = V_{f^{\text{c}:\iota}}\cap V\setminus V_{f^{\text{c}:h}}$; and this proves Item (4).
\end{proof}

\subsubsection{Correctness of EI-Jump (SubProcedure~\ref{subproc:energy_jumps})}
To begin, it is worth asserting some preliminary properties of $\texttt{ei-jump}()$ (SubProcedure~\ref{subproc:energy_jumps}).
\begin{Lem}\label{lem:ei-jump_state}
Assume $\texttt{ei-jump}(i,\J)$ (SubProcedure~\ref{subproc:energy_jumps}) 
is invoked by Algorithm~\ref{algo:solve_mpg} at line~\ref{algo:solve:l7}, 
say at step $\iota$, and for some $i\in [W^--1, W^+]$ (\ie for $i=i^\iota$). 
Assume ${L^{\text{inc}}}^\iota=\emptyset$ and $L^\iota_{\omega}\neq\emptyset$; 
and say that $\texttt{ei-jump}(i,\J)$ halts at step $h$. 
Then, the following two properties hold. 
\begin{enumerate}
\item The front element $(\bar{w}, L_{\alpha})$ of $L^{\iota}_{\omega}$ satisfies $\bar{w}=\min\{w_e \mid e \in E, w_e > i\}$; 
\item It holds that $\J.i^{h} \geq \bar{w} > i$.
\end{enumerate} 
\end{Lem}  
\begin{proof}
At the first invocation of $\texttt{ei-jump}(i,\J)$ (SubProcedure~\ref{subproc:energy_jumps}), 
made at line~\ref{algo:solve:l7} of Algorithm~\ref{algo:solve_mpg},  it holds $i=W^--1$ (by line~\ref{algo:solve:l5} of Algorithm~\ref{algo:solve_mpg}). 
Since ${L^{\text{inc}}}^\iota=\emptyset$, then $\texttt{ei-jump}()$ first assigns $\J.i\leftarrow i+1=W^-$ at line~\ref{subproc:energy_jump:l3}. 
Since $L_{\text{\bf{w}}}$ was sorted in increasing order at line~\ref{proc:jumper_init:l12} of $\texttt{init\_jumper}()$ (SubProcedure~\ref{proc:jumper_init}), 
the front entry of $L_{\text{\bf{w}}}$ has key $w=W^-$, and all of the subsequent entries of $L_{\text{\bf{w}}}$ are binded to greater keys. 
Actually, $\texttt{ei-jump}()$ consumes the front entry $(W^-, L_\alpha)$ of $L_{\text{\bf{w}}}$ at line~\ref{subproc:energy_jump:l7};
and $W^-$ is assigned to $\J.i$ (line~\ref{subproc:energy_jump:l3}). These observations imply both Item~1 and Item~2.
Now, consider any invocation of $\texttt{ei-jump}(i, \J)$ (SubProcedure~\ref{subproc:energy_jumps}) which is not the first, but any subsequent one. 
Let us check that the front element $(\bar{w}, L_{\alpha})$ of $L^{\iota}_{\text{\bf{w}}}$ satisfies $\bar{w}=\min\{w_e \mid e \in E, w_e > i^\iota\}$.
Consider each line of Algorithm~\ref{algo:solve_mpg} at which the value $i^\iota$ could have ever been assigned to $i$; this may happen only as follows: 

-- At line~\ref{subproc:energy_jump:l3} of $\texttt{ei-jump}()$ (SubProcedure~\ref{subproc:energy_jumps}), 
\ie $\J.i\leftarrow i+1$ ($=i^\iota$). But then the front element $(\hat w, L_{\alpha})$ of $L_{\text{\bf{w}}}$ 
is also checked at lines~\ref{subproc:energy_jump:l5}-\ref{subproc:energy_jump:l6} (because $L_{\text{\bf{w}}}\neq\emptyset$): 
and $\hat{w}$ is popped from $L_{\text{\bf{w}}}$ at line~\ref{subproc:energy_jump:l7}, in case $\hat w=\J.i$ ($=i^\iota$) holds at line~\ref{subproc:energy_jump:l6}. 

-- The same happens at lines~\ref{subproc:unitary_join:l2}-\ref{subproc:unitary_join:l5} of 
$\texttt{rejoin\_ua-jump}()$ (SubProcedure~\ref{subproc:ua-jumps}); 
just notice that in that case $i$ was incremented just before at line~\ref{subproc:ua-jumps:l4} of $\texttt{ua-jumps}()$ (SubProcedure~\ref{subproc:ua-jumps}).

-- At lines~\ref{subproc:energy_jump:l9}-\ref{subproc:energy_jump:l10} of $\texttt{ei-jump}()$ (SubProcedure~\ref{subproc:energy_jumps}), 
whenever the front element $(\hat w, L_{\alpha})$ of $L_{\text{\bf{w}}}$ is popped, then $\J.i\leftarrow \hat w$ is assigned.

Therefore, in any case, the following holds: 

When the variable $i$ got any of its possible values, say $\hat{i}$ (including $i^{\iota}$), the front entry $(\hat{w}, L_{\alpha})$ of $L_{\text{\bf w}}$ had always been checked, 
and then popped from $L_{\text{\bf{w}}}$ if $\hat{w}=\hat{i}$.

Recall, $L_{\text{\bf{w}}}$ was sorted in increasing order at line~\ref{proc:jumper_init:l12} of $\texttt{init\_jumper}()$ (SubProcedure~\ref{proc:jumper_init}).

Therefore, when $\texttt{ei-jump}(i^\iota,\J)$ is invoked at step $\iota$, all of the entries 
$(w,L_{\alpha})$ of $L_{\text{\bf{w}}}$ such that $w\leq i^\iota$ must already have been popped from $L_{\text{\bf w}}$ before step $\iota$. 

Therefore, $\bar{w} = \min\{w_e \mid e\in E, w_e > i^\iota\}$, 
if $\bar{w}$ is the key (weight) of the front entry of $L^{\iota}_{\text{\bf{w}}}$. 

Next, since ${L^{\text{inc}}}^\iota=\emptyset$ and $L^{\iota}_{\text{\bf{w}}}\neq\emptyset$ by hypothesis, 
and by line~\ref{subproc:energy_jump:l9} of $\texttt{ei-jump}()$, 
at least one further element $(w, L_{\alpha})$ of $L_{\text{\bf{w}}}$ must be popped from $L^{\iota}_{\text{\bf{w}}}$,
either at line~\ref{subproc:energy_jump:l7} or line~\ref{subproc:energy_jump:l10} of $\texttt{ei-jump}()$, soon after $\iota$.
Consider the last element, say $w'$, which is popped after $\iota$ and before $h$. 
Then, $\J.i^{h}\leftarrow w'$ is assigned either at line~\ref{subproc:energy_jump:l3} 
or line~\ref{subproc:energy_jump:l11} of $\texttt{ei-jump}()$. 
Notice, $w' \geq \bar{w} > i^\iota$. Thus, $\J.i^{h} \geq \bar{w} > i^\iota$. 
\end{proof}

The following proposition essentially asserts that $\texttt{ei-jump}()$ (SubProcedure~\ref{subproc:energy_jumps}) is correct. 
To begin, notice that, when $\texttt{ei-jump}(i,\J)$ is invoked at line~\ref{algo:solve:l7} of Algorithm~\ref{algo:solve_mpg}, 
then $i\in [W^--1,W^+]$. Also recall that any invocation of $\texttt{ei-jump}(i,\J)$ halts 
in finite time by Proposition~\ref{prop:ei-jump_halts}. 
\begin{Prop}\label{prop:ei-jump_correctness}
Consider any invocation of $\texttt{ei-jump}(i,\J)$ (SubProcedure~\ref{subproc:energy_jumps}) that is made 
at line~\ref{algo:solve:l7} of Algorithm~\ref{algo:solve_mpg}, say at step $\iota$, and for some $i\in [W^--1, W^+]$. 
Further assume that ${L^{\text{inc}}}^{\iota}=\emptyset$ and that $\texttt{ei-jump}()$ halts at step $h$. 

Suppose the following pre-conditions are all satisfied at invocation time $\iota$, for $s=|\F_{|V|}|$:
\begin{enumerate}
\item[(eij-PC-1)] $f^{\text{c}:\iota}$ is the least-SEPM of $\Gamma_{i,s-1}$; thus, $\text{Inc}(f^{\text{c}:\iota}, i, s-1)=\emptyset$. Also, $L^\iota_f=\emptyset$.
\item[(eij-PC-2)] $\{v\in V\mid 0 < f^{\text{c}:\iota}(v)\neq \top\}=\emptyset$;
\item[(eij-PC-3)] ${L^{\text{inc}}_{\text{cpy}}}^\iota \subseteq 
			\text{Inc}(f^{\text{c}:\iota}, i', j')$ for every $(i',j')>(i,s-1)$; 
\item[(eij-PC-4)] $\J.\texttt{cnt}^{\iota}$ and $\J.\texttt{cmp}^{\iota}$ are both coherent \wrt $f^{\text{c}:\iota}$ in $\Gamma_{i,s-1}$. 
\end{enumerate}  

Finally, let $i'\in [W^-,W^+]$, $j'\in [1,s-1]$ be any indices such that $(i,s-1) < (i',j') \leq (\J.i^{h},1)$. 
Then, the following holds. 
\begin{enumerate}
\item Suppose that $L^{\iota}_{\text{\bf{w}}}\neq\emptyset$. Let $(\hat{w}, L_{\hat\alpha})$ be any entry of $L^{\iota}_{\text{\bf{w}}}$ such that
	$\hat{w}=\J.i^{\iota'}=i'$ holds either at line~\ref{subproc:energy_jump:l6} or line~\ref{subproc:energy_jump:l11} of $\texttt{ei-jump}(i,\J)$, for some step $\iota'>\iota$. 
	When the $\texttt{repair}(L_{\hat\alpha}, \J)$ halts soon after, either at line~\ref{subproc:energy_jump:l8} or \ref{subproc:energy_jump:l12} (respectively), 
	say at some step $\iota''>\iota'$, both $\J.\texttt{cnt}^{\iota''}$ and $\J.\texttt{cmp}^{\iota''}$ are coherent \wrt $f^{\text{c}:\iota''}$ ($=f^{\text{c}:\iota}$) in $\Gamma_{i',j'}$.

\item If $(i',j')<(\J.i^h,1)$, then $\text{Inc}(f^{\text{c}:\iota},i',j')=\emptyset$; 

\item It holds that either ${L^{\text{inc}}}^{h}\neq\emptyset$ or both ${L^{\text{inc}}}^{h}=\emptyset$ and $L^{h}_{\omega}=\emptyset$.

Anyway, ${L^{\text{inc}}}^{h} = \text{Inc}(f^{\text{c}:h}, i^{h}, 1)$.
\end{enumerate}
\end{Prop} 
Notice that $f^{\text{c}}$ stands still during $\texttt{ei-jump}()$ (SubProcedure~\ref{subproc:energy_jumps}), \ie $f^{\text{c}:\iota}=f^{\text{c}:\iota'}=f^{\text{c}:\iota''}=f^{\text{c}:h}$, 
for steps $\iota, \iota', \iota'', h$ defined as in Proposition~\ref{prop:ei-jump_correctness}. In the proofs below, we can simply refer to $f^{\text{c}}$.
 
\begin{proof}[Proof of Item (1)] Let $u\in V_0$ and $v\in N_{\Gamma}^{\text{out}}(u)$, 
let $i',j'$ be fixed indices such that $(i,s-1) < (i',j') \leq (\J.i^h,1)$. 
By (\textit{eij-PC-2}), either $f^{\text{c}}(u)=\top$ or $f^{\text{c}}(u)=0$, either $f^{\text{c}}(v)=\top$ or $f^{\text{c}}(v)=0$. 
\begin{itemize}
\item If $f^{\text{c}}(u)=\top$, then $(u,v)\in E$ is compatible \wrt $f^{\text{c}}$ in $\Gamma_{i,s-1}$. 
So, $\J.\texttt{cmp}^{\iota}[(u,v)]=\texttt{T}$ holds by (\textit{eij-PC-4}).
Since $f^{\text{c}}(u)=\top$, $\texttt{ei-jump}()$ can't modify $\J.\texttt{cmp}[(u,v)]$; 
see line~\ref{subproc:repair:l2} of $\texttt{repair}()$ (SubProcedure~\ref{subproc:energy_jumps}).
So, $\J.\texttt{cmp}^{\iota''}[(u,v)]=\texttt{T}$ is still coherent \wrt $f^{\text{c}}$ in $\Gamma_{i',j'}$.
	
\item If $f^{\text{c}}(u)=0$ and $f^{\text{c}}(v)=\top$, then $(u,v)\in E$ is incompatible \wrt $f^{\text{c}}$ in $\Gamma_{i,s-1}$. 
So, $\J.\texttt{cmp}^{\iota}[(u,v)]=\texttt{F}$ holds by (\textit{eij-PC-4});
and it will hold for the whole execution of $\texttt{ei-jump}()$, because $\texttt{ei-jump}()$ never changes $\J.\texttt{cmp}$ from \texttt{F} to \texttt{T}.  
Thus, $\J.\texttt{cmp}^{\iota''}[(u,v)]=\texttt{F}$ is still coherent \wrt $f^{\text{c}}$ in $\Gamma_{i',j'}$.

\item Assume $f^{\text{c}}(u)=0$ and $f^{\text{c}}(v)=0$. 

Again, $\J.\texttt{cmp}^{\iota}[(u,v)]$ is coherent \wrt $f^{\text{c}}$ in $\Gamma_{i,s-1}$ by (\textit{eij-PC-4}). We have two cases: 
\begin{itemize}
\item If $\J.\texttt{cmp}^{\iota}[(u,v)]=\texttt{F}$, then $(u,v)$ is incompatible \wrt $f^{\text{c}}$ 
in $\Gamma_{i,s-1}$, \ie $f^{\text{c}}(u) < f^{\text{c}}(v)-(w(u,v)-i-F_{s-1})$.
Since $f^{\text{c}}(u)=f^{\text{c}}(v)=0$ and $F_{s-1}=1$, then: \[ 0=f^{\text{c}}(u) < f^{\text{c}}(v)-w(u,v)+i+F_{s-1} = -w(u,v)+i+1.\]
Therefore, $w(u,v)\leq i$, because $w(u,v)\in\Z$. Since $(i,s-1)<(i',j')$, then $i < i'$, so $w(u,v) < i'$; 
this means that $(u,v)$ is still incompatible \wrt $f^{\text{c}}$ in $\Gamma_{i',j'}$.
Meanwhile, $\J.\texttt{cmp}[(u,v)]=\texttt{F}$ stands still (because $\texttt{ei-jump}()$ never changes $\J.\texttt{cmp}$ from \texttt{F} to \texttt{T}), 
therefore, $\J.\texttt{cmp}^{\iota''}[(u,v)]=\texttt{F}$.

\item If $\J.\texttt{cmp}^{\iota}[(u,v)]=\texttt{T}$, then $(u,v)$ is compatible \wrt $f^{\text{c}}$ in $\Gamma_{i,s-1}$ by (\textit{eij-PC-4}). So,  
\begin{itemize} 
\item If $(u,v)$ is compatible \wrt $f^{\text{c}}$ in $\Gamma_{i',j'}$, \ie $f^{\text{c}}(u)\geq f^{\text{c}}(v) - (w(u,v) - i' - F_{j'})$, then, since $f^{\text{c}}(u)=f^{\text{c}}(v)=0$, we have: 
\[ 0=f^{\text{c}}(u) \geq f^{\text{c}}(v)-w(u,v)+i' + F_{j'} = -w(u,v)+i'+F_{j'}. \]
Then, $w(u,v)>i'$, because $j'\in [1,s-1]$ (so, $F_{j'}>0$) and $w(u,v)\in\Z$. Consider what happens in $\texttt{ei-jump}()$ at $\iota'$. 
Since $L_{\text{\bf{w}}}$ was sorted in increasing order, and since $w(u,v)>i'$, then the entry $(w(u,v), L_{\alpha})$ 
is still inside $L_{\text{\bf{w}}}$ at $\iota'$ (indeed, at step $\iota'$, the front entry of $L_{\text{\bf{w}}}$ has key value $i'$ by hypothesis). 
Therefore, neither the subsequent invocation of $\texttt{repair}()$ (line~\ref{subproc:energy_jump:l8} or line~\ref{subproc:energy_jump:l12} of $\texttt{ei-jump}()$), 
nor any of the previous invocations of $\texttt{repair}()$ (before $\iota'$), can alter the state of $\J.\texttt{cmp}([u,v])$ from $\texttt{T}$ to $\texttt{F}$, 
just because $(u,v)\in L_\alpha$ is still inside $L_{\text{\bf{w}}}$ at $\iota'$; so, $\J.\texttt{cmp}([u,v])=\texttt{T}$ stands still, thus, $\J.\texttt{cmp}^{\iota''}[(u,v)]=\texttt{T}$. 

\item If $(u,v)\in E$ is incompatible \wrt $f^{\text{c}}$ in $\Gamma_{i',j'}$, \ie $f^{\text{c}}(u) < f^{\text{c}}(v)-(w(u,v)-i'-F_{j'})$, then, since $f^{\text{c}}(u)=f^{\text{c}}(v)=0$, 
we have: \[ 0=f^{\text{c}}(u) < f^{\text{c}}(v) - w(u,v) + i' + F_{j'} = - w(u,v) + i' + F_{j'}. \]
Thus, $w(u,v)\leq i'$, because $f^{\text{c}}(u)=f^{\text{c}}(v)=0$ and $j'\in [1,s-1]$ (so $F_{j'}>0$). 
On the other side, since $(u,v)\in E$ is compatible \wrt $f^{\text{c}}$ in $\Gamma_{i,s-1}$, 
at this point the reader can check that $w(u,v)>i$. Then, by Item~1 of Lemma~\ref{lem:ei-jump_state}, and since $L_{\text{\bf{w}}}$ was sorted in increasing order, 
the entry $(w(u,v), L_\alpha)$ is still inside $L_{\text{\bf{w}}}$ at $\iota$. Therefore, since $w(u,v)\leq i'$, there must be some step $\hat{\iota}$ (such that $\iota < \hat{\iota} \leq \iota'$)   
at which the entry $(w(u,v), L_{\alpha})$ must have been considered, either at line~\ref{subproc:energy_jump:l5} or 
line~\ref{subproc:energy_jump:l10} of $\texttt{ei-jump}(i,\J)$, and thus popped from $L_{\text{\bf{w}}}$.  
Soon after $\hat{\iota}$, the subsequent invocation of $\texttt{repair}()$ (either at line~\ref{subproc:energy_jump:l8} or line~\ref{subproc:energy_jump:l12} of $\texttt{ei-jump}()$) 
changes the state of $\J.\texttt{cmp}^{\hat\iota}[(u,v)]$ from $\texttt{T}$ to $\texttt{F}$ (line~\ref{subproc:repair:l5} of $\texttt{repair}()$), 
and it decrements $\J.\texttt{cnt}^{\hat\iota}[(u,v)]$ by one unit (line~\ref{subproc:repair:l4} of $\texttt{repair}()$). 
Thus $\J.\texttt{cmp}$ gets \emph{repaired} so that to be coherent \wrt $f^{\text{c}}$ in $\Gamma_{w(u,v),j'}$.
Now, by Item~2 of Lemma~\ref{lem:ei-jump_state}, $\J.i$ can only increase during the execution of $\texttt{ei-jump}()$. 
So, from that point on, $(u,v)$ remains incompatible \wrt $f^{\text{c}}$ in $\Gamma_{\J.i,j'}$ for every $w(u,v) \leq \J.i\leq i'$. On the other hand, 
$\J.\texttt{cmp}^{\hat\iota}[(u,v)]=\texttt{F}$ stands still, since $\texttt{ei-jump}()$ (SubProcedure~\ref{subproc:energy_jumps}) never changes it from \texttt{F} to \texttt{T}. 
So, $\J.\texttt{cmp}^{\iota''}[(u,v)]=\texttt{F}$.
\end{itemize}
This proves that, in any case, $\J.\texttt{cmp}^{\iota''}[(u,v)]$ is coherent \wrt $f^{\text{c}}$ in $\Gamma_{i',j'}$.
\end{itemize}
This also proves that $\J.\texttt{cnt}^{\iota''}$ is coherent \wrt $f^{\text{c}}$ in $\Gamma_{i',j'}$: 
indeed, $\J.\texttt{cnt}^{\iota}$ was coherent \wrt $f^{\text{c}}$ in $\Gamma_{i,s-1}$ by (\textit{eij-PC-4}); 
then $\J.\texttt{cnt}$ was decremented by one unit (line~\ref{subproc:repair:l4} of $\texttt{repair}()$) 
each time that $\J.\texttt{cmp}$ was repaired (line~\ref{subproc:repair:l5} of $\texttt{repair}()$), as described above; 
therefore, at step $\iota''$, the coherency of $\J.\texttt{cnt}^{\iota''}$ follows by that of $\J.\texttt{cmp}^{\iota''}$.
\end{itemize}
\end{proof}

\begin{proof}[Proof of Item (2)]
Let $u\in V$. We want to prove that $u$ is consistent \wrt $f^{\text{c}}$ in $\Gamma_{i',j'}$ for every $i'\in [W^-,W^+]$ and $j'\in [1,s-1]$ such that $(i,s-1)<(i',j')<(\J.i^h,1)$ (if any). 

By (\textit{eij-PC-2}), either $f^{\text{c}}(u)=0$ or $f^{\text{c}}(u)=\top$. If $f^{\text{c}}(u)=\top$, the claim holds trivially. 
Assume $f^{\text{c}}(u)=0$. By (\textit{eij-PC-1}), $u$ is consistent \wrt $f^{\text{c}}$ in $\Gamma_{i,s-1}$. 
Assume $\hat{w}=\J.i^{\iota'}=i'$, either at line~\ref{subproc:energy_jump:l6} or line~\ref{subproc:energy_jump:l11} of $\texttt{ei-jump}(i,\J)$, for some step $\iota'>\iota$. 
Assume $\texttt{repair}(L_{\hat\alpha}, \J)$ halts soon after at line~\ref{subproc:energy_jump:l8} or \ref{subproc:energy_jump:l12} (respectively), for some step $\iota''$ where $\iota''>\iota'$. 
We claim $u$ is consistent \wrt $f^\text{c}$ in $\Gamma_{i',j'}$.
\begin{itemize} \item[If $u\in V_0$,]  
By (\textit{eij-PC-4}), $\J.\texttt{cnt}^{\iota}[u]$ is coherent \wrt $f^{\text{c}}$ in $\Gamma_{i,s-1}$. 
Thus, since $u$ is consistent \wrt $f^{\text{c}}$ in $\Gamma_{i,s-1}$, it holds that $\J.\texttt{cnt}^{\iota}[u]>0$. 
Now, since $(i',j')<(\J.i^h,1)$, then $i'<\J.i^h$, thus ${L^{\text{inc}}}^{\iota''}=\emptyset$. 
Therefore, $\J.\texttt{cnt}^{\iota''}[u]>0$ (otherwise $u$ would have been inserted into $L^{\text{inc}}$ within $\iota''$ at line~\ref{subproc:repair:l7} of $\texttt{repair}()$). 
By Item~1 of Proposition~\ref{prop:ei-jump_correctness}, $\J.\texttt{cnt}^{\iota''}[u]$ is coherent \wrt $f^{\text{c}}$ in $\Gamma_{i',j'}$.
Since $\J.\texttt{cnt}^{\iota''}[u]>0$ and $\J.\texttt{cnt}[u]^{\iota''}$ is coherent, $u$ is consistent \wrt $f^\text{c}$ in $\Gamma_{i',j'}$. 

\item[If $u\in V_1$,] Since $(i',j') < (\J.i^h,i)$, then $i'<\J.i^h$, thus ${L^{\text{inc}}}^{\iota''}=\emptyset$. 
Let $v\in N_{\Gamma}^{\text{out}}(u)$. By (\textit{eij-PC-2}), either $f^{\text{c}}(v)=0$ or $f^{\text{c}}(v)=\top$. Since $f^{\text{c}}(u)=0$ by assumption, $u\in V_1$, 
and $u$ is consistent \wrt $f^{\text{c}}$ in $\Gamma_{i,s-1}$ by (\textit{eij-PC-1}), then $f^{\text{c}}(v)=0$.

Now, we argue that $w(u,v)>i'$.

Assume, for the sake of contradiction, that $w(u,v)\leq i'$. On one side, since $u$ is consistent \wrt $f^{\text{c}}$ in $\Gamma_{i,s-1}$ and since $f^{\text{c}}(u)=f^{\text{c}}(v)=0$, 
then: \[0=f^{\text{c}}(u)\geq f^{\text{c}}(v) - w(u,v) + i + F_{s-1} = - w(u,v) + i + 1 .\]
Thus, $w(u,v)\geq i+1 > i$. Therefore, by Lemma~\ref{lem:ei-jump_state}, the entry $(w(u,v), L_{\alpha})$ is still inside $L_{\text{\bf{w}}}$ at step $\iota$.
On the other side, since $w(u,v)\leq i'$, it is easy to see at this point that within (or soon after) step $\iota'$ the entry 
$(w(u,v), L_{\alpha'})$ must have been popped from $L_{\text{\bf{w}}}$ either at line~\ref{subproc:energy_jump:l7} or at line~\ref{subproc:energy_jump:l10} of $\texttt{ei-jump}()$ (SubProcedure~\ref{subproc:energy_jumps}).  
So, $(w(u,v), L_{\alpha'})$ must have been popped from $L_{\text{\bf{w}}}$ after $\iota$ and within $\iota'$ (or soon after $\iota'$ at line~\ref{subproc:energy_jump:l7}). 
But soon after that, since $u\in V_1$, the subsequent invocation of $\texttt{repair}()$ would insert $u$ into $L^{\text{inc}}$ at line~\ref{subproc:repair:l9}, 
because $f^{\text{c}}(u)=f^{\text{c}}(v)=0$ and $L^{\text{inc}}[u]=\emptyset$ at line~\ref{subproc:repair:l2}.
Therefore ${L^{\text{inc}}}^{\iota''}\neq\emptyset$, which is a contradiction. Therefore, $w(u,v)>i'$. 
Since $v\in N_{\Gamma}^{\text{out}}(u)$ was chosen arbitrarily, $\forall^{v\in N_{\Gamma}^{\text{out}}(u)}$ $w(u,v)>i'$.
 
Since $\forall^{v\in N_{\Gamma}^{\text{out}}(u)}$ $w(u,v)>i'$ and $f^{\texttt{c}}(u)=f^{\texttt{c}}(v)=0$, 
then $u$ is consistent \wrt $f^{\text{c}}$ in $\Gamma_{i',j'}$.
\end{itemize}
So, $u$ is consistent \wrt $f^{\text{c}}$ in $\Gamma_{i',j'}$. Since $u\in V$ was chosen arbitrarily, then $\text{Inc}(f^{\text{c}:\iota},i',j')=\emptyset$.
\end{proof} 

\begin{proof}[Proof of Item (3)] 
By line~\ref{subproc:energy_jump:l9} of $\texttt{ei-jump}()$ (SubProcedure~\ref{subproc:energy_jumps}), 
when the \texttt{while} loop at lines~\ref{subproc:energy_jump:l9}-\ref{subproc:energy_jump:l12} halts, then 
$\texttt{ei-jump}()$ halts soon after at line~\ref{subproc:energy_jump:l13}, say at step $h$, and it must be that 
either ${L^{\text{inc}}}^{h}\neq\emptyset$ or both ${L^{\text{inc}}}^h=\emptyset$ and $L^h_{\text{\bf{w}}}=\emptyset$.
Now, we want to prove that ${L^{\text{inc}}}^h=\text{Inc}(f^{\text{c}}, \J.i^h, 1)$.
\begin{itemize}
	\item Firstly, ${L^{\text{inc}}}^h\subseteq \text{Inc}(f^{\text{c}}, \J.i^h, 1)$:
\end{itemize}
Assume $u\in {L^{\text{inc}}}^h$. We have three cases to check: 

-- If $u \in {L^{\text{inc}}_{\text{cpy}}}^\iota$, notice that $\J.i^h>i$ by Lemma~\ref{lem:ei-jump_state}, 
	then $u\in\text{Inc}(f^{\text{c}}, \J.i^h, 1)$ holds by (eij-PC-3);

-- If $u \in V_0\setminus {L^{\text{inc}}_{\text{cpy}}}^\iota$, then $\J.\texttt{cnt}^{h}[u]=0$ 
	by lines~\ref{subproc:repair:l6}-\ref{subproc:repair:l7} of $\texttt{repair}()$ (SubProcedure~\ref{subproc:energy_jumps}).
	By Item~1 of Proposition~\ref{prop:ei-jump_correctness}, $\J.\texttt{cnt}^{h}[u]$ is coherent \wrt $f^{\text{c}}$ in $\Gamma_{\J.i^h,1}$.
	Therefore, $u\in\text{Inc}(f^{\text{c}}, \J.i^h, 1)$.

-- If $u \in V_1\setminus {L^{\text{inc}}_{\text{cpy}}}^\iota$, then $\exists^{v\in N_{\Gamma}^{\text{out}}(u)}$ $f^{\text{c}}(u)=f^{\text{c}}(v)=0$ 
	and $w(u,v)=\J.i^h$ by lines~\ref{subproc:repair:l2}-\ref{subproc:repair:l8} of $\texttt{repair}()$. Therefore, $u\in\text{Inc}(f^{\text{c}}, \J.i^h, 1)$. 

	This proves, ${L^{\text{inc}}}^h\subseteq \text{Inc}(f^{\text{c}}, \J.i^h, 1)$.
\begin{itemize}
\item Secondly, ${L^{\text{inc}}}^h\supseteq \text{Inc}(f^{\text{c}}, \J.i^h, 1)$:  
\end{itemize}
Let $u \in \text{Inc}(f^{\text{c}}, \J.i^h, 1)$. By (\textit{eij-PC-2}), either $f^{\text{c}}(u)=0$ or $f^{\text{c}}(u)=\top$. 
Since $u\in \text{Inc}(f^{\text{c}}, \J.i^h, 1)$, then $f^{\text{c}}(u)=0$. Now, 
by (\textit{eij-PC-1}), $u$ is consistent \wrt $f^{\text{c}}$ in $\Gamma_{i,s-1}$.  
So, let $(u,\hat{v})\in E$ be any arc which is compatible \wrt $f^{\text{c}}$ in $\Gamma_{i,s-1}$.
\begin{itemize}
\item[If $u\in V_0$, ] 
then, \emph{at least one} such a compatible $\hat{v}\in N_{\Gamma}^{\text{out}}(u)$ exists (because $u\in V_0$ is consistent \wrt $f^{\text{c}}$ in $\Gamma_{i,s-1}$). 
By (\textit{eij-PC-2}), either $f^{\text{c}}(\hat{v})=0$ or $f^{\text{c}}(\hat{v})=\top$. 
Since $f^{\text{c}}(u)=0$ and $(u,\hat{v})$ is compatible \wrt $f^{\text{c}}$ in $\Gamma_{i,s-1}$, 
then $f^{\text{c}}(\hat{v})=0$. Since $f^{\text{c}}(u)=f^{\text{c}}(\hat{v})=0$ and $u\in \text{Inc}(f^{\text{c}}, \J.i^h, 1)$, then $w(u,\hat{v})\leq \J.i^h$.

We claim that at some step of execution of line~\ref{subproc:energy_jump:l7} or line~\ref{subproc:energy_jump:l10} 
in $\texttt{ei-jump}()$ (SubProcedure~\ref{subproc:energy_jumps}), say at step $\iota'$ for $\iota < \iota' < h$, 
the entry $(w(u,\hat{v}), L_\alpha)$ is popped from $L_{\text{\bf{w}}}$.

Since $f^{\text{c}}(u)=f^{\text{c}}(\hat{v})=0$, and $(u,\hat{v})$ is compatible \wrt $f^{\text{c}}$ in $\Gamma_{i,s-1}$, 
then $w(u,\hat{v}) > i$. Thus, by Lemma~\ref{lem:ei-jump_state}, when $\texttt{ei-jump}()$ is invoked (\ie at step $\iota$), 
the front entry $(\bar{w}, L_{\bar\alpha})$ of $L^{\iota}_{\text{\bf{w}}}$ 
satisfies $\bar{w}=\min\{w_e\mid e\in E, w_e>i\}\leq w(u,\hat{v})$. 
So, $\bar{w}\leq w(u,\hat{v})\leq \J.i^h$. Thus, at some step of execution $\iota'$ for $\iota < \iota' < h$, 
the entry $(w(u,\hat{v}), L_\alpha)$ must be popped from $L_{\text{\bf{w}}}$.

Soon after that, $\texttt{repair}(L_{\alpha}, \J)$ is invoked: there, since $f^{\text{c}}(u)=f^{\text{c}}(\hat{v})=0$ and $u\in V_0$, 
then $\J.\texttt{cnt}[u]$ is decremented by one unit at line~\ref{subproc:repair:l4} of $\texttt{repair}()$.

Indeed, this happens (after $\iota$ but before $h$),  
\emph{for every} $v\in N_{\Gamma}^{\text{out}}(u)$ such that $(u, v)\in E$ is compatible \wrt $f^{\text{c}}$ in $\Gamma_{i,s-1}$. 
Thus, eventually and before $h$, it will hold $\J.\texttt{cnt}[u]=0$. 
At that point, $u$ will be inserted into $L^{\text{inc}}$ at line~\ref{subproc:repair:l7} of $\texttt{repair}()$; and soon after, 
$\texttt{ei-jump}()$ halts (since $L^{\text{inc}}\neq\emptyset$ at line~\ref{subproc:energy_jump:l9}).
So, $u\in {L^{\text{inc}}}^h$. This holds for every $u\in \text{Inc}(f^{\text{c}}, \J.i^h, 1)\cap V_0$. 
Thus, $\text{Inc}(f^{\text{c}}, \J.i^h, 1)\cap V_0\subseteq  {L^{\text{inc}}}^h$.

\item[If $u\in V_1$, ] then, \emph{all} $\hat{v}\in N_{\Gamma}^{\text{out}}(u)$ are such 
that $(u, \hat{v})$ is compatible \wrt $f^{\text{c}}$ in $\Gamma_{i,s-1}$, 
because $u\in V_1$ is consistent \wrt $f^{\text{c}}$ in $\Gamma_{i,s-1}$ by (\textit{eij-PC-1}). 
The argument proceeds almost in the same way as before.
By \textit{(eij-PC-2)}, either $f^{\text{c}}(\hat{v})=0$ or $f^{\text{c}}(\hat{v})=\top$. 
Since $f^{\text{c}}(u)=0$ and $(u,\hat{v})$ is compatible \wrt $f^{\text{c}}$ in $\Gamma_{i,s-1}$, 
then $f^{\text{c}}(\hat{v})=0$. Since $f^{\text{c}}(u)=f^{\text{c}}(\hat{v})=0$ and 
$u\in \text{Inc}(f^{\text{c}}, \J.i^h, 1)$, then $w(u,\hat{v})\leq \J.i^h$.
By arguing as above, we see that at some step of execution of line~\ref{subproc:energy_jump:l3} 
in (the considered invocation of) $\texttt{ei-jump}()$ (SubProcedure~\ref{subproc:energy_jumps}), 
say at step $\iota'$ for $\iota < \iota' < h$, the entry $(w(u,\hat{v}), L_\alpha)$ 
is popped from $L_{\text{\bf{w}}}$. Soon after that, $\texttt{repair}(L_{\alpha}, \J)$ is invoked: 
there, since $f^{\text{c}}(u)=f^{\text{c}}(\hat{v})=0$ and $u\in V_1$, 
then $u$ is inserted into $L^{\text{inc}}$; soon after, $\texttt{ei-jump}()$ halts (since $L^{\text{inc}}\neq\emptyset$); 
so, $u\in {L^{\text{inc}}}^h$. This holds for every $u\in  \text{Inc}(f^{\text{c}}, \J.i^h, 1) \cap V_1$; 
so, $ \text{Inc}(f^{\text{c}}, \J.i^h, 1) \cap V_1\subseteq {L^{\text{inc}}}^h$.
\end{itemize} 
Therefore, $ \text{Inc}(f^{\text{c}}, \J.i^h, 1) = {L^{\text{inc}}}^h$. 
\end{proof}

\subsubsection{Correctness of $\texttt{ua-jumps}()$ (SubProcedure~\ref{subproc:ua-jumps})}
\begin{Prop}\label{prop:ua-jumps_correctness}
Consider any invocation of $\texttt{ei-jump}()$ (SubProcedure~\ref{subproc:energy_jumps}) that is made 
at line~\ref{algo:solve:l7} of Algorithm~\ref{algo:solve_mpg}.
Assume that the pre-conditions (\textit{eij-PC-1}), (\textit{eij-PC-2}), (\textit{eij-PC-3}), (\textit{eij-PC-4}),  
are all satisfied at invocation time. Further assume that $L^\text{inc}\neq\emptyset$ at line~\ref{algo:solve:l8} 
of Algorithm~\ref{algo:solve_mpg}, so that $\texttt{ua-jumps}()$ is invoked soon after at line~\ref{algo:solve:l9}.
Then, consider any invocation of $\texttt{J-VI}(i,s-1,F,\J,\Gamma)$ (SubProcedure~\ref{algo:j-value-iteration}) 
that is made at line~\ref{subproc:ua-jumps:l2} of $\texttt{ua-jumps}()$ (SubProcedure~\ref{subproc:ua-jumps}), 
for some $i\in [W^-, W^+]$, where $s=|\F_{|V|}|$. Then, the following properties hold.  
\begin{enumerate}
\item The (\emph{PC-1}), \emph{(w-PC-2)}, \emph{(w-PC-3)} are all satisfied by that invocation of $\texttt{J-VI}(i,s-1,F,\J,\Gamma)$.
\item When the $\texttt{J-VI}(i,s-1,F,\J,\Gamma)$ halts, say at step $h$, then the following holds: 
	\[ L^{h}_{\top} = \W_0(\Gamma_{i-1,s-1}) \cap \W_1(\Gamma_{i,s-1}). \]
\item Assume that $\texttt{backtrack\_ua-jumps}(i,s,F,\J,\Gamma)$ is invoked at line~\ref{subproc:ua-jumps:l7} of $\texttt{ua-jumps()}$, say at step $\iota$, and assume that it halts at step $h$.
	\begin{enumerate}
		\item At line~\ref{subproc:ua-jumps_backtrack:l1} of  $\texttt{backtrack\_ua-jumps}(i,s,F,\J,\Gamma)$,	
			it holds: \[ L^{\text{inc}}_{\text{cpy}} = \{v\in V\mid 0 < f^{\text{c}:\iota}(v)\neq\top\}. \]
		\item Consider the two induced games $\Gamma[L^\iota_\top]$ and $\Gamma[V\setminus L^\iota_\top]$. 
			The following holds:
		\begin{enumerate}
			\item $\forall({v\in L^\iota_\top})\; f^{\text{c}:h}(v) = f^*_{w'_{i-1, s-1}}(v)$;
			\item $\forall({v\in V\setminus L^\iota_\top})\; f^{\text{c}:h}(v) = f^*_{w'_{i, s-1}}(v)$;
			\item $\forall({u\in L^\iota_\top\cap V_0}) \forall({v\in N^{\text{out}}_{\Gamma[L^{\iota}_\top]}(u)})\; 
				\J.\texttt{cmp}^h(u,v)$ is coherent \wrt $f^{\text{c}:h}$ in $\Gamma[L^{\iota}_\top]_{i,1}$;
			\item $\forall({v\in L^\iota_\top\cap V_0})\; 
				\J.\texttt{cnt}^h(v)$ is coherent \wrt $f^{\text{c}:h}$ in $\Gamma[L^{\iota}_\top]_{i,1}$;
			\item $\forall(j'\in [1, s-1])\; \text{Inc}(f^{\text{c}:h}, i, j')\setminus L^\iota_\top=\emptyset$.
		\end{enumerate} 
	\end{enumerate}
\item Any invocation of $\texttt{ua-jumps()}$ (SubProcedure~\ref{subproc:ua-jumps}) (line~\ref{algo:solve:l7}, Algorithm~\ref{algo:solve_mpg}) halts in finite time. 
\end{enumerate}
\end{Prop}
\begin{proof}[Proof of Item~(1)] 
By induction on the number $k\in \N$ of invocations of $\texttt{J-VI}()$ that are made at line~\ref{subproc:ua-jumps:l2} of $\texttt{ua-jumps}()$.

\textit{Base Case: $k=1$.} 
Consider the first invocation of $\texttt{J-VI}()$ at line~\ref{subproc:ua-jumps:l2} of $\texttt{ua-jumps}()$, 
say it happens at step $\iota$. 
Just before step $\iota$, Algorithm~\ref{algo:solve_mpg} invoked $\texttt{ei-jump}()$ at line~\ref{algo:solve:l7}. 
By hypothesis, (\textit{eij-PC-1}), (\textit{eij-PC-2}), (\textit{eij-PC-3}), (\textit{eij-PC-4}) are all satisfied at that time. 
Then:  

-- (\textit{PC-1}): It is easy to check from the definitions that (\textit{eij-PC-1}) directly implies (\textit{PC-1}). 

-- (\textit{w-PC-2}): By Item~3 of Proposition~\ref{prop:ei-jump_correctness}, 
	it holds that ${L^{\text{inc}}}^\iota = \text{Inc}(f^{\text{c}:\iota},i, 1)$.
	Since $\text{Inc}(f^{\text{c}:\iota},i, 1)\subseteq \text{Inc}(f^{\text{c}:\iota},i, s-1)$, then (\textit{w-PC-2}) holds.

-- (\textit{w-PC-3}): Let $u \in V \setminus {L^{\text{inc}}}^{\iota}$ and $v\in N_{\Gamma}^{\text{out}}(u)$. 
We need to check the following two cases. 
\begin{itemize}
\item[If $u\in V_0$,] by Item~1 of Proposition~\ref{prop:ei-jump_correctness}, 
both $\J.\texttt{cmp}^\iota$ and $\J.\texttt{cnt}^\iota$ are coherent \wrt $f^{\text{c}:\iota}$ in $\Gamma_{i,1}$.
Therefore, (\textit{w-PC-3}) holds when $u\in V_0$. 
\item[If $u\in V_1$] and $(u,v)$ is incompatible \wrt $f^{\text{c}:\iota}$ in $\Gamma_{i,1}$, 
then, $u\in \text{Inc}(f^{\text{c}:\iota}, i, 1)$. By Item~3 of Proposition~\ref{prop:ei-jump_correctness}, 
$\text{Inc}(f^{\text{c}:\iota}, i, 1)={L^{\text{inc}}}^\iota$. 
Thus, $u\in {L^{\text{inc}}}^\iota$, so (\textit{w-PC-3}) holds when $u\in V_1$. 
\end{itemize}
Therefore, (\textit{w-PC-3}) holds when $k=1$.

\textit{Inductive Step: $k>1$.}  
Consider the $k$-th invocation of $\texttt{J-VI}()$ for $k>1$, 
at line~\ref{subproc:ua-jumps:l2} of $\texttt{ua-jumps}()$. Say it happens at step $\iota$. 
Since $k>1$, just before step $\iota$, Algorithm~\ref{algo:solve_mpg} performed 
the $(k-1)$-th invocation of $\texttt{J-VI}()$ at line~\ref{subproc:ua-jumps:l2} of $\texttt{ua-jumps}()$.
Say it happened at step $\iota_0$. By induction hypothesis, 
at step $\iota_0$ the (\textit{PC-1}), (\textit{w-PC-2}), (\textit{w-PC-3}) were all satisfied.
Therefore, the $(k-1)$-th invocation of $\texttt{J-VI}()$ at line~\ref{subproc:ua-jumps:l2} halted in a correct manner, 
as prescribed by Proposition~\ref{prop:saint_weak_coherency}. Soon after that, Algorithm~\ref{algo:solve_mpg} invoked 
$\texttt{rejoin\_ua-jump}()$ at line~\ref{subproc:ua-jumps:l5} of $\texttt{ua-jumps}()$.

($\star$) The key is that $\texttt{rejoin\_ua-jump}()$, apart from copying the energy-levels of $L_f$ back to $\J.f$ 
(with $\texttt{scl\_back\_}f(s-1,F,\J)$ at line~\ref{subproc:unitary_join:l1}),  
it takes care of repairing (at line~\ref{subproc:unitary_join:l6}) 
the coherency state of $\J.\texttt{cnt}[u]$ and $\J.\texttt{cmp}[(u,v)]$ 
for all those $(u,v)\in E$ such that: $u\in V_0$, $w(u,v)=i$ and $\J.f[u]=\J.f[v]=0$ (if any);
moreover, it checks the compatibility state of all those arcs $(u,v)\in E$ such that: 
$u\in V_1$, $w(u,v)=i$ and $\J.f[u]=\J.f[v]=0$ (if any). 
In doing so, if any $u\in V\setminus L^{\text{inc}}$ is recognized to be inconsistent \wrt $f^{\text{c}}$ in $\Gamma_{i, s-1}$, then $u$ is (correctly) inserted into $L^{\text{inc}}$. 
See the pseudo-code of $\texttt{repair}()$ in SubProcedure~\ref{subproc:energy_jumps}.
	
With ($\star$) in mind, we can check that 
(\textit{PC-1}), (\textit{w-PC-2}), (\textit{w-PC-3}) are all satisfied at step $\iota$. 

		-- (\textit{PC-1}). Since $\texttt{rejoin\_ua-jump}()$ empties $L_f$ 
			(by $\texttt{scl\_back\_}f()$ at line~\ref{subproc:unitary_join:l1}), 
			then $L^\iota_f=\emptyset$. Next, we argue $f^{\text{c}}\preceq f^*_{w'_{i, s-1}}$. 
			By induction hypothesis, when the $(k-1)$-th invocation of $\texttt{J-VI}()$ 
			at line~\ref{subproc:ua-jumps:l2} of $\texttt{ua-jumps}()$ halts, 
			Proposition~\ref{prop:saint_weak_coherency} holds, therefore, 
			$f^{\text{c}}=f^*_{w'_{i-1, s-1}}$. Since $F_{s-1}=1$, then $w'_{i-1, s-1}= w_{i-1, s-1}$ 
			and $w'_{i, s-1}=w_{i, s-1}$. Therefore, the following holds for every $v\in V$:
		\begin{align*}
		f^{\text{c}}(v) = & f^*_{w'_{i-1, s-1}}(v) & \text{\footnotesize [by induction 
			hypothesis and Proposition~\ref{prop:saint_weak_coherency}]} \\  
			     = & f^*_{i-1, s-1}(v) & 
			\text{\footnotesize [by $w'_{i-1, s-1}=w_{i-1, s-1}$]} \\ 
			\preceq & f^*_{i, s-1}(v) & 
			\text{\footnotesize [by $w_{i-1, s-1} > w_{i, s-1}$ and Lemma~\ref{lem:monotone_energy}}] \\ 
			= & f^*_{w'_{i, s-1}}(v) & \text{\footnotesize [by $w_{i, s-1}=w'_{i, s-1}$]} 
		\end{align*}
			In summary, $\forall^{v\in V}\, f^{\text{c}}(v)\preceq f^*_{w'_{i, s-1}}(v)$. This proves (PC-1).
	
			-- (w-PC-2). By induction hypothesis and Proposition~\ref{prop:jvalue_init}, 
			all vertices that were already inside $L^{\text{inc}}$ at the end of 	
			the $(k-1)$-th invocation of $\texttt{J-VI}()$, at line~\ref{subproc:ua-jumps:l2} of $\texttt{ua-jumps}()$, 
			they were all inconsistent \wrt $\J.f^{\text{c}}$ in $\Gamma_{i-1, s-1}$, 
			so they are still inconsistent \wrt $\J.f^{\text{c}}$ in $\Gamma_{i, s-1}$, 	
			because $w'_{i, s-1} < w'_{i-1, s-1}$. 
			In addition, the repairing process performed by $\texttt{rejoin\_ua-jumps}()$,  
			as mentioned in ($\star$), can only add inconsistent vertices to $L^{\text{inc}}$. 
			Therefore, (w-PC-2) holds. 
				
			-- (w-PC-3). Let $u \in V \setminus {L^{\text{inc}}}^{\iota}$ and $v\in N_{\Gamma}^{\text{out}}(u)$. 
				We need to check the following two cases. 
			\begin{itemize}
			\item[Case $u\in V_0$.] In order to prove Item~1 of (w-PC-3), we need to check three cases. 
			\begin{enumerate}
			\item If $\J.\texttt{cmp}^\iota[u,v]=\texttt{F}$, 
				we argue that $(u,v)\in E$ is incompatible \wrt $f^{\text{c}}$ in $\Gamma_{i, s-1}$. 
	
			Indeed, one of the following two cases (i) or (ii) holds: 
				
		(i) $\J.\texttt{cmp}[(u,v)]=\texttt{F}$ was already so at the end of the $(k-1)$-th invocation of $\texttt{J-VI}()$.
		By induction hypothesis and by Item~2 of Proposition~\ref{prop:saint_weak_coherency}, then $(u,v)$ was incompatible \wrt $f^{\text{c}}$ in $\Gamma_{i-1, s-1}$.
		So, $(u,v)$ is incompatible \wrt $f^{\text{c}}$ in $\Gamma_{i, s-1}$ (as $w'_{i, s-1} < w'_{i-1, s-1}$).
				
		(ii) at the end of the $(k-1)$-th invocation of $\texttt{J-VI}()$, it was $\J.\texttt{cmp}[(u,v)]=\texttt{T}$. 
		But then, $\texttt{rejoin\_ua-jump}(i, s, F, \J)$ repaired it by setting $\J.\texttt{cmp}[(u,v)]=\texttt{F}$ at line~\ref{subproc:repair:l5} of $\texttt{repair}()$; 
		notice that this (correctly) happens \textit{iff} $w(u,v)=i$ and $\J.f[u]=\J.f[v]=0$, so that $(u,v)$ is really incompatible \wrt $f^{\text{c}}$ in $\Gamma_{i, s-1}$. 
				
		Therefore, in any case, Item~1 of (w-PC-3) holds. 
				
		\item	If $\J.\texttt{cmp}^\iota[(u,v)]=\texttt{T}$, we argue that either $(u,v)$ is compatible \wrt $f^{\text{c}}$ in $\Gamma_{i, s-1}$ or $v\in {L^{\text{inc}}}^\iota$. 
		Indeed, assume $\J.\texttt{cmp}^\iota[(u,v)]=\texttt{T}$ and that $(u,v)$ is incompatible \wrt $f^{\text{c}}$ in $\Gamma_{i, s-1}$.
		Since $\J.\texttt{cmp}^\iota[(u,v)]=\texttt{T}$, then it was as such even when the $(k-1)$-th 
		invocation of $\texttt{J-VI}()$ halted at line~\ref{subproc:ua-jumps:l2} of $\texttt{ua-jumps}()$.
		By induction hypothesis and by Item~2 of Proposition~\ref{prop:saint_weak_coherency}, $(u,v)$ was compatible \wrt $f^{\text{c}}$ in $\Gamma_{i-1, s-1}$. 
		But $(u,v)$ is now incompatible \wrt $f^{\text{c}}$ in $\Gamma_{i, s-1}$, and still $\J.\texttt{cmp}^\iota[(u,v)]=\texttt{T}$. 
		Thus, the last invocation of $\texttt{repair}()$, within the last $\texttt{rejoin\_ua-jump}()$, 
		has not recognized $(u,v)$ as incompatible (otherwise, it would be $\J.\texttt{cmp}^\iota[(u,v)]=\texttt{F}$). 
		Therefore, it must be that $f^{\text{c}}(u)>0$ or $f^{\text{c}}(v)>0$: otherwise, 
		if $f^{\text{c}}(u)=f^{\text{c}}(v)=0$, since $(u,v)$ is compatible \wrt $f^{\text{c}}$ in $\Gamma_{i-1, s-1}$ but incompatible in $\Gamma_{i, s-1}$, 
		and $w(u,v)\in \Z$, then $w(u,v)=i$ (contradicting the fact that the last invocation of $\texttt{repair}()$ has not recognized $(u,v)$ as incompatible).
		Moreover, since $(u,v)$ is incompatible \wrt $f^{\text{c}}$ in $\Gamma_{i, s-1}$, 
		then $f^{\text{c}}(u)\neq \top$; and since $(u,v)$ was compatible \wrt $f^{\text{c}}$ in $\Gamma_{i-1, s-1}$ 
		and $f^{\text{c}}(u)\neq \top$, then $f^{\text{c}}(v)\neq \top$. Now, when the $(k-1)$-th invocation of $\texttt{J-VI}()$ halts, 
		$L^{\text{inc}}=\{q\in V\mid 0<f^{\texttt{c}}(q)\neq\top\}$ holds by induction hypothesis and Item~3 of Proposition~\ref{prop:saint_weak_coherency}. 
		Since $u\not\in {L^{\text{inc}}}^\iota$ and $f^{\text{c}}(u)\neq \top$, then $f^{\text{c}}(u)=0$. 
		Thus, since either $f^{\text{c}}(u)>0$ or $f^{\text{c}}(v)>0$, it holds that $f^{\text{c}}(v)>0$. 
		So, it is $0<f^{\text{c}}(v)\neq\top$. Therefore, $v\in {L^{\text{inc}}}^\iota$. 
	\item By induction hypothesis and by Proposition~\ref{prop:saint_weak_coherency}, 
		when the $(k-1)$-th invocation of $\texttt{J-VI}()$ halts at line~\ref{subproc:ua-jumps:l2} of $\texttt{ua-jumps}()$, say at step $\iota_0$,  
		$f^{\text{c}}$ is the least-SEPM of the \EG $\Gamma_{i-1,s-1}$ and 
		$\J.\texttt{cnt}^{\iota_0}$, $\J.\texttt{cmp}^{\iota_0}$ are both coherent \wrt $f^{\text{c}}$ in $\Gamma_{i-1,s-1}$.
		Therefore, 
\begin{align*} 
\J.\texttt{cnt}^{\iota_0}[u] & = \big|\big\{ v \in N_{\Gamma}^{\text{out}}(u) \mid 
		f^{\text{c}}(u) \succeq f^{\text{c}}(v)\ominus 
		w'_{i-1,s-1}(u,v)\big\}\big| & \text{\footnotesize[by coherency of $\J.\texttt{cnt}^{\iota_0}$]} \\
				     & = \big|\big\{ v \in N_{\Gamma}^{\text{out}}(u) 
		\mid \J.\texttt{cmp}^{\iota_0}[(u,v)] = \texttt{T} \big\}\big|. & \text{\footnotesize [by coherency of $\J.\texttt{cmp}^{\iota_0}$]}  
\end{align*}
Moreover, since $f^{\text{c}}$ is least-SEPM of $\Gamma_{i-1,s-1}$, then $u$ is consistent \wrt $f^{\text{c}}$ in $\Gamma_{i-1,s-1}$; 
thus $\J.\texttt{cnt}^{\iota_0}[u]>0$. Then, after $\iota_0$ and before $\iota$, $\texttt{ua-jumps}()$ increments $i$ by 
one unit at line~\ref{subproc:ua-jumps:l4} and it invokes $\texttt{rejoin\_ua-jump}()$ at line~\ref{subproc:ua-jumps:l5}.
There, $\texttt{repair}()$ can (possibly) alter the state of both $\J.\texttt{cnt}$ 
and $\J.\texttt{cmp}$ at lines~\ref{subproc:repair:l4}-\ref{subproc:repair:l5}.
Whenever the state of $\J.\texttt{cmp}$ is modified from $\texttt{T}$ to $\texttt{F}$, 
then $\J.\texttt{cnt}$ is decremented by one unit; moreover, whenever $\J.\texttt{cnt}[u]=0$, 
then $\texttt{repair}()$ takes care of inserting $u$ into $L^{\text{inc}}$. 
Therefore, $\J.\texttt{cnt}^{\iota}[u] = \big|\big\{ v \in N_{\Gamma}^{\text{out}}(u) 
\mid \J.\texttt{cmp}^{\iota}[(u,v)] = \texttt{T} \big\}\big|$; 
since $u\not\in {L^{\text{inc}}}^\iota$, then $\J.\texttt{cnt}^{\iota}[u]>0$.
\end{enumerate} 
\item[Case $u\in V_1$.] Let $v\in N^{\text{out}}_{\Gamma}(u)$ be such 
that $(u,v)$ is incompatible \wrt $f^{\text{c}}$ in $\Gamma_{i,j}$. 
We claim $v\in {L^{\text{inc}}}^\iota$. By induction hypothesis and by Proposition~\ref{prop:saint_weak_coherency}, 
when the $(k-1)$-th invocation of $\texttt{J-VI}()$ halts at line~\ref{subproc:ua-jumps:l2} of $\texttt{ua-jumps}()$,  
say at step $\iota_0$, $f^{\text{c}}$ is the least-SEPM of 
the \EG $\Gamma_{i-1,s-1}$ and ${L^{\text{inc}}}^{\iota_0}=\{q\in V\mid 0<f^{\texttt{c}}(q)\neq\top\}$. 
Thus, since $u\in V_1$, the arc $(u,v)$ is compatible \wrt $f^{\text{c}}$ in $\Gamma_{i-1,s-1}$. 
Moreover, since $u\not\in {L^{\text{inc}}}^\iota$ by hypothesis, then $u\not\in{L^{\text{inc}}}^{\iota_0}$, 
thus $f^{\text{c}}(u)=0$ or $f^{\text{c}}(u)=\top$; but since $u$ is incompatible \wrt $f^{\text{c}}$ in $\Gamma_{i,j}$, 
it is $f^{\text{c}}(u)=0$. Now, if $0<f^{\text{c}}(v)\neq\top$, 
then $v\in {L^{\text{inc}}}^{\iota_0}\subseteq {L^{\text{inc}}}^{\iota}$, so we are done. 
Otherwise, since $(u,v)$ is compatible \wrt $f^{\text{c}}$ in $\Gamma_{i-1,s-1}$ and $f^{\text{c}}(u)=0$, 
then $f^{\text{c}}(v)\neq\top$. So, $f^{\text{c}}(v)=0$. Since $f^{\text{c}}(u)=f^{\text{c}}(v)=0$ 
and $(u,v)$ is compatible \wrt $f^{\text{c}}$ in $\Gamma_{i-1,s-1}$, 
but incompatible \wrt $f^{\text{c}}$ in $\Gamma_{i,s-1}$, and $w(u,v)\in\Z$, then $w(u,v)=i$.
Whence, soon after $\iota_0$, when $\texttt{ua-jumps}()$ invokes $\texttt{rejoin\_ua-jump}()$ at line~\ref{subproc:ua-jumps:l5};
there inside, $\texttt{repair}()$ takes care of inserting $v$ into $L^{\text{inc}}$. Therefore, $v\in {L^{\text{inc}}}^{\iota}$.
\end{itemize}
This concludes the inductive step, and thus the proof of Item~1 of Proposition~\ref{prop:ua-jumps_correctness}.
\end{proof}

\begin{proof}[Proof of Item~(2)] 
Consider the first invocation of $\texttt{J-VI}()$ at line~\ref{subproc:ua-jumps:l2} of $\texttt{ua-jumps}()$, 
say it happens at step $\iota$. Notice that, just before step $\iota$, 
the $\texttt{ei-jump}()$ was invoked by Algorithm~\ref{algo:solve_mpg} at line~\ref{algo:solve:l7}, say at step $\iota_0$. 
By hypothesis, (\textit{eij-PC-1}), (\textit{eij-PC-2}), (\textit{eij-PC-3}), (\textit{eij-PC-4}) were all satisfied at step $\iota_0$. 
By (\textit{eij-PC-1}) and Item~2 of Proposition~\ref{prop:ei-jump_correctness}, 
$f^{\text{c}:\iota_0}$ is the least-SEPM of $\Gamma_{i-1, s-1}$; 
therefore, $V_{f^{\text{c}:\iota_0}}=\W_0(\Gamma_{i-1,s-1})$ by Proposition~\ref{prop:least_energy_prog_measure}.
By Item~1 of Proposition~\ref{prop:ua-jumps_correctness} and Item~1 of Proposition~\ref{prop:saint_weak_coherency}, 
when the first invocation of $\texttt{J-VI}()$ halts, say at step $h$, then $f^{\text{c}:h}$ is the least-SEPM of $\Gamma_{i,s-1}$;
therefore, $V\setminus V_{f^{\text{c}:h}}=\W_1(\Gamma_{i, s-1})$ by Proposition~\ref{prop:least_energy_prog_measure}. 
Moreover, by Item~4 of Proposition~\ref{prop:saint_weak_coherency}, 
$L^h_\top = V_{f^{\text{c}:\iota}} \cap V\setminus V_{f^{\text{c}:h}}$. 
Therefore, $L^h_\top = \W_0(\Gamma_{i-1,s-1})\cap \W_1(\Gamma_{i, s-1})$.

Next, consider the $k$-th invocation of $\texttt{J-VI}()$, for $k>1$, at line~\ref{subproc:ua-jumps:l2} of $\texttt{ua-jumps}()$.
By Item~1 of Proposition~\ref{prop:ua-jumps_correctness} and Item~1 of Proposition~\ref{prop:saint_weak_coherency}, 
the following two hold: 
(i) when the $(k-1)$-th invocation of $\texttt{J-VI}()$ halts, at line~\ref{subproc:ua-jumps:l2} of $\texttt{ua-jumps}()$,  
say at step $\iota$, then $f^{\text{c}:\iota}$ is the least-SEPM of $\Gamma_{i-1,s-1}$; 
and $V_{f^{\text{c}:\iota}}=\W_0(\Gamma_{i-1,s-1})$ by Proposition~\ref{prop:least_energy_prog_measure}.
(ii) when the $k$-th invocation of $\texttt{J-VI}()$ halts, at line~\ref{subproc:ua-jumps:l2} of $\texttt{ua-jumps}()$,
say at step $h$, then $f^{\text{c}:h}$ is the least-SEPM of $\Gamma_{i,s-1}$; 
and $V\setminus V_{f^{\text{c}:h}}=\W_1(\Gamma_{i, s-1})$ by Proposition~\ref{prop:least_energy_prog_measure}.
Notice that, when the $k$-th invocation of $\texttt{J-VI}()$ takes place, soon after $\iota$, 
the current energy-levels are still $f^{\text{c}:\iota}$ (\ie they are not modified by 
$\texttt{rejoin\_ua-jumps}()$ at line~\ref{subproc:ua-jumps:l5} of $\texttt{ua-jumps}()$).
Moreover, by Item~4 of Proposition~\ref{prop:saint_weak_coherency}, 
$L^h_\top = V_{f^{\text{c}:\iota}} \cap V\setminus V_{f^{\text{c}:h}}$.
Therefore, by (i) and (ii), it holds $L^h_\top = \W_0(\Gamma_{i-1,s-1})\cap \W_1(\Gamma_{i, s-1})$.
\end{proof}

\begin{proof}[Proof of Item~(3)] We need to check the following two items (a) and (b).
\begin{enumerate}
\item[(a)] Consider the state of $L^{\text{inc}}_{\text{cpy}}$ at line~\ref{subproc:ua-jumps_backtrack:l1} 
	   of $\texttt{backtrack\_ua-jumps}(i,s,F,\J,\Gamma)$. By Item~1 of Proposition~\ref{prop:ua-jumps_correctness}, 
	   and by Item~3 of Proposition~\ref{prop:saint_weak_coherency}, 
	   when the last invocation of $\texttt{J-VI}()$ halts at line~\ref{subproc:ua-jumps:l2} of $\texttt{ua-jumps}()$,
	   say at step $h_0$, it holds ${\J.L^{\text{inc}}}^{h_0} = \{v\in V\mid 0 < f^{\text{c}:h_0}(v) \neq \top \}$.  
           By the copy operation which is performed at line~\ref{subproc:ua-jumps_backtrack:l1} 
	   of $\texttt{backtrack\_ua-jumps}()$, then ${\J.L^{\text{inc}}_{\text{cpy}}} = {\J.L^{\text{inc}}}^{h_0} 
						       = \{ v\in V\mid 0 < f^{\text{c}:h_0}(v) \neq \top \}$.
	   This proves (a).
\item[(b)] Let us focus on the two induced games $\Gamma[L^\iota_\top]$ and $\Gamma[V\setminus L^\iota_\top]$. 
		\begin{enumerate} 
			\item[i.] $\forall({v\in L^\iota_\top})\; f^{\text{c}:h}(v) = f^*_{w'_{i-1, s-1}}(v)$: indeed, 
				by arguing similarly as in the proof of Item~2 of Proposition~\ref{prop:ua-jumps_correctness}, 
				$\forall^{v\in V}\; \J.f^{\iota}[v] = f^*_{w'_{i-1, s-1}}(v)$.
				Notice that $\texttt{backtrack\_ua-jump}()$ modifies 
				the energy-levels only at 
				lines~\ref{subproc:ua-jumps_backtrack:l2}-\ref{subproc:ua-jumps_backtrack:l3}, 
				where the following assignment is performed: 
				\[ L^h_f[u]\leftarrow \left\{
					\begin{array}{ll} 
						\bot &, \text{ if } u \in L^{\iota}_\top; \\
						L^{\iota}_f[u] &, \text{ if } u \in V\setminus L^{\iota}_\top.
					\end{array}\right.\] 
				and $\texttt{scl\_back\_}f()$ is invoked (respectively).
				Since $\forall({v\in L^\iota_\top})\; L^{h}_f[u]=\bot$, 
				soon after the invocation of $\texttt{scl\_back\_}f()$ at 
				line~\ref{subproc:ua-jumps_backtrack:l3}, 
				it must be that $\forall({v\in L^\iota_\top})\; f^{\text{c}:h}(v) = 
				\J.f^{\iota}[v] = f^*_{w'_{i-1, s-1}}(v)$. This proves (i). 
			\item[ii.] $\forall({v\in V\setminus L^\iota_\top})\; f^{\text{c}:h}(v) = f^*_{w'_{i, s-1}}(v)$: indeed, 
				by Item~1 of Proposition~\ref{prop:ua-jumps_correctness} and 
				Item~1 of Proposition~\ref{prop:saint_weak_coherency}, 
				$\forall({v\in V})\; f^{\text{c}:\iota}(v) = f^*_{w'_{i, s-1}}(v)$.
				As mentioned, $\texttt{backtrack\_ua-jump}()$ modifies the energy-levels 
				only at lines~\ref{subproc:ua-jumps_backtrack:l2}-\ref{subproc:ua-jumps_backtrack:l3}, 
				where $L^h_f[u]$ is assigned (as above in (i)). 
				Since $\forall({v\in V\setminus L^\iota_\top})$ $(L^h_f[u]=L^{\iota}_f[u]$ and 
				$f^{\text{c}:\iota}(v) = f^*_{w'_{i, s-1}}(v))$, then (ii) holds.		
			\item[iii.] $\forall({u\in L^\iota_\top\cap V_0}) \forall({v\in N^{\text{out}}_{\Gamma[L^\iota_\top]}(u)})\; 
				    \J.\texttt{cmp}^h(u,v)$ is coherent 
					\wrt $f^{\text{c}:h}$ in $\Gamma[L^\iota_\top]_{i,1}$: indeed, 		
				at lines~\ref{subproc:ua-jumps_backtrack:l4}-\ref{subproc:ua-jumps_backtrack:l7} 
				of $\texttt{backtrack\_ua-jump}()$, for each $u\in L^\iota_\top\cap V_0$, 
				it is invoked the $\texttt{init\_cnt\_cmp}(u,i,1,F,\J,\Gamma[L^\iota_\top])$ 
				(line~\ref{subproc:ua-jumps_backtrack:l7}). 
				Therefore, (iii) holds.
			\item[iv.] $\forall({v\in L^\iota_\top\cap V_0})\; 
			   \J.\texttt{cnt}^h(v)$ is coherent \wrt $f^{\text{c}:h}$ 
				in $\Gamma[L^\iota_\top]_{i,1}$: same argument as in (iii). 
			\item[v.] $\forall(j'\in [1, s-1])\; \text{Inc}(f^{\text{c}:h}, i, j')\setminus L^\iota_\top=\emptyset$: 
			indeed, let $u\in V\setminus L^\iota_\top$ and let $j'\in [1, s-1]$ be fixed arbitrarily.
			We want to show that $u\not\in \text{Inc}(f^{\text{c}:h}, i, j')$.
			Since $f^*_{w'_{i, s-1}}$ is the least-SEPM of $\Gamma_{i, s-1}$, 
			then $u\not\in \text{Inc}(f^*_{w'_{i, s-1}}, i, s-1)$. We have two cases.
			\begin{itemize}
			\item[Case $u\in V_0\setminus L^\iota_\top$] Since $u\in V_0\setminus \text{Inc}(f^*_{w'_{i, s-1}}, i, s-1)$, 
				for \emph{some} $v\in N^{\text{out}}_\Gamma(u)$ it holds that: 
			\[f^*_{w'_{i, s-1}}(u)\succeq f^*_{w'_{i, s-1}}(v)\ominus w'_{i,s-1}(u,v). \tag{$*_0$}\] 
				By Item~(i) of Proposition~\ref{prop:ua-jumps_correctness}, it holds that 
				$\forall({v\in L^\iota_\top})\; f^{\text{c}:h}(v) = f^*_{w'_{i-1, s-1}}(v)$.
				By Item~(ii) of Proposition~\ref{prop:ua-jumps_correctness}, it holds that  
				$\forall({v\in V\setminus L^\iota_\top})\; f^{\text{c}:h}(v) = f^*_{w'_{i, s-1}}(v)$; 
				so $f^{\text{c}:h}(u)=f^*_{w'_{i, s-1}}(u)$.
				By Lemma~\ref{lem:monotone_energy}, $f^*_{w'_{i-1, s-1}}\preceq f^*_{w'_{i, s-1}}$.
				Then, since $f^{\text{c}:h}(u)=f^*_{w'_{i, s-1}}(u)$, 	
				since $f^{\text{c}:h}(v)\in\{f^*_{w'_{i-1, s-1}}(v), f^*_{w'_{i, s-1}}(v)\}$ 
				and $f^*_{w'_{i-1, s-1}}\preceq f^*_{w'_{i, s-1}}$, from ($*_0$) we obtain the following inequality: 
			\[f^{\text{c}:h}(u)\succeq f^{\text{c}:h}(v)\ominus w'_{i,s-1}(u,v).\]
				Now, since $w'_{i,j'}(u,v) \geq w'_{i,s-1}(u,v)$, 
				it also holds $f^{\text{c}:h}(u)\succeq f^{\text{c}:h}(v)\ominus w'_{i,j'}(u,v)$.
				This proves that $u\not\in \text{Inc}(f^{\text{c}:h}, i, j')$.
			\item[Case $u\in V_1\setminus L^\iota_\top$] Since $u\in V_1\setminus \text{Inc}(f^*_{w'_{i, s-1}}, i, s-1)$, 
				for \emph{all} $v\in N^{\text{out}}_\Gamma(u)$ it holds that: 
			\[f^*_{w'_{i, s-1}}(u)\succeq f^*_{w'_{i, s-1}}(v)\ominus w'_{i,s-1}(u,v).\]
			By arguing as in the previous case, we obtain that for every $v\in N^{\text{out}}_\Gamma(u)$ 
			the following holds: $f^{\text{c}:h}(u)\succeq f^{\text{c}:h}(v)\ominus w'_{i,s-1}(u,v)$.
			This proves that $u\not\in \text{Inc}(f^{\text{c}:h}, i, j')$.
			\end{itemize}
		\end{enumerate} 
\end{enumerate}
\end{proof}

\begin{proof}[Proof of Item~(4)] The fact that $\texttt{ua-jumps}()$ halts in finite time 
follows directly from Item~1 of Proposition~\ref{prop:ua-jumps_correctness} and 
the definition of $\texttt{rejoin\_ua-jump}()$ and that of $\texttt{backtrack\_ua-jump}()$. 
\end{proof}

\subsubsection{Correctness of $\texttt{solve\_MPG}()$ (Algorithm~\ref{algo:solve_mpg})}
As shown next, it turns out that (PC-1), (w-PC-2), (w-PC-3) are all satisfied by Algorithm~\ref{algo:solve_mpg}. 
\begin{Prop}\label{prop:all_pre_conditions_hold} 
Let $i\in [W^--1, W^+]$ and $j\in [1,s-1]$. The following two propositions hold.
\begin{enumerate}
\item Consider any invocation of $\texttt{ei-jump}(i,\J)$ 
	(SubProcedure~\ref{subproc:energy_jumps}) at line~\ref{algo:solve:l7} 
	of Algorithm~\ref{algo:solve_mpg} such that $L^{\text{inc}} = \emptyset$. 
	Then, \textit{(eij-PC-1)}, \textit{(eij-PC-2)}, \textit{(eij-PC-3)}, \textit{(eij-PC-4)} are all satisfied \wrt $\Gamma$.
\item Consider any invocation of $\texttt{J-VI}(i, j, F, \J, \Gamma[S])$ 
	at line~\ref{algo:jvalue:l11} of Algorithm~\ref{algo:solve_mpg}. 
	Then, \textit{(PC-1)}, \textit{(w-PC-2)}, \textit{(w-PC-3)} are all satisfied \wrt the sub-arena $\Gamma[S]$.
\end{enumerate}
\end{Prop}
\begin{proof}
We prove Item~1 and 2 jointly, arguing by induction on the number $k_1$ of invocations of $\texttt{ei-jump}()$ 
at line~\ref{algo:solve:l7} of Algorithm~\ref{algo:solve_mpg} \textit{and} the 
	number $k_2$ of invocations of $\texttt{J-VI}()$ at line~\ref{algo:jvalue:l11}.

\emph{Base Case: $k_1=1$ and $k_2=0$.}
So, the first subprocedure to be invoked is $\texttt{ei-jump}(i,\J)$ at line~\ref{algo:solve:l7} of Algorithm~\ref{algo:solve_mpg}, 
say at step $\iota$. Notice that: $i^{\iota}=W^--1$; $\forall (v\in V)\; f^{\text{c}:\iota}(v)=0$; $\forall (v\in V_0)\; 
\J.\texttt{cnt}^{\iota}[v]=|N^{\text{out}}_\Gamma(v)|$ and $\forall (u\in V_0) \forall (v\in N^{\text{out}}_\Gamma(u))\; 
\J.\text{cmp}^{\iota}[u,v]=\texttt{T}$; $L^\iota_f={L^\text{inc}}^{\iota}={L^\text{inc}_{\text{cpy}}}^{\iota}=\emptyset$. 
Also notice that for every $u\in V$ and $v\in N^{\text{out}}_\Gamma(u)$ the following holds: 
\[w'_{i^{\iota}, s-1}(u,v)= w'_{W^{-}-1, s-1}(u,v)=w(u,v)-W^{-}\geq 0.\] 
With this, it is straightforward to check that \textit{(eij-PC-1)}, 
	\textit{(eij-PC-2)}, \textit{(eij-PC-3)}, \textit{(eij-PC-4)} hold.   

\emph{Inductive Step: $k_1=1$ and $k_2\geq1$, or $k_1>1$.} We need to check three cases.
\begin{enumerate}
\item Assume that $\texttt{J-VI}(i,j,F,\J,\Gamma[S])$ is invoked 
at line~\ref{algo:solve:l11} of Algorithm~\ref{algo:solve_mpg}, say at step $\iota_1$, 
soon after that $\texttt{ua-jumps}()$ halted at line~\ref{algo:solve:l9} of Algorithm~\ref{algo:solve_mpg}. 
So, we aim at showing Item~2.

Notice that: $j^{\iota_1}=1$ holds (by line~\ref{algo:solve:l10} of Algorithm~\ref{algo:solve_mpg}). 
Let us check the (PC-1), (w-PC-2), (w-PC-3) \wrt $\Gamma[S]$. By line~\ref{subproc:ua-jumps_backtrack:l4} 
of $\texttt{backtrack\_ua-jump}()$, $S=L^\iota_\top$ for some step $\iota < \iota_1$.

-- \emph{PC-1}: By line~\ref{subproc:ua-jumps_backtrack:l2} of $\texttt{backtrack\_ua-jump}()$, it holds  
	$\forall (u\in L^{\iota}_\top)\; L_f^{\iota}[u]=\bot$. By Item~[3, (b), (i)] of Proposition~\ref{prop:ua-jumps_correctness}, 
	$\forall({v\in L^\iota_\top})\; f^{\text{c}:\iota_1}(v) = f^*_{w'_{i-1, s-1}}(v)$. 
	By Lemma~\ref{lem:monotone_energy}, $f^*_{w'_{i-1, s-1}}\preceq f^*_{w'_{i, j^{\iota_1}}}$. 
	Therefore, $\forall (v\in L^\iota_\top)\; f^{\text{c}:\iota_1}(v)\preceq f^*_{w'_{i,j^{\iota_1}}}(v)$. 
	This proves that (PC-1) holds \wrt $\Gamma[L^\iota_\top]=\Gamma[S]$.
		
-- \emph{w-PC-2:} By lines~\ref{subproc:ua-jumps_backtrack:l5}-\ref{subproc:ua-jumps_backtrack:l17} of $\texttt{backtrack\_ua-jumps}()$, 
	and since $\texttt{init\_cnt\_cmp}()$ is correct, ${L^{\text{inc}}}^{\iota_1}\subseteq \text{Inc}(f^{\text{c}:\iota_1}, i, 1)$.
	Since $j^{\iota_1}=1$, then (w-PC-2) holds \wrt $\Gamma[L^\iota_\top]=\Gamma[S]$.

-- \emph{w-PC-3:} By induction hypothesis and by Item~[3, (b), (iii) and (iv)] of Proposition~\ref{prop:ua-jumps_correctness}, 
	$\J.\texttt{cnt}^{\iota_1}$ and $\J.\texttt{cmp}^{\iota_1}$ are coherent \wrt $f^{\texttt{c}:\iota_1}$ in 
	$\Gamma[L^\iota_\top]_{i,1}$; also, if $u\in V_1\cap L^\iota_\top$ and $u\in \text{Inc}(f^{\text{c}:\iota_1}, i, 1)$, 
	then $u\in {L^{\text{inc}}}^{\iota_1}$ by lines~\ref{subproc:ua-jumps_backtrack:l11}-\ref{subproc:ua-jumps_backtrack:l17} 
	of $\texttt{backtrack\_ua-jump}()$. Since $j^{\iota_1}=1$, this proves that (PC-3) holds \wrt $\Gamma[L^\iota_\top]=\Gamma[S]$, 
	so (w-PC-3) holds as well.

\item Assume that $\texttt{J-VI}(i,j,F,\J,\Gamma[S])$ is invoked at line~\ref{algo:solve:l11} of Algorithm~\ref{algo:solve_mpg}, 
say at step $\iota_2$, soon after that a previous invocation of $\texttt{J-VI}(i,j-1,F,\J,\Gamma[S])$ halted 
at line~\ref{algo:solve:l11} say at step $\iota_1$. Notice that $j\in [2, s-1]$ in that case.
Let us check (PC-1), (w-PC-2), (w-PC-3) \wrt $\Gamma[S]$. 
By line~\ref{subproc:ua-jumps_backtrack:l4} of $\texttt{backtrack\_ua-jump}()$, $S=L^\iota_\top$ for some step $\iota < \iota_1$.

-- \emph{(PC-1)}: By lines~\ref{subproc:scale_back:l2}-\ref{subproc:scale_back:l3} of $\texttt{scl\_back}()$ 
	(which was executed at line~\ref{algo:solve:l13} of Algorithm~\ref{algo:solve_mpg}, just before $\iota_2$), 
	it holds $\forall (u\in L^{\iota}_\top)\; L_f^{\iota_1}[u]=\bot$.   
	By induction hypothesis and by Item~1 of Proposition~\ref{prop:saint_weak_coherency}, the following holds:   
	\[\forall({v\in L^\iota_\top})\; f^{\text{c}:\iota_2}(v) = f^*_{w'_{i, j^{\iota_1}}}(v)= f^*_{w'_{i, j^{\iota_2}-1}}(v).\]   
	By Lemma~\ref{lem:monotone_energy}, $f^*_{w'_{i, j^{\iota_2}-1}}\preceq f^*_{w'_{i, j^{\iota_2}}}$.   
	Therefore, $\forall (v\in L^\iota_\top)\; f^{\text{c}:\iota_2}(v)\preceq f^*_{w'_{i,j^{\iota_2}}}(v)$.	
	Whence, (PC-1) holds \wrt $\Gamma[L^\iota_\top]=\Gamma[S]$.  
		
-- \emph{(w-PC-2):} By induction hypothesis and by Item~3 of Proposition~\ref{prop:saint_weak_coherency}, 
	then: \[{L^{\text{inc}}}^{\iota_2} = \{v\in V\mid 0<f^{\text{c}:\iota_2}(v)\neq\top\}.\]
	Thus, by Lemma~\ref{lem:jvalue_init_correct}, 
	${L^{\text{inc}}}^{\iota_2}\subseteq \text{Inc}(f^{\text{c}:\iota_2}, i, j^{\iota_2})$.
	So, (w-PC-2) holds \wrt $\Gamma[L^\iota_\top]=\Gamma[S]$.

-- \emph{(w-PC-3):} Let $u\in V\setminus {L^{\text{inc}}}^{\iota_2}$, and let $v\in N^{\text{out}}_{\Gamma[L^\iota_\top]}(v)$.
\begin{itemize}
\item[If $u\in V_0$, ] we need to check the state of $\J.\texttt{cmp}^{\iota_2}[(u,v)]$ and $\J.\texttt{cnt}^{\iota_2}[u]$.
\end{itemize}

\emph{1.} If $\J.\texttt{cmp}^{\iota_2}[u,v]=\texttt{F}$, 
	we argue that $(u,v)\in E$ is incompatible \wrt $f^{\text{c}:\iota_2}$ in $\Gamma[S]_{i, j}$. 
	Indeed, it was already $\J.\texttt{cmp}^{\iota_2}[(u,v)]=\texttt{F}$ 
	when the previous $\texttt{J-VI}()$ (that invoked at step $\iota_1$) halted.
	Then, by induction hypothesis and by Item~3 of Proposition~\ref{prop:saint_weak_coherency}, 
	$(u,v)$ is incompatible \wrt $f^{\text{c}:\iota_2}$ in $\Gamma[S]_{i, j-1}$. 
	Thus, $(u,v)$ is incompatible \wrt $f^{\text{c}:\iota_2}$ also in 
	$\Gamma[S]_{i, j}$ (because $w'_{i, j} < w'_{i, j-1}$). 
	So, $\J.\texttt{cmp}^{\iota_2}[(u,v)]$ is coherent \wrt $f^{\texttt{c}:\iota_2}$ in $\Gamma[S]_{i,j}$.
			
\emph{2.} If $\J.\texttt{cmp}^{\iota_2}[(u,v)]=\texttt{T}$, 
	we argue that either $(u,v)$ is compatible \wrt $f^{\text{c}:\iota_2}$ in 
	$\Gamma[S]_{i, j}$ or it holds that $v\in {L^{\text{inc}}}^{\iota_2}$. 
	Indeed, assume that $\J.\texttt{cmp}^{\iota_2}[(u,v)]=\texttt{T}$ and 
	that $(u,v)$ is incompatible \wrt $f^{\text{c}:\iota_2}$ in $\Gamma[S]_{i, j}$.
	Since $\J.\texttt{cmp}^{\iota_2}[(u,v)]=\texttt{T}$, then it was as such even when 
	the previous $\texttt{J-VI}()$ (that invoked at step $\iota_1$) halted.
	So, $(u,v)$ was compatible \wrt $f^{\text{c}:\iota_2}$ in $\Gamma[S]_{i, j-1}$. 
	Since $u\not\in {L^{\text{inc}}}^{\iota_2}$, then $f^{\text{c}:\iota_2}(u)=0$ 
	(indeed, if $f^{\text{c}:\iota_2}(u)=\top$, then $(u,v)$ would have been compatible). 	
	Therefore, it is not possible that $f^{\text{c}:\iota_2}(v)=0$; 
	since, otherwise, from the fact that $f^{\text{c}:\iota_2}(u)=f^{\text{c}:\iota_2}(v)=0$ 
	and $(u,v)$ is compatible \wrt $f^{\text{c}:\iota_2}$ in $\Gamma[S]_{i, j-1}$,  
	it would be $w'(u,v)_{i,j-1}\geq 0$; and since $w(u,v)\in\Z$ and $0<F_{j-1}<F_{j}\leq 1$ where $j\in [2, s-1]$, 
	it would be $w'_{i,j}(u,v)\geq 0$ as well, so $(u,v)$ would be 
	compatible \wrt $f^{\text{c}}$ in $\Gamma[S]_{i, j}$. Also, 	
	it is not possible that $f^{\text{c}:\iota_2}(v)=\top$, since otherwise 
	$(u,v)$ would have been incompatible \wrt $f^{\text{c}:\iota_2}$ in $\Gamma[S]_{i, j-1}$ 
	(because $f^{\text{c}:\iota_2}(u)=0$). Therefore, $0<f^{\text{c}:\iota_2}(v)<\top$.
	Then, induction hypothesis and by Item~3 of Proposition~\ref{prop:saint_weak_coherency}, 
	$v\in {L^{\text{inc}}}^{\iota_2}$.
		 
\emph{3.} By induction hypothesis and by Proposition~\ref{prop:saint_weak_coherency},   
	$f^{\text{c}:\iota_2}$ is the least-SEPM of $\Gamma[S]_{i,j-1}$ and 
	$\J.\texttt{cnt}^{\iota_2}$, $\J.\texttt{cmp}^{\iota_2}$ are both 
	coherent \wrt $f^{\text{c}:\iota_2}$ in $\Gamma[S]_{i,j-1}$. Therefore, 
\begin{align*} 
\J.\texttt{cnt}^{\iota_2}[u] & = \big|\big\{ v \in N^{\text{out}}_{\Gamma[S]}(u) \mid 
f^{\text{c}:\iota_2}(u) \succeq f^{\text{c}:\iota_2}(v)\ominus 
w'_{i,j-1}(u,v)\big\}\big| & \text{\footnotesize [by coherency of $\J.\texttt{cnt}^{\iota_2}$]} \\
     & = \big|\big\{ v \in N^{\text{out}}_{\Gamma[S]}(u) 
\mid \J.\texttt{cmp}^{\iota_2}[(u,v)] = \texttt{T} \big\}\big|. & \text{\footnotesize [by coherency of $\J.\texttt{cmp}^{\iota_2}$]}  
\end{align*}
Moreover, since $f^{\text{c}:\iota_2}$ is least-SEPM of $\Gamma[S]_{i,j-1}$, 
then $u$ is consistent \wrt $f^{\text{c}:\iota_2}$ in $\Gamma[S]_{i, j-1}$, 
thus $\J.\texttt{cnt}^{\iota_2}[u]>0$. 

\begin{itemize}
\item[If $u\in V_1$,] assume $(u,v)$ is incompatible \wrt $f^{\text{c}:\iota_2}$ 
	in $\Gamma[S]_{i,j}$, for some $v\in N^{\text{out}}_{\Gamma[S]}(u)$. 
\end{itemize}
We want to prove $v\in {L^{\text{inc}}}^{\iota_2}$. By induction hypothesis and by Proposition~\ref{prop:saint_weak_coherency}, 
$f^{\text{c}:\iota_2}$ is the least-SEPM of $\Gamma[S]_{i,j-1}$ and 
${L^{\text{inc}}}^{\iota_2}=\{q\in V\mid 0<f^{\texttt{c}:\iota_2}(q)\neq\top\}$. 
Thus, since $v\in V_1$, $(u,v)$ is compatible \wrt $f^{\text{c}:\iota_2}$ in $\Gamma[S]_{i,j-1}$. 
Moreover, since $u\not\in {L^{\text{inc}}}^{\iota_2}$ by hypothesis, then $f^{\text{c}:\iota_2}(u)=0$ or $f^{\text{c}}(u)=\top$; 
but since $(u,v)$ is incompatible \wrt $f^{\text{c}:\iota_2}$ in $\Gamma[S]_{i,j}$, 
it is $f^{\text{c}:\iota_2}(u)=0$. Therefore, it is not possible that $f^{\text{c}:\iota_2}(v)=0$; 
since, otherwise, from the fact that $f^{\text{c}:\iota_2}(u)=f^{\text{c}:\iota_2}(v)=0$ 
and $(u,v)$ is compatible \wrt $f^{\text{c}:\iota_2}$ in $\Gamma[S]_{i, j-1}$, it would be: 
\[w'(u,v)_{i,j-1} = w(u,v)-i-F_{j-1}\geq 0,\] since $w(u,v)\in\Z$ and $0<F_{j-1}<F_{j}\leq 1$ where $j\in [2, s-1]$, 
it would be $w'_{i,j}(u,v)\geq 0$, so $(u,v)$ would have been compatible \wrt $f^{\text{c}}$ in $\Gamma[S]_{i, j}$. Also, 	
it is not possible that $f^{\text{c}:\iota_2}(v)=\top$, since otherwise $(u,v)$ would 
have been incompatible \wrt $f^{\text{c}:\iota_2}$ in $\Gamma[S]_{i, j}$ (because $f^{\text{c}:\iota_2}(u)=0$). 
Therefore, $0<f^{\text{c}:\iota_2}(v)<\top$. 
Then, induction hypothesis and by Item~3 of Proposition~\ref{prop:saint_weak_coherency}, $v\in {L^{\text{inc}}}^{\iota_2}$.

\item Assume that $\texttt{ei-jump}(i,\J)$ is invoked at line~\ref{algo:solve:l7} of Algorithm~\ref{algo:solve_mpg}, 
	say at step $\iota_1$, and that ${L^{\text{inc}}}^{\iota_1} = \emptyset$. Then, the following properties hold. 

-- (eij-PC-1) $f^{\text{c}:\iota_1}$ is the least-SEPM of $\Gamma_{i, s-1}$: indeed, 
	consider the previous invocation $\texttt{J-VI}(i,j^{\iota_0},F,\J,\Gamma[S^{\iota_0}])$ 
	at line~\ref{algo:solve:l11} of Algorithm~\ref{algo:solve_mpg}, 
	say it was invoked at step $\iota_0$ before $\iota_1$.
	By induction hypothesis and by Item~1 of Proposition~\ref{prop:saint_weak_coherency}, 
	$f^{\text{c}:\iota_1}$ is the least-SEPM of $\Gamma[S^{\iota_0}]_{i, j^{\iota_0}}$. 
	Since ${L^{\text{inc}}}^{\iota_1} = \emptyset$ by assumption, 
	by induction hypothesis and Item~3 of Proposition~\ref{prop:saint_weak_coherency}, 
	then $\{v\in S^{\iota_0} \mid 0<f^{\text{c}:\iota_1}(v)\neq\top\}={L^{\text{inc}}}^{\iota_1}=\emptyset$.
	We claim that $\forall (u\in S^{\iota_0})\; f^{\text{c}:\iota_1}(u)=\top$. 
	Indeed, the following holds.
	\begin{itemize}
	\item[If $u\in V_0$,] since $f^{\text{c}:\iota_1}$ is the least-SEPM of $\Gamma[S^{\iota_0}]_{i, j^{\iota_0}}$, 
		there exists $v\in N^{\text{out}}_{\Gamma[S^{\iota_0}]}(u)$ such that $(u,v)$ 
		is compatible \wrt $f^{\text{c}:\iota_1}$ in $\Gamma[S^{\iota_0}]_{i,j^{\iota_0}}$.
		So, it is not possible that $f^{\text{c}:\iota_1}(u)=0$: otherwise, it would be 
		$f^{\text{c}:\iota_1}(v)=0$ as well (because either $f^{\text{c}:\iota_1}(v)=0$ or $f^{\text{c}:\iota_1}(v)=\top$), 
		and since $w(u,v)\in \Z$ and $0<F_{j^{\iota_0}}\leq 1$ where $j^{\iota_0}\in [1, s-1]$, 
		then $(u,v)$ would be compatible \wrt $f^{\text{c}:\iota_1}$ even in $\Gamma[S^{\iota_0}]_{i, s-1}$, 
		thus $f^*_{w'_{i, s-1}}(u)=0$. But this contradicts the fact that, 
		by induction hypothesis, Item~1 of Proposition~\ref{prop:ua-jumps_correctness} and 
		Item~1 of Proposition~\ref{prop:saint_weak_coherency}, $f^*_{w'_{i, s-1}}(u)=\top$. 
		Therefore, $f^{\text{c}:\iota_1}(u)=\top$.
		
	\item[If $u\in V_1$,] since $f^{\text{c}:\iota_1}$ is the least-SEPM of $\Gamma[S^{\iota_0}]_{i, j^{\iota_0}}$, 
		for every $v\in N^{\text{out}}_{\Gamma[S^{\iota_0}]}(u)$, the arc $(u,v)$ 
		is compatible \wrt $f^{\text{c}:\iota_1}$ in $\Gamma[S^{\iota_0}]_{i,j^{\iota_0}}$.
		Now, by arguing as in the previous case (\ie $u\in V_0$), 
			it holds that $\forall (u\in S^{\iota_0})\; f^{\text{c}:\iota_1}(u)=\top$. 
	\end{itemize} 
	Thus, $\forall (u\in S^{\iota_0})\; f^{\text{c}:\iota_1}(u)=\top=f^*_{i,s-1}(u)$. 
	By induction hypothesis and Item~[3, (b), (ii)] of Proposition~\ref{prop:ua-jumps_correctness}, 
	$\forall({u\in V\setminus S^{\iota_0}})\; f^{\text{c}:\iota_1}(u) = f^*_{w'_{i, s-1}}(u)$.
	So, $f^{\text{c}:\iota_1} = f^*_{w'_{i, s-1}}$.
	
-- (eij-PC-2) ${L^{\text{inc}}}^{\iota_1}=\{v\in S^{\iota_0} \mid 0<f^{\text{c}:\iota_1}(v)\neq\top\}$: 
	this holds by induction hypothesis and by Item~3 of Proposition~\ref{prop:saint_weak_coherency}.

-- (eij-PC-3) ${L^{\text{inc}}_{\text{cpy}}}^\iota \subseteq \text{Inc}(f^{\text{c}:\iota}, i', j')$ for every $(i',j')>(i,s-1)$: 
	this holds by induction hypothesis plus Item~[3, (a)] of Proposition~\ref{prop:ua-jumps_correctness} 
	and Lemma~\ref{lem:jvalue_init_correct}.

-- (eij-PC-4): consider the previous invocation of $\texttt{J-VI}(i,j^{\iota_0},F,\J,\Gamma[S^{\iota_0}])$ 
	at line~\ref{algo:solve:l11} of Algorithm~\ref{algo:solve_mpg}, say at step $\iota_0$, just before $\iota_1$.
	By induction hypothesis and by Item~2 of Proposition~\ref{prop:saint_weak_coherency}, 
	for every $u\in V_0\cap S^{\iota_0}$, $J.\texttt{cnt}^{\iota_1}[u]$ 
	and $\J.\texttt{cmp}^{\iota_1}[(u,\cdot)]$ are both coherent \wrt $f^{\text{c}:\iota_1}$ in 
	$\Gamma[S^{\iota_0}]_{i, j^{\iota_0}}$; also, for every $u\in V_0\setminus S^{\iota_0}$, 
	$J.\texttt{cnt}^{\iota_1}[u]$ and $\J.\texttt{cmp}^{\iota_1}[(u,\cdot)]$ 
	are both coherent \wrt $f^*_{w'_{i,s-1}}$ in $\Gamma_{i, s-1}$.
	Since (eij-PC-1) holds, then $f^{\text{c}:\iota_1}=f^*_{w'_{i,s-1}}$. So, (eij-PC-4) holds.  
\end{enumerate}
\end{proof}

\begin{Lem}\label{lem:eventually_invoke}
Let $\hat{v}\in V$, assume $\val{\Gamma}{\hat{v}}=\hat{i}-F_{\hat{j}-1}$, for some $\hat{i}\in [W^-,W^+]$ and $\hat{j}\in [1,s-1]$. 

Then, eventually, Algorithm~\ref{algo:solve_mpg} invokes 
$\texttt{J-VI}(\hat{i},\hat{j}, F, \J, \Gamma[S])$ at line~\ref{algo:solve:l11}, for some $S\subseteq V$.
\end{Lem}
\begin{proof}
For the sake of contradiction, for any $S\subseteq V$, assume that $\texttt{J-VI}(\hat{i},\hat{j}, F, \J, \Gamma[S])$ 
is never invoked at line~\ref{algo:solve:l11} of Algorithm~\ref{algo:solve_mpg}.
At each iteration of the main \texttt{while} loop of Algorithm~\ref{algo:solve_mpg} 
(lines~\ref{algo:solve:l6}-\ref{algo:solve:l14}), $j$ is incremented (line~\ref{algo:solve:l14}); 
meanwhile, the value of $i$ stands still until (eventually) $\texttt{ei-jump}()$ 
and (possibly) $\texttt{ua-jumps}()$ increase it (also resetting $j\leftarrow 1$). 
Therefore, since $\texttt{J-VI}(\hat{i}, \hat{j}, F, \J, \Gamma[S])$ is never invoked at line~\ref{algo:solve:l11}, 
there are $i_0\in [W^-,W^+]$ and $j_0\in [1,s-1]$, where $(i_0, j_0) < (\hat{i},\hat{j})$, such that one of the following two hold: 
\begin{itemize}
\item Either $\texttt{ei-jump}(i_0, \J)$ (line~\ref{algo:solve:l7}) is invoked and, 
	when it halts say at step $h$, it holds $\J.i^h>\hat{i}$. 

In that case, by Item~1 of Proposition~\ref{prop:all_pre_conditions_hold} and Item~2 of Proposition~\ref{prop:ei-jump_correctness}, 
$f^*_{\hat{i},\hat{j}-1}(\hat{v})=f^*_{\hat{i},\hat{j}}(\hat{v})$. 
On the other hand, since $\val{\Gamma}{\hat{v}}=\hat{i}-F_{\hat{j}-1}$, 
then $\hat{v}\in \W_0(\Gamma_{\hat{i},\hat{j}-1}) \cap \W_1(\Gamma_{\hat{i},\hat{j}})$ by Theorem~\ref{Thm:transition_opt_values}; 
so, by Propositions~\ref{prop:least_energy_prog_measure}, 
$f^*_{\hat{i},\hat{j}-1}(\hat{v})\neq\top$ and $f^*_{\hat{i},\hat{j}}(\hat{v})=\top$. 
So, $\top\neq f^*_{\hat{i},\hat{j}-1}(\hat{v})=\top$; this is absurd. 

\item Or $\texttt{ua-jumps}(i_0, s, F, \J, \Gamma)$ (line~\ref{algo:solve:l9}) is invoked and, 
	when it halts say at step $h$, $J.i^h>\hat{i}$.

In that case, during the execution of $\texttt{ua-jumps}(i_0, s, F, \J, \Gamma)$, at some step $\hat\iota$, it is invoked 
$\texttt{J-VI}(\hat{i}, s-1, F, \J,\Gamma)$ (line~\ref{subproc:ua-jumps:l2} of $\texttt{ua-jumps}()$); 
and when it halts, say at step $\hat\iota_h$, by Item~2 of Proposition~\ref{prop:ua-jumps_correctness} 
and by line~6 of $\texttt{ua-jumps}()$, then 
$L^{\hat\iota_h}_{\top}=\W_0(\Gamma_{\hat{i}-1,s-1}) \cap \W_1(\Gamma_{\hat{i},s-1})=\emptyset$. 
Still, $\val{\Gamma}{\hat{v}}=\hat{i}-F_{\hat{j}-1}$, 
then $\hat{v}\in \W_0(\Gamma_{\hat{i},\hat{j}-1}) \cap \W_1(\Gamma_{\hat{i},\hat{j}})$ by Theorem~\ref{Thm:transition_opt_values}.
Since $\rho=\{w_{i,j}\}_{i,j}$ is monotone decreasing, 
then $\W_0(\Gamma_{\hat{i},\hat{j}-1})\subseteq \W_0(\Gamma_{\hat{i}-1,s-1})$ and 
$\W_1(\Gamma_{\hat{i},\hat{j}})\subseteq \W_1(\Gamma_{\hat{i},s-1})$.
Then, $ \hat{v}\in\W_0(\Gamma_{\hat{i},\hat{j}-1}) \cap \W_1(\Gamma_{\hat{i},\hat{j}}) 
\subseteq \W_0(\Gamma_{\hat{i}-1,s-1}) \cap \W_1(\Gamma_{\hat{i},s-1})=\emptyset$, but this is absurd.
\end{itemize} 
In either case, we arrive at some contradiction. 

Therefore, eventually, Algorithm~\ref{algo:solve_mpg} invokes $\texttt{J-VI}()$ 
at line~\ref{algo:solve:l11} on input $(\hat{i},\hat{j})$.
\end{proof}

\begin{Thm}
Given any input \MPG $\Gamma=(V,E,w,\langle V_0, V_1\rangle)$, Algorithm~\ref{algo:solve_mpg} halts in finite time. 

If $(\W_0, \W_1, \nu, \sigma^*_0)$ is returned, then: $\W_0$ is the winning set of Player~0 in $\Gamma$, 
$\W_1$ is that of Player~1, $\forall^{v\in V}\nu(v)=\val{\Gamma}{v}$, $\sigma^*_0$ 
is an optimal positional strategy for Player~0 in $\Gamma$. 
\end{Thm}
\begin{proof} 
Firstly, we argue that Algorithm~\ref{algo:solve_mpg} halts in finite time. 
Recall, by Propositions~[\ref{prop:saint_weak_coherency}, \ref{prop:ei-jump_correctness}, 
	\ref{prop:ua-jumps_correctness}, \ref{prop:all_pre_conditions_hold}], 
it holds that any invocation of $\texttt{J-VI}()$, $\texttt{ei-jump}()$, $\texttt{ua-jumps}()$ (respectively) halts in finite time.
It is easy to check at this point that (by lines~\ref{algo:solve:l14} of Algorithm~\ref{algo:solve_mpg}, 
line~\ref{subproc:energy_jump:l4} and \ref{subproc:energy_jump:l8} of $\texttt{ei-jump}()$, 
line~\ref{subproc:ua-jumps:l5} of $\texttt{ua-jumps}()$, and since $L_{\omega}$ was sorted in increasing order),  
whenever $\texttt{J-VI}()$ is invoked at line~\ref{algo:solve:l11} of Algorithm~\ref{algo:solve_mpg} -- at 
any two consequential steps $\iota_0,\iota_1$ (\ie such that $\iota_0<\iota_1$) -- 
then $(i^{\iota_0}, j^{\iota_0} ) < (i^{\iota_1}, j^{\iota_1})$. Also, 
by Propositions~\ref{prop:saint_weak_coherency}-\ref{prop:all_pre_conditions_hold}, 
whenever $\texttt{J-VI}()$ is invoked at line~\ref{algo:solve:l11} of Algorithm~\ref{algo:solve_mpg}, 
if it halts say at step $\iota_h$, then $f^{\text{c}:\iota_h}$ is the least-SEPM of $\Gamma_{i^{\iota_h},j^{\iota_h}}$; 
so, eventually, say when $(i^{\hat\iota_h},j^{\hat\iota_h})$ are sufficiently large, 
then $\forall^{u\in V}\, f^{\text{c}:\hat\iota_h}(u)=\top$; and, by Item~3 of Proposition~\ref{prop:saint_weak_coherency}, 
${L^{\text{inc}}}^{\hat\iota_h}=\{v\in V\mid 0<f^{\text{c}:\hat\iota_h}\neq \top\}$, so, ${L^{\text{inc}}}^{\hat\iota_h}=\emptyset$. 
Consider the first invocation of $\texttt{ei-jump}()$ (line~\ref{algo:solve:l7} 
of Algorithm~\ref{algo:solve_mpg}) that is made soon after this $\hat\iota_h$, and say it halts at step $h$. 
By Item~3 of Proposition~\ref{prop:ei-jump_correctness}, ${L^{\text{inc}}}^h=\text{Inc}(f^{\text{c}:h}, \J.i^h, 1)$.
Since $\forall^{u\in V}\, f^{\text{c}:\hat\iota_h}(u)=f^{\text{c}:h}(u)=\top$, 
then $\text{Inc}(f^{\text{c}:h}, \J.i^h, 1)=\emptyset$. Therefore, ${L^{\text{inc}}}^h=\emptyset$.
Therefore, Algorithm~\ref{algo:solve_mpg} halts at line~\ref{algo:solve:l8} soon after $h$.

Secondly, we argue that Algorithm~\ref{algo:solve_mpg} returns $(\W_0, \W_1, \nu, \sigma^*_0)$ correctly. 

On one side, $\W_0, \W_1, \nu, \sigma^*_0$ are accessed only when $\texttt{set\_vars}()$ is invoked (line~\ref{algo:solve:l12} of  
Algorithm~\ref{algo:solve_mpg}). Just before that, at line~\ref{algo:solve:l11}, some $\texttt{J-VI}()$ must have been invoked; 
say it halts at step $h$. By Items~1 and 2 of Proposition~\ref{prop:ua-jumps_correctness},  
$L^{h}_\top = \W_0(\Gamma_{i^{h},j^{h}-1}) \cap \W_1(\Gamma_{i^{h},j^{h}})$. 
Therefore, by Theorem~\ref{Thm:transition_opt_values}, $\nu$ is assigned correctly; so, $\W_0,\W_1$ are also assigned correctly. 
At this point, by Theorem~\ref{Thm:pos_opt_strategy}, also $\sigma^*_0$ is assigned correctly.

Conversely, let $\hat{v}\in V$ and assume $\val{\Gamma}{v}=\hat{i}-F_{\hat{j}-1}$ 
for some $\hat{i}\in [W^-,W^+]$ and $\hat{j}\in [1,s-1]$. By Lemma~\ref{lem:eventually_invoke}, eventually, 
Algorithm~\ref{algo:solve_mpg} invokes $\texttt{J-VI}(\hat{i},\hat{j}, F, \J, \Gamma)$ at line~\ref{algo:solve:l11}. 
By Items~1 and 2 of Proposition~\ref{prop:ua-jumps_correctness}, when $\texttt{J-VI}(\hat{i},\hat{j}, F, \J, \Gamma)$ halts, 
say at step $h$, it holds that $L^{h}_\top = \W_0(\Gamma_{\hat{i},\hat{j}-1}) \cap \W_1(\Gamma_{\hat{i},\hat{j}})$. 
Therefore, soon after at line~\ref{algo:solve:l12}, $\texttt{set\_vars}()$ assigns to $\W_0,\W_1,\nu,\sigma^*_0$ a correct state.
\end{proof}

\subsection{Complexity of Algorithm~\ref{algo:solve_mpg}}
The complexity of Algorithm~\ref{algo:solve_mpg} follows, essentially, from the fact that [\emph{Inv-EI}] is satisfed.
\begin{Prop} 
Algorithm~\ref{algo:solve_mpg} satisfies [\emph{Inv-EI}]: 
whenever a Scan-Phase is executed (each time that a Value-Iteration is invoked), 
an energy-level $f(v)$ strictly increases for at least one $v\in V$. 
So, the energy-lifting operator $\delta$ is applied (successfully) at least once per each $\texttt{J-VI}()$. 
\end{Prop}
\begin{proof}
By lines~\ref{subproc:energy_jump:l1} and \ref{subproc:energy_jump:l9} 
of $\texttt{ei-jump}()$ (SubProcedure~\ref{subproc:energy_jumps}), lines~\ref{subproc:ua-jumps:l1}-6 of $\texttt{ua-jumps}()$, 
and line~\ref{algo:solve:l8} of Algorithm~\ref{algo:solve_mpg}, 
whenever $\texttt{J-VI}()$ is invoked either at line~\ref{algo:jvalue:l11} of Algorithm~\ref{algo:solve_mpg} or 
at line~\ref{subproc:ua-jumps:l2} of $\texttt{ua-jumps}()$ (SubProcedure~\ref{subproc:ua-jumps}), 
say at step $\iota$, then ${L^{\text{inc}}}^{\iota}\neq\emptyset$. Moreover, 
by Proposition~\ref{prop:all_pre_conditions_hold}, by Item~3 of Propositions~\ref{prop:saint_weak_coherency} 
and Lemma~\ref{lem:jvalue_init_correct}, ${L^{\text{inc}}}^{\iota}\subseteq \text{Inc}(f^{\text{c}:\iota}, i^\iota, j^\iota)$.
Therefore, during each $\texttt{J-VI}()$ that is possibly invoked by Algorithm~\ref{algo:solve_mpg}, 
at least one application of $\delta$ is performed (line~\ref{algo:jvalue:l2}-\ref{algo:jvalue:l3} 
of $\texttt{J-VI()}$) (because ${L^{\text{inc}}}^{\iota}\neq\emptyset$) and 
every single application of $\delta(f^{\text{c}}, v)$ that is made during $\texttt{J-VI}()$, say at step $\hat\iota$, 
for any $v\in V$, really increases $f^{\text{c}:\hat\iota}(v)$ 
(because ${L^{\text{inc}}}^{\hat\iota}\subseteq \text{Inc}(f^{\text{c}:\hat\iota}, i^{\hat\iota}, j^{\hat\iota})$).
\end{proof}

\begin{Thm}
Given an input \MPG $\Gamma=(V,E,w,\langle V_0, V_1\rangle)$, Algorithm~\ref{algo:solve_mpg} halts within the following time bound: 
\[ O(|E|\log|V|) + \Theta\Big( \sum_{v\in V} \texttt{deg}_{\Gamma}(v)\cdot\ell_{\Gamma}^1(v)\Big) = O(|V|^2|E|W), \]
The working space is $\Theta(|V|+|E|)$. 
\end{Thm}
\begin{proof} 
The initialization of $L_{\omega}$ takes $O(|E|\log|V|)$ time, \ie the cost for sorting $\{w_e\mid e\in W\}$. 
Each single application of $\delta$, which is possibly done during any execution 
of $\texttt{J-VI}()$ throughout Algorithm~\ref{algo:solve_mpg}, it takes time $\Theta(\texttt{deg}_{\Gamma}(v))$. 
So, the total aggregate time spent for all applications of $\delta$ in Algorithm~\ref{algo:solve_mpg} is 
$\Theta\big(\sum_{v\in V} \texttt{deg}_{\Gamma}(v)\cdot\ell_{\Gamma}^1(v)\big)$. 
It is not difficult to check from the description of Algorithm~\ref{algo:solve_mpg}, 
at this point, that the time spent between any two subsequent applications of $\delta$ can increase the total time amount 
$\sum_{v\in V} \texttt{deg}_{\Gamma}(v)\cdot\ell_{\Gamma}^1(v)$ of Algorithm~\ref{algo:solve_mpg} only by a constant factor. 
Next, notice that the aggregate total cost of all the invocations of $\texttt{repair}()$ is $O(|E|)$. 
Also recall that in Section~\ref{section:values} it was shown how to generate $\F_{|V|}$ iteratively, 
one term after another, in $O(1)$ time-delay and $O(1)$ total space,~\cite{PaPa09}. Also it is easy to check, at this point, 
that Algorithm~\ref{algo:solve_mpg} works with $\Theta(|V|+|E|)$ space.
\end{proof}

\subsection{An Experimental Evaluation of Algorithm~1}\label{subsect:experiments}

This section describes an empirical evaluation of Algorithm~\ref{algo:solve_mpg}.
All algorithms and procedures employed in this practical evaluation have been implemented in C/C++ 
and executed on a Linux machine having the following characteristics:

-- Intel Core i5-4278U CPU @ 2.60GHz x2; 

-- 3.8GB RAM; 

-- Ubuntu 15.10 Operating System.

\begin{figure}[h!]
	\centering
	\subfloat[Test~1, $\mu$, Algo-0.]{\label{Table:Test1-a}
		\begin{tabular}[b]{| r | r | r |}
			\hline
		 	\multicolumn{1}{|c|}{$|V|$}  & \multicolumn{1}{c|}{$\mu$ (sec)} & \multicolumn{1}{c|}{$\sigma$} \\
			\hline
			 	      20  &     0.69  &  0.21   \\
				      25  &    1.69   &  0.43   \\
				      30  &    4.37   &  1.47   \\
				      35  &    8.77   &  3.79   \\
				      40  &    19.95  &  6.99   \\
				      45  &    35.06  &  12.0   \\ 
				      50  &    57.10  &  18.9   \\
			\hline
		\end{tabular}
	}
	\qquad
	\subfloat[Test~1, $\mu$, Algo-1.]{\label{Table:Test1-b}
		\begin{tabular}[b]{| r | r | r |}
			\hline
		 	\multicolumn{1}{|c|}{$|V|$}  & \multicolumn{1}{c|}{$\mu$ (sec)} & \multicolumn{1}{c|}{$\sigma$} \\
			\hline
				      20 &  0.09  &  0.05    \\
				      25 &  0.13  &  0.07    \\
				      30 &  0.30  &  0.24     \\
				      35 &  0.52  &  0.39     \\
				      40 &  1.18  &  0.77     \\
				      45 &  1.83  &  1.37     \\ 
				      50 &  2.56  &  2.59     \\
				      60 &  5.30  &  6.08     \\
				      70 &  9.57  &  8.83     \\
			\hline
		\end{tabular}
	}
	\qquad
	\subfloat[Interpolation of average execution times in Test~1 for Algo-0 (red, mark=o) and Algo-1 (blue, mark=x).]{\label{SubFig:Test1}
 	    \begin{tikzpicture}[scale=0.54]
			\begin{axis}[legend pos=north east, 
				xlabel={$n$},
				scaled x ticks=false,
				minor x tick num=1,
				ylabel={Time},
				y unit=s,
 				/pgfplots/ylabel near ticks,
 				/pgfplots/xlabel near ticks,
				xlabel style={ at={(ticklabel cs:1)}, anchor=south west}
				]
				\addplot[red, mark=o, 
						error bars/.cd,
						y dir=both, y explicit, 
					]
				    	table[row sep=crcr, x=x,y=y,y error=yerr] {
				        x       y       	yerr        \\
				      20      0.699   0.219   \\
				      25      1.696   0.430  \\
				      30      4.377   1.475 \\
				      35      8.772   3.797 \\
				      40      19.95   6.998  \\
				      45      35.06   12.05   \\ 
				      50      57.10   18.99  \\
				}; 
				\addlegendentry{Avg Execution Time of Algo.~0}
				\addplot[blue, mark=x, 
						error bars/.cd,
						y dir=both, y explicit, 
					]
				    	table[row sep=crcr, x=x,y=y,y error=yerr] {
				        x       y       	yerr        \\
				      20     0.090  		0.05  \\
				      25     0.135  		0.076  \\
				      30     0.309 		 0.24\\
				      35     0.524 		 0.39 \\
				      40     1.183 		 0.77 \\
				      45     1.839 		1.37\\ 
				      50     2.568 		 2.59\\
				      60     5.30           6.08     \\
				      70     9.57           8.83     \\
				}; 
			\addlegendentry{Avg. Execution Time of Algo.~1}
			\end{axis}
		\end{tikzpicture}
	}
\caption{Results of Test~1 on Average Execution Time}\label{fig:test1}
\end{figure}
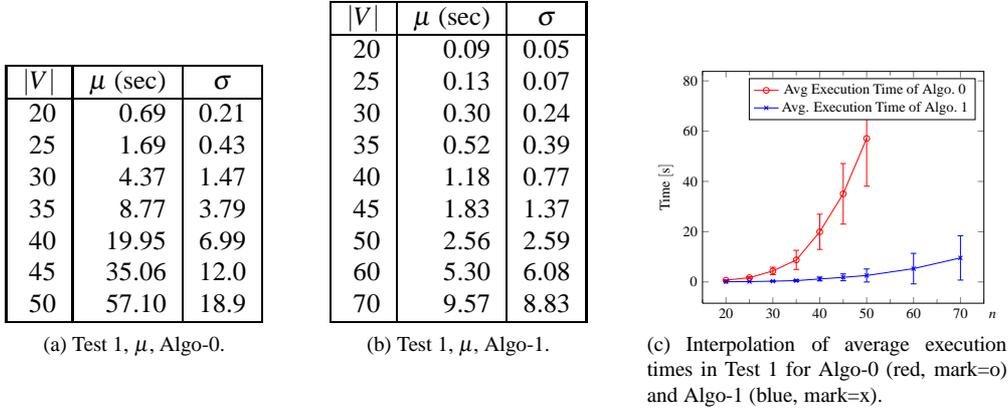

\begin{figure}[h!]
	\centering
	\subfloat[Test~1, $\ell^0_\Gamma$, Algo-0.]{\label{Table:Test1-c}
		\begin{tabular}[b]{| r | r | r |}
			\hline
		 	\multicolumn{1}{|c|}{$|V|$}  & \multicolumn{1}{c|}{$\mu$ ($\ell_\Gamma^0$)} & \multicolumn{1}{c|}{$\sigma$} \\
			\hline
			 	      20  &     3.71E+06    &   1.05E+06     \\
				      25  &     7.67E+06    &   2.09E+06     \\
				      30  &     1.69E+07    &   5.35E+06     \\
				      35  &     2.62E+07    &   9.52E+06     \\
				      40  &     5.22E+07    &   1.84E+07     \\
				      45  &     8.12E+07    &   3.16E+07     \\ 
				      50  &     1.07E+08    &   2.99E+07     \\
			\hline
		\end{tabular}
	}
	\qquad
	\subfloat[Test~1, $\ell^1_\Gamma$, Algo-1.]{\label{Table:Test1-d}
		\begin{tabular}[b]{| r | r | r |}
			\hline
		 	\multicolumn{1}{|c|}{$|V|$}  & \multicolumn{1}{c|}{$\mu$ ($\ell^1_\Gamma$)} & \multicolumn{1}{c|}{$\sigma$} \\
			\hline
				      20 &      3.53E+05  &  2.04E+05        \\
				      25 &      4.45E+05  &  2.17E+05        \\
				      30 &      7.92E+05  &  5.83E+05        \\
				      35 &      9.91E+05  &  7.53E+05        \\
				      40 &      2.13E+06  &  1.34E+06        \\
				      45 &      2.88E+06  &  1.98E+06        \\ 
				      50 &      3.08E+06  &  2.72E+06        \\
				      60 &	5.28E+06  &  5.82E+06        \\
				      70 &      8.66E+06  &  7.71E+06	     \\
			\hline
		\end{tabular}
	}
	\qquad
	\subfloat[Interpolation of avgerage values of $\ell^0_\Gamma, \ell^1_\Gamma$ in Test~1 for Algo-0 (orange, mark=o) and Algo-1 (cyan, mark=x).]{\label{SubFig:Test1}
 	    \begin{tikzpicture}[scale=0.6]
			\begin{axis}[legend pos=north west, 
				xlabel={$n$},
				scaled x ticks=false,
				minor x tick num=1,
				ylabel={$\ell^{0,1}_\Gamma$},
				/pgfplots/ylabel near ticks,
 				/pgfplots/xlabel near ticks,
				xlabel style={ at={(ticklabel cs:1)}, anchor=south west}
				]
				\addplot[orange, mark=o, 
						error bars/.cd,
						y dir=both, y explicit, 
					]
				    	table[row sep=crcr, x=x,y=y,y error=yerr] {
				        x       y       	yerr      \\
				      20      3713741.583    1045337.448  \\
				      25      7667912.045    2091357.946  \\
				      30       16907960.26   5345047.564  \\
				      35       26232493.27   9520624.638  \\
				      40       52211621.39   18351036.68  \\
				      45       81243895.04   31583603.27  \\ 
				      50      107116091.7    29914778.19  \\
				}; 
				\addlegendentry{Avg $\ell^0_{\Gamma}$ of Algo.~0}
				\addplot[cyan, mark=x, 
						error bars/.cd,
						y dir=both, y explicit, 
					]
				    	table[row sep=crcr, x=x,y=y,y error=yerr] {
				        x       y       	yerr         \\
				      20     352691.875      204321.9423     \\
				      25     445189.9545     217162.8688     \\
				      30     791978.7391     582729.0047     \\
				      35     990700.6364     752976.0787     \\
				      40     2133071.87      1342249.076     \\
				      45     2879754.792     1975571.826     \\ 
				      50     3078611.083     2724718.282     \\
				      60 	5280000      5820000        \\
				      70       8660000       7710000 	     \\
				}; 
			\addlegendentry{Avg. $\ell^1_{\Gamma}$ of Algo.~1}
			\end{axis}
		\end{tikzpicture}
	}
\caption{Results of Test~1 on $\ell^0_\Gamma, \ell^1_\Gamma$}\label{fig:test1-ell}
\end{figure}

\begin{figure}[h!]
	\centering
	\subfloat[Test~2, $\mu$, Algo-0.]{\label{Table:Test2-a}
		\begin{tabular}[b]{| r | r | r |}
			\hline
		 	\multicolumn{1}{|c|}{$W$}  & \multicolumn{1}{c|}{$\mu$ (sec)} & \multicolumn{1}{c|}{$\sigma$} \\
			\hline
				      50   &  0.97   & 0.33  \\
			 	      100  &  1.93   & 0.72  \\
				      150  &  2.66   & 0.80  \\
				      200  &  3.77  & 1.11  \\
				      250  &  5.00  & 1.55 \\
				      300  &  5.63  & 1.53 \\
				      350  &  7.38  & 2.34 \\ 
			\hline
		\end{tabular}
	}
	\qquad
	\subfloat[Test~2, $\mu$, Algo-1.]{\label{Table:Test2-b}
		\begin{tabular}[b]{| r | r | r |}
			\hline
		 	\multicolumn{1}{|c|}{$W$}  & \multicolumn{1}{c|}{$\mu$ (sec)} & \multicolumn{1}{c|}{$\sigma$} \\
			\hline
				     50   &  0.09  & 0.07   \\
				     100  &  0.22  & 0.18  \\
				     150  &  0.23  & 0.12   \\
				     200  &  0.31  & 0.16 \\
				     250  &  0.50  & 0.36    \\
				     300  &  0.42  & 0.30  \\
				     350  &  0.55  & 0.34   \\ 
			\hline
		\end{tabular}
	}
	\qquad
	\subfloat[Interpolation of average execution times in Test~2 for Algo-0 (red, mark=o) and Algo-1 (blue, mark=x).]{\label{SubFig:Test2}
 	    \begin{tikzpicture}[scale=0.53]
			\begin{axis}[legend pos=north east, 
				xlabel={$W$},
				scaled x ticks=false,
				minor x tick num=1,
				ylabel={Time},
				y unit=s,
 				/pgfplots/ylabel near ticks,
 				/pgfplots/xlabel near ticks,
				xlabel style={ at={(ticklabel cs:1)}, anchor=south west}
				]
				\addplot[red, mark=o, 
						error bars/.cd,
						y dir=both, y explicit, 
					]
				    	table[row sep=crcr, x=x,y=y,y error=yerr] {
				        x       y     yerr  \\
					50     0.97   0.33  \\			
                                        100    1.93   0.72  \\
                                        150    2.66   0.80  \\
                                        200    3.77   1.11  \\
                                        250    5.00   1.55  \\
                                        300    5.63   1.53  \\
                                        350    7.38   2.34  \\ 
				}; 
				\addlegendentry{Avg Execution Time of Algo.~0}
				\addplot[blue, mark=x, 
						error bars/.cd,
						y dir=both, y explicit, 
					]
				    	table[row sep=crcr, x=x,y=y,y error=yerr] {
				        x       y       	yerr   \\
				      50     0.09   0.07   \\  
				      100    0.22   0.18   \\   
				      150    0.23   0.12   \\  
				      200    0.31   0.16   \\    
				      250    0.50   0.36    \\ 
				      300    0.42   0.30   \\    
				      350    0.55   0.34   \\  
				}; 
			\addlegendentry{Avg. Execution Time of Algo.~1}
			\end{axis}
		\end{tikzpicture}
	}
\caption{Results of Test~2 on Average Execution Time}\label{fig:test2}
\end{figure}

\begin{figure}[h!]
	\centering
	\subfloat[Test~2, $\ell^0_\Gamma$, Algo-0.]{\label{Table:Test2-c}
		\begin{tabular}[b]{| r | r | r |}
			\hline
		 	\multicolumn{1}{|c|}{$W$}  & \multicolumn{1}{c|}{$\mu$ ($\ell_\Gamma^0$)} & \multicolumn{1}{c|}{$\sigma$} \\
			\hline
			 	         50   &  4.24E+06   & 1.18E+06  \\
				         100  &  8.05E+06   & 2.77E+06   \\
				         150  &  1.28E+07    & 3.57E+06  \\
				         200  &  1.81E+07   & 4.98E+06   \\
				         250  &  2.35E+07   & 7.13E+06    \\
				         300  &  2.77E+07   & 6.29E+06     \\   
				         350  &   3.54E+07  & 1.22E+07    \\ 
			\hline
		\end{tabular}
	}
	\qquad
	\subfloat[Test~2, $\ell^1_\Gamma$, Algo-1.]{\label{Table:Test2-d}
		\begin{tabular}[b]{| r | r | r |}
			\hline
		 	\multicolumn{1}{|c|}{$W$}  & \multicolumn{1}{c|}{$\mu$ ($\ell^1_\Gamma$)} & \multicolumn{1}{c|}{$\sigma$} \\
			\hline
					 50   &    2.84E+05  & 2.02E+05    \\
				         100  &    6.21E+05  & 4.62E+05    \\
				         150  &    7.75E+05  & 4.24E+05    \\
				         200  &    1.08E+06  & 5.22E+05    \\
				         250  &    1.62E+06 & 1.10E+06    \\
				         300  &    1.38E+06  & 8.07E+05    \\   
				         350  &    1.84E+06  & 1.09E+06    \\ 
			\hline
		\end{tabular}
	}
	\qquad
	\subfloat[Interpolation of avgerage values of $\ell^0_\Gamma, \ell^1_\Gamma$ in Test~2 
				for Algo-0 (orange, mark=o) and Algo-1 (cyan, mark=x).]{\label{SubFig:Test2}
 	    \begin{tikzpicture}[scale=0.6]
			\begin{axis}[legend pos=north west, 
				xlabel={$W$},
				scaled x ticks=false,
				minor x tick num=1,
				ylabel={$\ell^{0,1}_\Gamma$},
				/pgfplots/ylabel near ticks,
 				/pgfplots/xlabel near ticks,
				xlabel style={ at={(ticklabel cs:1)}, anchor=south west}
				]
				\addplot[orange, mark=o, 
						error bars/.cd,
						y dir=both, y explicit, 
					]
				    	table[row sep=crcr, x=x,y=y,y error=yerr] {
				        x       y           yerr  \\
					 50    4240000    1180000   \\
				         100   8050000    2770000    \\
				         150   12800000   3570000    \\
				         200   18100000   4980000    \\
				         250   23500000   7130000    \\
				         300   27700000   6290000    \\   
				         350   35400000   12200000   \\
				}; 
				\addlegendentry{Avg $\ell^0_{\Gamma}$ of Algo.~0}
				\addplot[cyan, mark=x, 
						error bars/.cd,
						y dir=both, y explicit, 
					]
				    	table[row sep=crcr, x=x,y=y,y error=yerr] {
				        x       y       	yerr     \\
					 50     284000     202000      \\
				         100    621000     462000      \\  		
                        	         150    775000     424000      \\
                        	         200    1080000    522000      \\
                        	         250    1620000    1100000      \\
	                	         300    1380000    807000      \\   
				         350    1840000    1090000     \\   	
			}; 
			\addlegendentry{Avg. $\ell^1_{\Gamma}$ of Algo.~1}
			\end{axis}
		\end{tikzpicture}
	}
\caption{Results of Test~2 on $\ell^0_\Gamma, \ell^1_\Gamma$}\label{fig:test2-ell}
\end{figure}

Source codes and scripts are (will be soon, \wrt time of submission) fully available online.

The main goal of this experiment was: (i) to determine the average computation time 
	of Algorithm~\ref{algo:solve_mpg}, with respect to randomly-generated \MPG{s}, 
in order to give an idea of the practical behavior it; (ii) to offer an experimental comparison between 
Algorithm~\ref{algo:solve_mpg} and the algorithm which is offered in~\cite{CR16}, \ie Algorithm~0, 
in order to give evidence and experimental confirmation of the algorithmic improvements made over~\cite{CR16}.
Here we propose a summary of the obtained results presenting a brief report about, Test~1, Test~2.

In all of our tests, in order to generate a suitable dataset of \MPG{s}, 
	our choice has been to use the \texttt{randomgame} procedure of \texttt{pgsolver} suite~\citep{pgsolver}, 
that can produce random arenas instances for any given number of nodes. 
We exploited \texttt{randomgame} as follows:
\begin{enumerate} 
	\item First, \texttt{randomgame} was used to generate random directed graphs, 
		with out-degree taken uniformly at random in $[1,|V|]$ ;
	\item Then, the resulting graphs were translated into \MPG{s} by 
		weighting each arc with an integer randomly chosen in the interval $[-W,W]$, 
		where $W$ was chosen accordingly to the test type;
\end{enumerate}
With such settings, the resulting \MPG{s} are characterized by $|V|$ and $W$.

In Test~1 the average computation time was determined for different orders of $|V|$.
For each $n\in \{20, 25, 30, 35, 40, 45, 50\}$, $25$ \MPG{s} instances 
with maximum weight $W=100$ were generated by \texttt{randomgame}. 
Each instance had been solved both with Algorithm~0 and Algorithm~1. 
In addition, to experiment a little further on Algorithm~\ref{algo:solve_mpg}, 
for each $n\in \{60, 70\}$, $25$ \MPG{s} instances with maximum 
weight (fixed to) $W=100$ were also generated by \texttt{randomgame} and solved only with Algorithm~1. 
The results of the test are summarized in \figref{fig:test1}, 
where each execution mean time is depicted as a 
point with a vertical bar representing its confidence interval determined according to its std-dev. 
As shown by \figref{fig:test1}, 
Test~1 gives experimental evidence of the supremacy of Algorithm~\ref{algo:solve_mpg} over Algorithm~0.
In order to provide a better insight on the behavior of the algorithms, 
a comparison between the values of $\ell^0_\Gamma$ and $\ell^1_\Gamma$ is offered in \figref{fig:test1-ell}.
Test~1 confirms that $\ell^1_\Gamma \ll \ell^0_\Gamma$ 
(by a factor $\geq 10^2$ when $|V|\geq 50$) on randomly generated \MPG{s}. 
The numerical results of Table~\ref{Table:Test1-a}-\ref{Table:Test1-b} 
suggest that the std-dev of both the avgerage  
running time of Algorithm~1 and of $\ell^1_\Gamma$ is greater 
(in proportion) than that of Algorithm~0 and $\ell^0_\Gamma$; 
but thinking about it this actually turns out to be a benefit: 
as a certain proportion of \MPG{s} instances can now exhibit quite a smaller value of $\ell^1_\Gamma$, 
then the running time improves, but the std-dev fluctuates more meanwhile.
 
In Test~2 the average computation time was determined for different orders of $W$. 
For each $W\in \{50, 100, 150, 200, 250, 300, 350\}$, $25$ \MPG{s} instances with maximum weight $W$, 
and $|V|=25$ (fixed), were generated by \texttt{randomgame}. 
Each instance had been solved both with Algorithm~0 and Algorithm~1. 
The results of the test are summarized in \figref{fig:test2} 
and \figref{fig:test2-ell}, where each execution mean time 
and $\ell^{0,1}_\Gamma$ is depicted as a point with a vertical bar representing 
its confidence interval determined according to its std-dev. 

In summary our experiments suggest that, even in practice, Algorithm~\ref{algo:solve_mpg} 
is significantly faster than the Algorithm~0 devised in~\citep{CR15, CR16}.

\section{An Energy-Lattice Decomposition of $\texttt{opt}_\Gamma\Sigma^M_0$}\label{sect:energy}
Recall the example arena $\Gamma_{\text{ex}}$ shown in \figref{fig:ex1_arena}. 
It is easy to see that $\forall^{v\in V} \val{\Gamma_{\text{ex}}}{v}=-1$. 
Indeed, $\Gamma_{\text{ex}}$ contains only two cycles, \ie $C_L=[A,B,C,D]$ and $C_R=[F,G]$, 
also notice that $w(C_L)/C_L=w(C_R)/C_R=-1$. The least-SEPM $f^*$ of the reweighted \EG $\Gamma_{\text{ex}}^{w+1}$ 
can be computed by running a Value Iteration~\citep{brim2011faster}. 
Taking into account the reweighting $w\leadsto w+1$, as in \figref{fig:ex1_reweighted_leastSEPM}: 
$f^*(A)=f^*(E)=f^*(G)=0$, $f^*(B)=f^*(D)=f^*(F)=4$, and $f^*(C)=8$. 
\begin{figure}[!h] \center
\begin{tikzpicture}[scale=.6, arrows={-triangle 45}, node distance=1.5 and 2]
 		\node[node, thick, color=red, label={$\sizedcircled{.15ex}{0}$}] (E) {$E$};
		\node[node, thick, color=blue, fill=blue!20, left=of E, 
					label={above right:$\sizedcircled{.15ex}{8}$}] (C) {$C$};
		\node[node, thick, color=red, above=of C, xshift=-8.5ex, yshift=-5ex, 
					label={$\sizedcircled{.15ex}{4}$}] (B) {$B$};
		\node[node, thick, color=blue, fill=blue!20, left=of C, 
					label={left:$\sizedcircled{.15ex}{0}$}] (A) {$A$};
		\node[node, thick, color=red, below=of C, xshift=-8.5ex, yshift=5ex, 
					label={below:$\sizedcircled{.15ex}{4}$}] (D) {$D$}; 	
		\node[node, thick, color=blue, fill=blue!20, right=of E, 
					label={$\sizedcircled{.15ex}{4}$}] (F) {$F$};
		\node[node, thick, color=red, right=of F, 
					label={right:$\sizedcircled{.15ex}{0}$}] (G) {$G$};	
		\draw[] (E) to [bend left=0] node[above] {$+1$} (C);
		\draw[] (E) to [bend left=0] node[above] {$+1$} (F);
		\draw[color=red, thick] (E) to [bend left=22] 
					node[above left, xshift=-4ex] {$+1$} (A.south east);
		\draw[color=red, thick] (E) to [bend left=50] 
					node[above left, xshift=-3ex, yshift=-1.5ex] {$+1$} (G.north);
		\draw[] (A) to [bend left=40] node[left] {$+4$} (B);
		\draw[color=red, thick] (B) to [bend left=40] node[xshift=2ex, yshift=1ex] {$+4$} (C);
		\draw[] (C) to [bend left=40] node[xshift=2ex, yshift=0ex] {$-4$} (D);
		\draw[color=red, thick] (D) to [bend left=40] node[xshift=-2ex, yshift=-1ex] {$-4$} (A);
		\draw[] (F) to [bend left=40] node[above, yshift=-.75ex] {$-4$} (G);
		\draw[color=red, thick] (G) to [bend left=40] node[below, yshift=.75ex] {$+4$} (F);
\end{tikzpicture}
\caption{The least-SEPM $f^*$ of $\Gamma_{\text{ex}}^{w+1}$ (energy-levels are depicted in circled boldface). 
All and only those arcs of Player~0 that are compatible with $f^*$ 
are $(B,C), (D,A), (E,A), (E,G), (G,F)$ (thick red arcs).}\label{fig:ex1_reweighted_leastSEPM}
\end{figure}
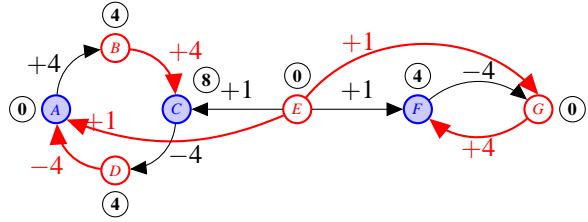

So $\Gamma_{\text{ex}}$ (\figref{fig:ex1_reweighted_leastSEPM}) implies the following.
\begin{Prop}\label{prop:counter_example}
The converse statement of Theorem~\ref{Thm:pos_opt_strategy} does not hold. 
There exist infinitely many \MPG{s} $\Gamma$ having at least one $\sigma_0\in\texttt{opt}_{\Gamma}\Sigma^M_0$ 
which is not compatible with the least-SEPM of $\Gamma$.
\end{Prop}
\begin{proof} Consider the $\Gamma_{\text{ex}}$ of \figref{fig:ex1_reweighted_leastSEPM}, 
and the least-SEPM $f^*$ of the \EG $\Gamma_{\text{ex}}^{w+1}$. The only vertex at which Player~0 really has a choice is $E$. 
Every arc going out of $E$ is optimal in the \MPG $\Gamma_{\text{ex}}$: 
whatever arc $(E,X)\in\E$ (for any $X\in \{A,C,F,G\}$) Player~0 chooses at $E$, the resulting payoff 
equals $\val{\Gamma_{\text{ex}}}{E}=-1$. Let $f^*$ be the least-SEPM of $f^*$ in $\Gamma_{\text{ex}}^{w+1}$. 
Observe, $(E,C)$ and $(E,F)$ are not compatible with $f^*$ in $\Gamma_{\text{ex}}^{w+1}$, only $(E,A)$ and $(E,G)$ are. 
For instance, the positional strategy $\sigma_0\in\Sigma^M_0$ defined 
as $\sigma_0(E)\triangleq F$, $\sigma_0(B)\triangleq C$, $\sigma_0(D)\triangleq A$, $\sigma_0(G)\triangleq F$ ensures a payoff 
$\forall^{v\in V}\val{\Gamma_{\text{ex}}}{v}=-1$, but it is not compatible with the least-SEPM $f^*$ of 
$\Gamma_{\text{ex}}^{w+1}$ (because $f^*(E) = 0 < 3 = f^*(F) \ominus w(E,F)$). 
It is easy to turn the $\Gamma_{\text{ex}}$ of \figref{fig:ex1_reweighted_leastSEPM} 
into a family on infinitely many similar examples. 
\end{proof}

We now aim at strengthening the relationship between $\text{opt}_{\Gamma}\Sigma^M_0$ 
and the Energy-Lattice $\mathcal{E}_\Gamma$. For this, we assume \textit{wlog} 
$\exists^{\nu\in\Q} \forall^{v\in V} \val{\Gamma}{v}=\nu$. This follows from Theorem~\ref{thm:ergodic_partition}, 
which allows one to partition $\Gamma$ into several domains $\Gamma_i \triangleq \Gamma_{|_{C_i}}$   
each one satisfying: $\exists^{\nu_i\in\Q}\forall^{v\in C_i} \val{\Gamma_i}{v}=\nu_i$.
By Theorem~\ref{thm:ergodic_partition} we can study $\texttt{opt}_{\Gamma_i}\Sigma^M_0$, 
independently \wrt $\texttt{opt}_{\Gamma_j}\Sigma^M_0$ for $j\neq i$. 

We say that an \MPG $\Gamma$ is \emph{$\nu$-valued} \textit{iff} $\exists^{\nu\in\Q}\forall^{v\in V} \val{\Gamma}{v}=\nu$.

Given an \MPG $\Gamma$ and $\sigma_0\in\Sigma^M_0(\Gamma)$, recall, 
$G(\Gamma,\sigma_0)\triangleq (V, E', w')$ is obtained from $G^\Gamma$ 
by deleting all and only those arcs that are not part of $\sigma_0$, 
\ie \[E' \triangleq \big\{(u,v)\in E \mid u\in V_0 \text{ and } 
		v=\sigma_0(u)\big\} \cup \big\{(u,v)\in E \mid u\in V_1\big\}, \] 
where each $e\in E'$ is weighted as in $\Gamma$, \ie $w':E'\rightarrow \Z:e\mapsto w_e$. 

When $G=(V,E,w)$ is a weighted directed graph, a \emph{feasible-potential (FP)} for $G$ 
is any map $\pi:V\rightarrow \C_{G}$ such that 
$\forall^{u\in V}\forall^{v\in N^{\text{out}}(u)} \pi(u)\succeq \pi(v)\ominus w(u,v)$. 
The \emph{least}-FP $\pi^*=\pi^*_{G}$ is the (unique) FP such that, for any other FP $\pi$, 
it holds $\forall^{v\in V} \pi^*(v)\preceq \pi(v)$. Given $G$, 
the Bellman-Ford algorithm can be used to produce $\pi^*_{G}$ in $O(|V| |E|)$ time. 
Let $\pi^*_{G(\Gamma, \sigma_0)}$ be the \emph{least-FP} of $G(\Gamma, \sigma_0)$. 
Notice, for every $\sigma_0\in\Sigma_0^M$, the least-FP $\pi^*_{G(\Gamma, \sigma_0)}$ 
is actually a SEPM for the \EG $\Gamma$; still it can differ from the least-SEPM of $\Gamma$, due to $\sigma_0$. 
We consider the following family of strategies.

\begin{Def}[$\Delta^M_0(f, \Gamma)$-Strategies] 
Let $\Gamma=\langle V, E, w, (V_0, V_1) \rangle$ and let $f:V\rightarrow\mathcal{C}_{\Gamma}$ be a SEPM for the \EG $\Gamma$. 
Let $\Delta_0^M(f, \Gamma)\subseteq \Sigma^M_0(\Gamma)$ be the family of all 
and only those positional strategies of Player~0 in $\Gamma$ 
such that $\pi^{*}_{G(\Gamma,\sigma_0)}$ coincides with $f$ pointwisely, 
\ie \[\Delta_0^M(f, \Gamma)\triangleq\Big\{\sigma_0\in\Sigma_{0}^M(\Gamma) 
	\left|\right. \forall^{ v\in V } \, \pi^*_{G(\Gamma,\sigma_0)}(v)=f(v)\Big\}.\]
\end{Def}  

We now aim at exploring further on the relationship between $\mathcal{E}_\Gamma$ and 
$\texttt{opt}_{\Gamma}\Sigma^M_0$, via $\Delta_0^M(f, \Gamma)$. 
\begin{Def}[The Energy-Lattice of $\texttt{opt}_{\Gamma}\Sigma^M_0$] 
Let $\Gamma$ be a $\nu$-valued \MPG. Let $\mathcal{X} \subseteq \mathcal{E}_{\Gamma^{w-\nu}}$ be a sub-lattice 
of SEPM{s} of the reweighted \EG $\Gamma^{w-\nu}$. 

We say that $\mathcal{X}$ is an ``\emph{Energy-Lattice} of $\texttt{opt}_{\Gamma}\Sigma^M_0$" 
\textit{iff} $\forall^{f\in\mathcal{X}} \Delta_0^M(f, \Gamma^{w-\nu})\neq\emptyset$ and the following disjoint-set 
decomposition holds: \[\displaystyle \texttt{opt}_{\Gamma}\Sigma^M_0 = 
				\bigsqcup_{f\in\mathcal{X}} \Delta_0^M(f, \Gamma^{w-\nu}).\] 
\end{Def} 

\begin{Lem}\label{lem:pre_main_thm} 
Let $\Gamma$ be a $\nu$-valued \MPG, and let $\sigma^*_0\in\texttt{opt}_{\Gamma}\Sigma^M_0$. Then,  
$G(\Gamma^{w-\nu},\sigma^*_0)$ is \emph{conservative} (\ie it contains \emph{no} negative cycle).
\end{Lem}
\begin{proof}
Let $C \triangleq (v_1\ldots, v_{k}, v_1)$ by any cycle in $G(\Gamma^{w-\nu},\sigma^*_0)$. 
Since we have $\sigma^*_0\in\texttt{opt}_\Gamma\Sigma^M_0$ and $\forall^{v\in V} \val{\Gamma}{v}=\nu$, thus  
$w(C)/k=\frac{1}{k}\sum_{i=1}^{k} w(v_i, v_{i+1}) 
	\geq \nu$ (for $v_{k+1}\triangleq v_1$) by Proposition~\ref{prop:reachable_cycle}, 
so that, assuming $w'\triangleq w-\nu$, then: $w'(C)/k=\frac{1}{k}\sum_{i=1}^{k} 
	\big(w(v_i, v_{i+1})-\nu\big) = w(C)/k - \nu \geq \nu - \nu = 0$. 
\end{proof}

Some aspects of the following Proposition~\ref{prop:energy_existence} rely heavily on Theorem~\ref{Thm:pos_opt_strategy}: 
the compatibility relation comes again into play. 
Moreover, we observe that Proposition~\ref{prop:energy_existence} is equivalent to the following fact, 
which provides a sufficient condition for a positional strategy to be optimal.
Consider a $\nu$-valued \MPG $\Gamma$, for some $\nu\in \Q$, 
and let $\sigma^*_0\in\texttt{opt}_{\Gamma}\Sigma^M_0$.
Let $\hat{\sigma}_0\in\Sigma^M_0(\Gamma)$ be any (not necessarily optimal) positional strategy for Player~0 in the \MPG $\Gamma$.  
Suppose the following holds: \[\forall^{v\in V} \pi^*_{G(\Gamma^{w-\nu}, \hat{\sigma}_0)}(v)=\pi^*_{G(\Gamma^{w-\nu},\sigma^*_0)}(v).\] 
Then, by Proposition~\ref{prop:energy_existence}, $\hat{\sigma_0}$ 
	is an optimal positional strategy for Player~0 in the \MPG $\Gamma$.

We are thus relying on the same \emph{compatibility} relation between $\Sigma^M_0$ and SEPM{s} in reweighted \EG{s} 
which was at the \emph{base} of Theorem~\ref{Thm:pos_opt_strategy}, aiming at extending Theorem~\ref{Thm:pos_opt_strategy} 
so to describe the whole $\texttt{opt}_{\Gamma}\Sigma^M_0$ (and not just the join part of it).

\begin{Prop}\label{prop:energy_existence}
Let the \MPG $\Gamma$ be $\nu$-valued, for some $\nu\in \Q$. 

There is at least one Energy-Lattice of $\texttt{opt}_{\Gamma}\Sigma^M_0$:
\[\mathcal{B}\triangleq \{ \pi^*_{G(\Gamma^{w-\nu},\sigma_0)} 
	\mid {\sigma_0\in \texttt{opt}_{\Gamma}\Sigma_0^M}\}.\]
\end{Prop}
\begin{proof}
The only non-trivial point to check being: 
$ \bigsqcup_{f\in\mathcal{B}} \Delta^M_0(f, \Gamma^{w-\nu}) \subseteq \texttt{opt}_\Gamma\Sigma^M_0 $. 

For this, we shall rely on Theorem~\ref{Thm:pos_opt_strategy}. 
Let $\hat{f}\in \mathcal{B}$ and $\hat{\sigma}_0\in \Delta^M_0(\hat{f}, \Gamma^{w-\nu})$ be fixed (arbitrarily).
Since $\hat{f}\in\mathcal{B}$, then 
$\hat{f}=\pi^*_{G(\Gamma^{w-\nu}, \sigma^*_0)}$ for some $\sigma^*_0\in\texttt{opt}_\Gamma\Sigma^M_0$.
Therefore, the following holds: 
	\[ \pi^*_{G(\Gamma^{w-\nu}, \hat{\sigma}_0)} = \hat{f} = \pi^*_{G(\Gamma^{w-\nu}, \sigma^*_0)}. \]
Clearly, $\hat{\sigma}_0$ is compatible 
with $\hat{f}$ in the \EG $\Gamma^{w-\nu}$, because $\hat{f}=\pi^*_{G(\Gamma^{w-\nu}, \hat{\sigma}_0)}$.
By Lemma~\ref{lem:pre_main_thm}, since $\sigma^*_0$ is optimal, 
then $G(\Gamma^{w-\nu}, \sigma^*_0)$ is conservative. Therefore: 
\[V_{\hat{f}}=V_{\pi^*_{G(\Gamma^{w-\nu}, \sigma^*_0)}}=V.\] 
Notice, $\hat{\sigma}_0$ satisfies exactly the hypotheses required by Theorem~\ref{Thm:pos_opt_strategy}. 
Therefore, $\hat{\sigma}_0\in\texttt{opt}_\Gamma\Sigma^M_0$. 
This proves (*).
This also shows $\texttt{opt}_{\Gamma}\Sigma^M_0 = 
	\bigsqcup_{f\in\mathcal{B}} \Delta^M_0(f, \Gamma^{w-\nu})$, and concludes the proof. 
\end{proof}

\begin{Prop}\label{prop:energy_uniqueness}
Let the \MPG $\Gamma$ be $\nu$-valued, for some $\nu\in \Q$. Let $\mathcal{B}_1$ and 
$\mathcal{B}_2$ be two Energy-Lattices for $\texttt{opt}_{\Gamma}\Sigma^M_0$. 
Then, $\mathcal{B}_1=\mathcal{B}_2$.
\end{Prop}
\begin{proof}
By symmetry, it is sufficient to prove that $\mathcal{B}_1\subseteq \mathcal{B}_2$. Let $f_1\in\mathcal{B}_1$ be fixed (arbitrarily). 
Then, $f_1=\pi^*_{G(\Gamma^{w-\nu}, \hat{\sigma}_0)}$ 
for some $\hat{\sigma}_0\in\texttt{opt}_{\Gamma}\Sigma^M_0$. 
Since $\hat{\sigma}_0\in\texttt{opt}_{\Gamma}\Sigma^M_0$ and since $\mathcal{B}_2$ is an Energy-Lattices, 
there exists $f_2\in \mathcal{B}_2$ such that $\hat{\sigma}_0\in\Delta^M_0(f_2, \Gamma^{w-\nu})$, which implies  
$\pi^*_{G(\Gamma^{w-\nu},\hat{\sigma}_0)}=f_2$. Thus, $f_1= \pi^*_{G(\Gamma^{w-\nu},\hat{\sigma}_0)} = f_2$. 
This implies $f_1\in \mathcal{B}_2$. 
\end{proof}

The next theorem summarizes the main point of this section. 
\begin{Thm}\label{thm:main_energystructure}
Let $\Gamma$ be a $\nu$-valued \MPG, for some $\nu\in \Q$. 
Then, $\mathcal{B}^*_\Gamma\triangleq \{ \pi^*_{G(\Gamma^{w-\nu},\sigma_0)} 
				\mid {\sigma_0\in \texttt{opt}_{\Gamma}\Sigma_0^M}\}$ 
	is the unique Energy-Lattice of $\texttt{opt}_\Gamma\Sigma^M_0$.
\end{Thm}
\begin{proof}
By Proposition~\ref{prop:energy_existence} and Proposition~\ref{prop:energy_uniqueness}.
\end{proof}

\begin{Exp}\label{exp1} 
Consider the \MPG $\Gamma_{\text{ex}}$, as defined in \figref{fig:ex1_arena}. Then,  
$\mathcal{B}^*_{\Gamma_{\texttt{ex}}}=\{f^*, f_1, f_2\}$, 
where $f^*$ is the least-SEPM of the reweighted \EG $\Gamma^{w+1}_{\texttt{ex}}$, 
and where the following holds: 
$f_1(A)=f_2(A)=f^*(A)=0$; $f_1(B)=f_2(B)=f^*(B)=4$; $f_1(C)=f_2(C)=f^*(C)=8$; $f_1(D)=f_2(D)=f^*(D)=4$;
$f_1(F)=f_2(F)=f^*(F)=4$; $f_1(G)=f_2(G)=f^*(G)=0$; finally, $f^*(E)=0$, $f_1(E)=3$, $f_2(E)=7$. 
An illustration of $f_1$ is offered in \figref{fig:ex1_reweighted_f1SEPM} (energy-levels are depicted in circled boldface). 
whereas $f_2$ is depicted in \figref{fig:ex1_reweighted_f2SEPM}.
Notice that $f^*(v)\leq f_1(v)\leq f_2(v)$ for every $v\in V$, 
and this ordering relation is illustrated in \figref{fig:ex1_ordering}.
\end{Exp}

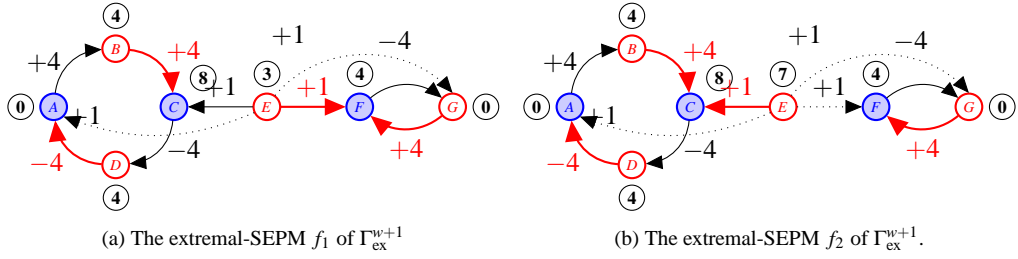
\begin{figure}[!h]
\center
\subfloat[The extremal-SEPM $f_1$ of $\Gamma_{\text{ex}}^{w+1}$]{\label{fig:ex1_reweighted_f1SEPM}
\begin{tikzpicture}[scale=.57, arrows={-triangle 45}, node distance=1.5 and 1.5]
 		\node[node, thick, color=red, label={$\sizedcircled{.15ex}{3}$}] (E) {$E$};
		\node[node, thick, color=blue, fill=blue!20, left=of E, label={above right:$\sizedcircled{.15ex}{8}$}] (C) {$C$};
		\node[node, thick, color=red, above=of C, xshift=-8.5ex, yshift=-5ex, label={$\sizedcircled{.15ex}{4}$}] (B) {$B$};
		\node[node, thick, color=blue, fill=blue!20, left=of C, xshift=-4ex, label={left:$\sizedcircled{.15ex}{0}$}] (A) {$A$};
		\node[node, thick, color=red, below=of C, xshift=-8.5ex, yshift=5ex, label={below:$\sizedcircled{.15ex}{4}$}] (D) {$D$}; 	
		\node[node, thick, color=blue, fill=blue!20, right=of E, label={$\sizedcircled{.15ex}{4}$}] (F) {$F$};
		\node[node, thick, color=red, right=of F, label={right:$\sizedcircled{.15ex}{0}$}] (G) {$G$};	
		\draw[] (E) to [bend left=0] node[above] {$+1$} (C);
		\draw[color=red, thick] (E) to [bend left=0] node[above] {$+1$} (F);
		\draw[dotted] (E) to [bend left=22] node[above left, xshift=-4ex] {$+1$} (A.south east);
		\draw[dotted] (E) to [bend left=50] node[above left, xshift=-4ex] {$+1$} (G.north);
		\draw[] (A) to [bend left=40] node[left] {$+4$} (B);
		\draw[color=red, thick] (B) to [bend left=40] node[xshift=2ex, yshift=1ex] {$+4$} (C);
		\draw[] (C) to [bend left=40] node[xshift=2ex, yshift=0ex] {$-4$} (D);
		\draw[color=red, thick] (D) to [bend left=40] node[xshift=-2ex, yshift=-1ex] {$-4$} (A);
		\draw[] (F) to [bend left=40] node[above, xshift=-1ex, yshift=2ex] {$-4$} (G);
		\draw[color=red, thick] (G) to [bend left=40] node[below] {$+4$} (F);
\end{tikzpicture}
}
\subfloat[The extremal-SEPM $f_2$ of $\Gamma_{\text{ex}}^{w+1}$.]{\label{fig:ex1_reweighted_f2SEPM}
\begin{tikzpicture}[scale=.57, arrows={-triangle 45}, node distance=1.5 and 1.5]
 		\node[node, thick, color=red, label={$\sizedcircled{.15ex}{7}$}] (E) {$E$};
		\node[node, thick, color=blue, fill=blue!20, left=of E, label={above right:$\sizedcircled{.15ex}{8}$}] (C) {$C$};
		\node[node, thick, color=red, above=of C, xshift=-8.5ex, yshift=-5ex, label={$\sizedcircled{.15ex}{4}$}] (B) {$B$};
		\node[node, thick, color=blue, fill=blue!20, left=of C, xshift=-4ex, label={left:$\sizedcircled{.15ex}{0}$}] (A) {$A$};
		\node[node, thick, color=red, below=of C, xshift=-8.5ex, yshift=5ex, label={below:$\sizedcircled{.15ex}{4}$}] (D) {$D$}; 	
		\node[node, thick, color=blue, fill=blue!20, right=of E, label={$\sizedcircled{.15ex}{4}$}] (F) {$F$};
		\node[node, thick, color=red, right=of F, label={right:$\sizedcircled{.15ex}{0}$}] (G) {$G$};	
		\draw[color=red, thick] (E) to [bend left=0] node[above] {$+1$} (C);
		\draw[dotted] (E) to [bend left=0] node[above] {$+1$} (F);
		\draw[dotted] (E) to [bend left=22] node[above left, xshift=-4ex] {$+1$} (A.south east);
		\draw[dotted] (E) to [bend left=50] node[above left, xshift=-4ex] {$+1$} (G.north);
		\draw[] (A) to [bend left=40] node[left] {$+4$} (B);
		\draw[color=red, thick] (B) to [bend left=40] node[xshift=2ex, yshift=1ex] {$+4$} (C);
		\draw[] (C) to [bend left=40] node[xshift=2ex, yshift=0ex] {$-4$} (D);
		\draw[color=red, thick] (D) to [bend left=40] node[xshift=-2ex, yshift=-1ex] {$-4$} (A);
		\draw[] (F) to [bend left=40] node[above, xshift=-1ex, yshift=2ex] {$-4$} (G);
		\draw[color=red, thick] (G) to [bend left=40] node[below] {$+4$} (F);
\end{tikzpicture}
}
\end{figure}

\begin{Def}[Extremal-SEPM]
Each element $f\in\mathcal{B}^*_{\Gamma}$ is said to be an \emph{extremal}-SEPM. 
\end{Def}

The next lemma is the converse of Lemma~\ref{lem:pre_main_thm}.
\begin{Lem}\label{lem:conservative_implies_optimal}
Let the \MPG $\Gamma$ be $\nu$-valued, for some $\nu\in \Q$. 
Consider any $\sigma_0\in\Sigma_0^M(\Gamma)$, and assume that $G(\Gamma^{w-\nu}, \sigma_0)$ is conservative. 
Then, $\sigma_0\in\texttt{opt}_\Gamma\Sigma_0^M$. 
\end{Lem}
\begin{proof}
Let $C=(v_1, \ldots, v_{\ell}v_1)$ any cycle in $G(\Gamma, \sigma_0)$. 
Then, the following holds (if $v_{\ell+1}=v_1$):
$\frac{w(C)}{\ell}  = \frac{1}{\ell}\sum_{i=1}^{\ell} w(v_i, v_{i+1})  
		   = \nu + \frac{1}{\ell}\sum_{i=1}^{\ell} \Big( w(v_i, v_{i+1})-\nu \Big) \geq \nu$, 
where $\frac{1}{\ell}\sum_{i=1}^{\ell} \big( w(v_i, v_{i+1})-\nu \big) \geq 0$ 
holds because $G(\Gamma^{w-\nu}, \sigma_0)$ is conservative.
By Proposition~\ref{prop:reachable_cycle}, since $w(C)/\ell \geq \nu$ for 
	every cycle $C$ in $G^{\Gamma}_{\sigma_0}$, then $\sigma_0\in\texttt{opt}_\Gamma\Sigma_0^M$.
\end{proof}

\begin{figure}[!h]
\center
\begin{tikzpicture}[scale=.65, node distance=1 and 1]
 		\node[node, thick, color=red, label={$\sizedcircled{.75}{7}$}] (E2) {$E$};
\node[node, thick, color=blue, fill=blue!20,  left=of E2, label={left, xshift=.5ex, yshift=1ex:$\sizedcircled{.6}{8}$}] (C2) {$C$};
		\node[node, thick, color=red, above=of C2, xshift=-4.75ex, yshift=-4ex, label={$\sizedcircled{.6}{4}$}] (B2) {$B$};
		\node[node, thick, color=blue, fill=blue!20, left=of C2, label={left:$\sizedcircled{.6}{0}$}] (A2) {$A$};
		\node[node, thick, color=red, below=of C2, xshift=-4.75ex, yshift=4ex, label={below:$\sizedcircled{.6}{4}$}] (D2) {$D$}; 	
		\node[node, thick, color=blue, fill=blue!20, right=of E2, label={$\sizedcircled{.6}{4}$}] (F2) {$F$};
		\node[node, thick, color=red, right=of F2, label={right:$\sizedcircled{.6}{0}$}] (G2) {$G$};	
		\draw[->, color=red, thick] (E2) to [bend left=0] node[above, yshift=-.5ex] {$+1$} (C2); 
		\draw[->, dotted] (E2) to [bend left=0] node[above, yshift=-.5ex] {$+1$} (F2);
		\draw[->, dotted] (E2) to [bend left=22] node[above left, xshift=-2ex, yshift=-1.25ex] {$+1$} (A2.south east);
		\draw[->, dotted] (E2) to [bend left=60] node[above left, xshift=0ex, yshift=-1ex] {$+1$} (G2.north);
		\draw[->] (A2) to [bend left=40] node[left] {$+4$} (B2);
		\draw[->,color=red, thick] (B2) to [bend left=40] node[xshift=1.5ex, yshift=1ex] {$+4$} (C2);
		\draw[->] (C2) to [bend left=40] node[xshift=1ex, yshift=-1ex] {$-4$} (D2);
		\draw[->,color=red, thick] (D2) to [bend left=40] node[xshift=-2ex, yshift=-1ex] {$-4$} (A2);
		\draw[->] (F2) to [bend left=40] node[above, yshift=-.5ex] {$-4$} (G2);
		\draw[->,color=red, thick] (G2) to [bend left=40, yshift=1.5ex] node[below] {$+4$} (F2);
		\node[node, thick, right=of G2, xshift=5ex] (fA) {$f_2$};
		\node[node, thick, below=of fA, yshift=-20ex] (fB) {$f_1$};
 		\node[node, thick, left=of fB, xshift=-25ex, color=red, label={$\sizedcircled{.75}{3}$}] (E1) {$E$};
\node[node, thick, color=blue, fill=blue!20, left=of E1, label={left, xshift=.5ex, yshift=1ex:$\sizedcircled{.6}{8}$}] (C1) {$C$};
		\node[node, thick, color=red, above=of C1, xshift=-4.75ex, yshift=-4ex, label={$\sizedcircled{.6}{4}$}] (B1) {$B$};
		\node[node, thick, color=blue, fill=blue!20, left=of C1, label={left:$\sizedcircled{.6}{0}$}] (A1) {$A$};
		\node[node, thick, color=red, below=of C1, xshift=-4.75ex, yshift=4ex, label={below:$\sizedcircled{.6}{4}$}] (D1) {$D$}; 	
		\node[node, thick, color=blue, fill=blue!20, right=of E1, label={$\sizedcircled{.6}{4}$}] (F1) {$F$};
		\node[node, thick, color=red, right=of F1, label={right:$\sizedcircled{.6}{0}$}] (G1) {$G$};	
		\draw[->] (E1) to [bend left=0] node[above, yshift=-.5ex] {$+1$} (C1); 
		\draw[->, color=red, thick] (E1) to [bend left=0] node[above, yshift=-.5ex] {$+1$} (F1);
		\draw[->, dotted] (E1) to [bend left=22] node[above left, xshift=-2ex, yshift=-1.25ex] {$+1$} (A1.south east);
		\draw[->, dotted] (E1) to [bend left=60] node[above left, xshift=0ex, yshift=-1ex] {$+1$} (G1.north);
		\draw[->] (A1) to [bend left=40] node[left] {$+4$} (B1);
		\draw[->,color=red, thick] (B1) to [bend left=40] node[xshift=1.5ex, yshift=1ex] {$+4$} (C1);
		\draw[->] (C1) to [bend left=40] node[xshift=1ex, yshift=-1ex] {$-4$} (D1);
		\draw[->,color=red, thick] (D1) to [bend left=40] node[xshift=-2ex, yshift=-1ex] {$-4$} (A1);
		\draw[->] (F1) to [bend left=40] node[above, yshift=-.5ex] {$-4$} (G1);
		\draw[->,color=red, thick] (G1) to [bend left=40, yshift=1.5ex] node[below] {$+4$} (F1);
		\node[node, thick, below=of fB, yshift=-20ex] (fC) {$f^*$};
 		\node[node, thick, left=of fC, xshift=-25ex, color=red, label={$\sizedcircled{.75}{0}$}] (E0) {$E$};
\node[node, thick, color=blue, fill=blue!20, left=of E0, label={left, xshift=.5ex, yshift=1ex:$\sizedcircled{.6}{8}$}] (C0) {$C$};
		\node[node, thick, color=red, above=of C0, xshift=-4.75ex, yshift=-4ex, label={$\sizedcircled{.6}{4}$}] (B0) {$B$};
		\node[node, thick, color=blue, fill=blue!20, left=of C0, label={left:$\sizedcircled{.6}{0}$}] (A0) {$A$};
		\node[node, thick, color=red, below=of C0, xshift=-4.75ex, yshift=4ex, label={below:$\sizedcircled{.6}{4}$}] (D0) {$D$}; 	
		\node[node, thick, color=blue, fill=blue!20, right=of E0, label={$\sizedcircled{.6}{4}$}] (F0) {$F$};
		\node[node, thick, color=red, right=of F0, label={right:$\sizedcircled{.6}{0}$}] (G0) {$G$};	
	\draw[->] (E0) to [bend left=0] node[above, yshift=-.5ex] {$+1$} (C0); 
	\draw[->] (E0) to [bend left=0] node[above, yshift=-.5ex] {$+1$} (F0);
	\draw[->, color=red, thick] (E0) to [bend left=22] node[above left, xshift=-2ex, yshift=-1.25ex] {$+1$} (A0.south east);
	\draw[->, color=red, thick] (E0) to [bend left=60] node[above left, xshift=0ex, yshift=-1ex] {$+1$} (G0.north);
	\draw[->] (A0) to [bend left=40] node[left] {$+4$} (B0);
	\draw[->,color=red, thick] (B0) to [bend left=40] node[xshift=1.5ex, yshift=1ex] {$+4$} (C0);
	\draw[->] (C0) to [bend left=40] node[xshift=1ex, yshift=-1ex] {$-4$} (D0);
	\draw[->,color=red, thick] (D0) to [bend left=40] node[xshift=-2ex, yshift=-1ex] {$-4$} (A0);
	\draw[->] (F0) to [bend left=40] node[above, yshift=-.5ex] {$-4$} (G0);
	\draw[->,color=red, thick] (G0) to [bend left=40, yshift=1.5ex] node[below] {$+4$} (F0);
		\draw[thick] (fA) to [] node[] {} (fB);
		\draw[thick] (fB) to [] node[] {} (fC);

		\node[scale=0.75, right=of fC, yshift=0ex, rectangle split, rectangle split parts=6, draw, 
			minimum width=4cm, font=\small, rectangle split part align={center}] (t1)
  		{            
		  {$\Delta^M_0(f^*, \Gamma^{w-\nu})=\{\sigma^{(1)}_0, \sigma^{(2)}_0\}$}
		     \nodepart{two}
       		             $\sigma^{(1)}_0(B)=C$ \hspace*{4ex} $\sigma^{(2)}_0(B)=C$ 
		     \nodepart{three}
       	  	             $\sigma^{(1)}_0(D)=A$ \hspace*{4ex} $\sigma^{(2)}_0(D)=A$
	     	     \nodepart{four}
                             \textcolor{red}{$\sigma^{(1)}_0(E)=A$} \hspace*{4ex} \textcolor{red}{$\sigma^{(2)}_0(E)=G$} 
	 	     \nodepart{five}
                  	     $\sigma^{(1)}_0(G)=F$ \hspace*{4ex} $\sigma^{(2)}_0(G)=F$
 	     	     \nodepart{six}
		};
	 	 \draw (t1.text split) -- (t1.two split);
		 \draw (t1.two split) -- (t1.three split);
		 \draw (t1.three split) -- (t1.four split);
		 \draw (t1.four split) -- (t1.five split);
		 \draw (t1.five split) -- (t1.six split);
		 \draw (t1.six split) -- (t1.seven split);
 		 \draw (t1.seven split) -- (t1.eight split);

	\node[right=of fB, scale=0.75, yshift=0ex, rectangle split, rectangle split parts=6, draw, 
			minimum width=4cm, font=\small, rectangle split part align={center}] (t2)
  		  {            
		{$\Delta^M_0(f_1, \Gamma^{w-\nu})=\{\sigma^{(3)}_0\}$}
	     \nodepart{two}
     	              $\sigma^{(3)}_0(B)=C$ 
	     \nodepart{three}
                      $\sigma^{(3)}_0(D)=A$ 
	     \nodepart{four}
                      \textcolor{red}{$\sigma^{(3)}_0(E)=F$}  
	     \nodepart{five}
                      $\sigma^{(3)}_0(G)=F$ 
 	     \nodepart{six}
		};

		\node[right=of fA, yshift=0ex, scale=0.75, rectangle split, rectangle split parts=6, draw, 
			minimum width=4cm, font=\small, rectangle split part align={center}] (t2)
 		  {            
		{$\Delta^M_0(f_1, \Gamma^{w-\nu})=\{\sigma^{(4)}_0\}$}
	     \nodepart{two}
     	              $\sigma^{(4)}_0(B)=C$ 
	     \nodepart{three}
                      $\sigma^{(4)}_0(D)=A$ 
	     \nodepart{four}
                      \textcolor{red}{$\sigma^{(4)}_0(E)=C$}  
	     \nodepart{five}
                      $\sigma^{(4)}_0(G)=F$ 
 	     \nodepart{six}
		};
\end{tikzpicture}
\caption{The decomposition of $\texttt{opt}_\Gamma\Sigma^M_0$, for the \MPG $\Gamma_{\text{ex}}$,
which corresponds to the Energy-Lattice $\mathcal{B}^*_{\Gamma_{\text{ex}}}=\{f^*, f_1, f_2\}$
(computed in Example~\ref{exp1}). Here, $f^*\leq f_1\leq f_2$. This also brings a lattice 
$\mathcal{D}^*_{\Gamma_{\text{ex}}}$ of 3 sub-games of $\Gamma_{\text{ex}}$. }\label{fig:ex1_ordering}
\end{figure}
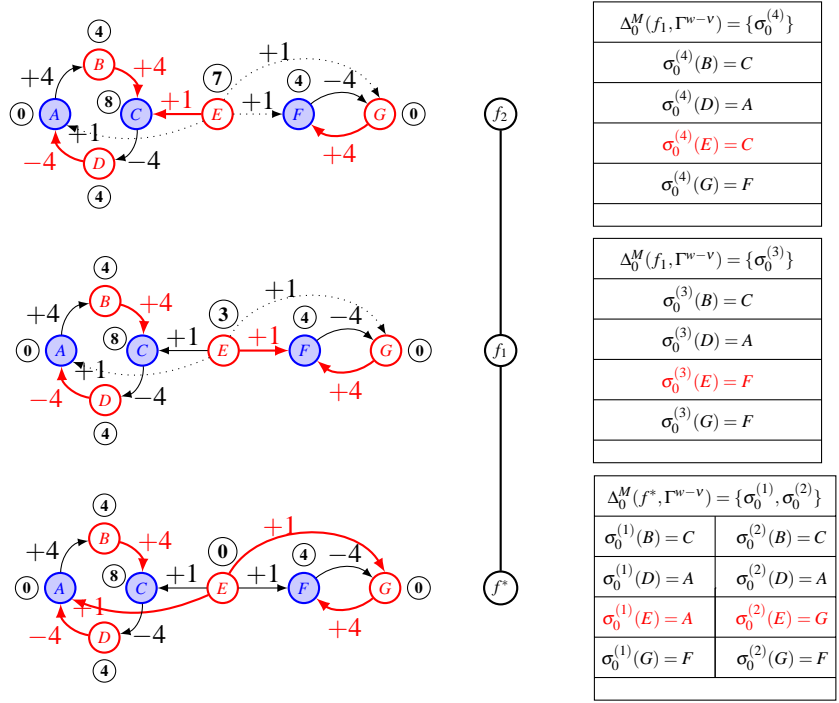

The following proposition asserts some properties of the Extremal-SEPM{s}.
\begin{Prop}\label{prop:extremal}
Let the \MPG $\Gamma$ be $\nu$-valued, for some $\nu\in \Q$. 
Let $\mathcal{B}^*_{\Gamma}$ be the Energy-Lattice of $\texttt{opt}_\Gamma\Sigma^M_0$. 
Moreover, let $f:V\rightarrow \C_{\Gamma}$ be a SEPM for the 
reweighted \EG $\Gamma^{w-\nu}$. Then, the following three properties are equivalent: 
\begin{enumerate}
\item $f\in \mathcal{B}^*_{\Gamma}$; 
\item There exists $\sigma_0\in\texttt{opt}_\Gamma\Sigma^M_0$ such that 
$\pi^*_{G(\Gamma^{w-\nu}, \sigma_0)}(v)=f(v)$ for every $v\in V$.
\item $V_f=\mathcal{W}_0(\Gamma^{w-\nu})=V$ and $\Delta^M_0(f, \Gamma^{w-\nu})\neq\emptyset$;
\end{enumerate}
\end{Prop}
\begin{itemize}
\item[($1\iff 2$)] Indeed, $\mathcal{B}^*_{\Gamma} = 
	\{ \pi^*_{G(\Gamma^{w-\nu}, \sigma_0)}\mid {\sigma_0\in \texttt{opt}_{\Gamma}\Sigma_0^M}\} $. 

\item[($1\Rightarrow 3$)] Assume $f\in \mathcal{B}^*_{\Gamma}$. 
Since ($1\iff 2$), there exist $\sigma_0\in\texttt{opt}_\Gamma\Sigma_0^M$ 
such that $\pi^*_{G(\Gamma^{w-\nu}, \sigma_0)}=f$. Thus, $\sigma_0\in\Delta^M_0(f, \Gamma^{w-\nu})$, 
	so that $\Delta^M_0(f, \Gamma^{w-\nu})\neq\emptyset$.
We claim $V_f=\mathcal{W}_0(\Gamma^{w-\nu})=V$. Since $\forall (v\in V)\, \val{\Gamma}{v}=\nu$, 
then $\mathcal{W}_0(\Gamma^{w-\nu})=V$ by Proposition~\ref{prop:relation_MPG_EG}.
Next, $G(\Gamma^{w-\nu}, \sigma_0)$ is conservative by Lemma~\ref{lem:pre_main_thm}. 
Since $G(\Gamma^{w-\nu}, \sigma_0)$ is conservative and $f=\pi^*_{G(\Gamma^{w-\nu}, \sigma_0)}$, 
then $V_f=V$. Therefore, $V_f=\mathcal{W}_0(\Gamma^{w-\nu})=V$. 

\item[($1\Leftarrow 3$)]
Since $\Delta^M_0(f, \Gamma^{w-\nu})\neq\emptyset$, pick some $\sigma_0\in\Delta^M_0(f, \Gamma^{w-\nu})$; 
so, $f=\pi^*_{G(\Gamma^{w-\nu}, \sigma_0)}$. Since $V_f=V$ and $f=\pi^*_{G(\Gamma^{w-\nu}, \sigma_0)}$, 
then $G(\Gamma^{w-\nu}, \sigma_0)$ is conservative. 
Since $G(\Gamma^{w-\nu}, \sigma_0)$ is conservative, 
then $\sigma_0\in\texttt{opt}_\Gamma\Sigma_0^M$ by Lemma~\ref{lem:conservative_implies_optimal}.
Since $f=\pi^*_{G^*}$ and $\sigma_0\in\texttt{opt}_\Gamma\Sigma_0^M$, 
then $f\in\mathcal{B}^*_{\Gamma}$ (as $2\Rightarrow 1$).
\end{itemize}

\section{A Recursive Enumeration of $\mathcal{B}^*_\Gamma$ 
	and $\texttt{opt}_\Gamma\big(\Sigma_{0}^M\big)$}\label{sect:recursive_enumeration}

An enumeration algorithm for a set $S$ provides an exhaustive listing of all the elements of $S$ (without repetitions). 
As mentioned in Section~\ref{sect:energy}, by Theorem~\ref{thm:ergodic_partition}, 
no loss of generality occurs if we assume $\Gamma$ to be $\nu$-valued for some $\nu\in\Q$.
One run of Algorithm~\ref{algo:solve_mpg} allows one to partition an \MPG $\Gamma$, 
into several domains $\Gamma_{i}\triangleq \Gamma_{|_{C_i}}$ each one being $\nu_{i}$-valued for $\nu_{i}\in S_\Gamma$; 
in $O(|V|^2|E|W)$ time and linear space. Still, by Proposition~\ref{prop:counter_example}, 
Theorem~\ref{Thm:pos_opt_strategy} is not sufficient for enumerating the whole 
$\texttt{opt}_\Gamma(\Sigma^M_0)$ by means of Algorithm~\ref{algo:solve_mpg}; 
it is enough only for $\Delta^M_0(f^*_\nu,\Gamma^{w-\nu})$ where $f^*_\nu$ is the least-SEPM of $\Gamma^{w-\nu}$, 
which is just the ``join" component of $\texttt{opt}_\Gamma(\Sigma^M_0)$.
However, we now have a more detailed description of $\texttt{opt}_\Gamma \Sigma^M_0$ in terms $\mathcal{B}^*_\Gamma$, thanks to Theorem~\ref{thm:main_energystructure}. 

This section offers a recursive enumeration of all the extremal-SEPM{s}, \ie $\mathcal{B}^*_\Gamma$, 
and for computing the corresponding partitioning of $\texttt{opt}_\Gamma\big(\Sigma_{0}^M\big)$. 
In order to avoid duplicate elements in the enumeration, the algorithm needs to store a lattice $\mathcal{B}^*_{\Gamma}$ of sub-games of $\Gamma$, which is related to $\mathcal{X}^*_\Gamma$. \emph{($T_{\Gamma}$)}. 
We shall assume to dispose of a data-structure $T_{\Gamma}$ supporting the following operations, 
given a sub-arena $\Gamma'$ of $\Gamma$: $\texttt{insert}(\Gamma', T_{\Gamma})$ stores $\Gamma'$ into $T_{\Gamma}$; 
$\texttt{contains}(\Gamma', T_{\Gamma})$ returns $\texttt{T}$ \textit{iff} $\Gamma'$ is in $T_{\Gamma}$ and $\texttt{F}$ otherwise. 
A simple implementation of $T_{\Gamma}$ goes by indexing $N^{\text{out}}_{\Gamma'}(v)$ for each $v\in V$.
This runs in $O(|V|^2)$ time, consuming $O(|E|)$ space per stored element. 
The same approach can be used to store and retrieve SEPM{s} in $O(|V|^2)$ time and $O(|V|)$ space.

The listing procedure is named $\texttt{enum}()$, it takes in input a $\nu$-valued \MPG $\Gamma$; going as follows.
\begin{enumerate}
\item Compute the least-SEPM $f^*$ of $\Gamma$, and \texttt{print} $\Gamma$ to output. 
	Theorem~\ref{Thm:pos_opt_strategy} can be employed at this 
	stage for enumerating $\Delta^M_0(f^*, \Gamma^{w-\nu})$:
	indeed, these are all and only those positional strategies 
	lying in the \emph{Cartesian} product of all the arcs $(u,v)\in E$ 
	that are \emph{compatible} with $f^*$ in $\Gamma^{w-\nu}$ (because $f^*$ is the least-SEPM of $\Gamma$).
\item Let $S_t\leftarrow\emptyset$ be an empty stack.
\item For each $\hat{u}\in V_0$, do the following:
	\begin{itemize}
		\item Compute $E_{\hat{u}}\leftarrow \{(\hat{u},v)\in E\mid f^*(\hat{u}) 
			\prec f^*(v)\ominus(w(\hat{u},v)-\nu)\}$; If $E_{\hat{u}}\neq \emptyset$, then: 
		\begin{itemize} 
		\item Let $E'\leftarrow E_{\hat{u}} \cup \{ (u,v)\in E\mid u\neq \hat{u}\}$ 
			and $\Gamma'\leftarrow (V,E',w,\langle V_0, V_1\rangle)$.
		\item If $\texttt{contains}(\Gamma', T_{\Gamma})=\texttt{F}$, do the following: 
		\begin{itemize} 
			\item Compute the least-SEPM ${f'}^*$ of ${\Gamma'}^{w-\nu}$;			 
			\item If $V_{{f'}^*}=V$: 

				-- Push $\hat{u}$ on top of $S_t$ and $\texttt{insert}(\Gamma', T_{\Gamma})$.

				-- If $\texttt{contains}({f'}^*, T_{\Gamma})=\texttt{F}$, 
					then $\texttt{insert}({f'}^*, T_{\Gamma})$ and \texttt{print} ${f'}^*$.
		\end{itemize} 
		\end{itemize}
	\end{itemize}
\item While $S_t\neq\emptyset$:
	\begin{itemize}
	\item \texttt{pop} $\hat{u}$ from $S_t$; 
		Let $E_{\hat{u}}\leftarrow \{(\hat{u},v)\in E\mid f^*(\hat{u}) \prec f^*(v)\ominus(w(\hat{u},v)-\nu)\}$, 
		and $E'\leftarrow E_{\hat{u}} \cup \{ (u,v)\in E\mid u\neq \hat{u}\}$, 
		and $\Gamma'\leftarrow (V,E',w,\langle V_0, V_1\rangle)$; 
	\item Make a \emph{Recursive} call to $\texttt{enum}()$ on input $\Gamma'$.
	\end{itemize}
\end{enumerate}
Down the recursion tree, when computing least-SEPMs, children Value-Iterations can amortize by 
starting from the energy-levels of the parent. The lattice of sub-games $\mathcal{B}^*_{\Gamma}$ 
comprises all and only those sub-games $\Gamma'\subseteq \Gamma$ that are eventually inserted 
into $T_{\Gamma}$ at Step~(3) of $\texttt{enum}()$; these are called the \emph{basic} sub-games of $\Gamma$. 
The correctness of $\texttt{enum}()$ follows by Theorem~\ref{thm:main_energystructure} 
and Theorem~\ref{Thm:pos_opt_strategy}. In summary, we obtain Theorem~\ref{thm:listing_algo}. 
\begin{Thm}\label{thm:listing_algo} There is a recursive algorithm for enumerating (w/o repetitions) 
all the elements of $\mathcal{B}^*_{\Gamma}$, on any input \MPG $\Gamma$, with time \emph{delay} $O(|V|^3 |E|\, W)$. 
For this, the algorithm employs $O(|E||V|)+\Theta\big(|E||\mathcal{B}^*_{\Gamma}|)\big)$ working space. 
The algorithm enumerates $\mathcal{X}^*_\Gamma$ (w/o repetitions) in 
$O\big(|V|^3|E|W|\mathcal{B}^*_{\Gamma}|\big)$ \emph{total} time, and 
$O(|V||E|)+\Theta\big(|E||\mathcal{B}^*_{\Gamma}|\big)$ space.
\end{Thm} (Say $O(f(n))$ time \emph{delay} when the time spent between any two consecutives is $O(f(n))$.)

To conclude we observe that $\mathcal{B}^*_{\Gamma}$ and $\mathcal{X}^*_\Gamma$ are not isomorphic as lattices, 
not even as sets (the cardinality of $\mathcal{B}^*_{\Gamma}$ can be greater that that of $\mathcal{X}^*_\Gamma$).
Indeed, there is a surjective antitone mapping $\varphi_\Gamma$ from $\mathcal{B}^*_{\Gamma}$ onto $\mathcal{X}^*_\Gamma$, 
(\ie $\varphi_\Gamma$ sends $\Gamma'\in \mathcal{B}^*_{\Gamma}$ to its least-SEPM $f^*_{\Gamma'}\in \mathcal{X}^*_\Gamma$); 
still, we can construct examples of \MPG{s} such that 
$|\mathcal{B}^*_{\Gamma}| > |\mathcal{X}^*_\Gamma|$, \ie $\varphi_\Gamma$ is not into (in case of degeneracy).

\section{Conclusion}\label{sect:conclusions}
We offered a faster $O(|E|\log |V|)+\Theta(\sum_{v\in V} \texttt{deg}_{\Gamma}(v)\cdot\ell(v))=O(|V|^2|E|W)$ 
time energy algorithm for the Value Problem and Optimal Strategy Synthesis in \MPG{s}. 
The result was achieved by introducing a novel scheme based on so called \emph{Energy-Increasing} and \emph{Unitary-Advance} Jumps. 

In addition, we observed a unique complete decomposition of $\texttt{opt}_{\Gamma}\Sigma^M_0$ 
in terms of extremal-SEPM{s} in reweighted \EG{s}, offering a pseudo-polynomial total-time recursive algorithm for enumerating 
(w/o repetitions) all the elements of $\mathcal{X}^*_\Gamma$, \ie all the extremal-SEPMs, 
and for computing the components of the corresponding partitioning $\mathcal{B}^*_{\Gamma}$ of $\texttt{opt}_{\Gamma}\Sigma^M_0$.

It would be interesting to study further properties enjoyed by $\mathcal{B}^*_{\Gamma}$ and $\mathcal{X}^*_\Gamma$; 
also, we ask for more efficient algorithms for enumerating $\mathcal{X}^*_\Gamma$, 
\eg we ask for pseudo-polynomial time delay and polynomial space enumerations. 
We also ask whether the least-SEPM of reweighted \EG{s} of the kind $\Gamma^{w-q}$, 
for $q\in S_\Gamma$, can be computed in $O(|V||E|W)$ time, instead of $O(|V|^2|E|W)$: together with Algorithm~\ref{algo:solve_mpg}, 
this could lead to an improved time complexity upper bound for \MPG{s} (\ie matching the time spent for solving \EG{s}). 
To conclude, it would be very interesting to adapt Algorithm~\ref{algo:solve_mpg} to 
work with the strategy-improvement framework, instead of the value-iteration, 
as it seems to exhibit a faster converge in practice. 

Many questions remain open on this way.

\paragraph*{Acknowledgments} 
This work was supported by \emph{Department of Computer Science, University of Verona, Verona, Italy}, 
under PhD grant “Computational Mathematics and Biology”, on a co-tutelle agreement with 
\emph{LIGM, Universit\'e Paris-Est in Marne-la-Vall\'ee, Paris, France}.

\section*{References}
\bibliographystyle{elsarticle-num-names}
\bibliography{biblio}


\end{document}